\documentclass[12pt]{ucsddissertation}

\usepackage{mathptmx}		
\usepackage[NoDate]{currvita}
\usepackage{array}
\usepackage{tabularx}
\usepackage{booktabs}
\usepackage{ragged2e}
\usepackage{microtype}
\usepackage[breaklinks=true,pdfborder={0 0 0}]{hyperref}
\usepackage{graphicx}
\usepackage[english]{babel}
\usepackage{blindtext}
\usepackage{amsmath}
\usepackage{amssymb}
\usepackage{amsthm}
\usepackage{array}
\usepackage[caption=false]{subfig}
\usepackage{pgfplots}
\usepackage{algorithm}
\usepackage[noend]{algpseudocode}
\usepackage{color}
\usepackage{cite}
\usepackage{float}
\usepackage{tikz-cd}
\usepackage[utf8]{inputenc}

\renewcommand{\natural}{{\mathbb{N}}}

\newcommand{\real}{\ensuremath{\mathbb{R}}}

\newcommand{\complex}{\ensuremath{\mathbb{C}}}

\newcommand{\until}[1]{\{1,\dots, #1\}}
\newcommand{\subscr}[2]{#1_{\textup{#2}}}
\newcommand{\supscr}[2]{#1^{\textup{#2}}}
\newcommand{\setdef}[2]{\{#1 \; | \; #2\}}

\newcommand{\diag}[1]{\operatorname{diag}(#1)}

\newcommand{\spn}{\operatorname{span}}
\newcommand{\image}{\operatorname{image}}

\renewcommand{\epsilon}{\varepsilon}

\newcommand{\argmin}{\ensuremath{\operatorname{argmin}}}
\newcommand{\argmax}{\ensuremath{\operatorname{argmax}}}

\newcommand{\Pc}{{\mathcal{P}}}

\usepackage{psfrag}

\usepackage[mathscr]{euscript}
\usepackage{enumerate}
\usepackage{makecell}

\AtBeginDocument{%
	\settowidth\cvlabelwidth{\cvlabelfont 0000--0000}%
}

\increasemargins{.1in}
\def\getfirst#1.#2\relax{#1}

\newtheorem{thm}{Theorem}
\newtheorem{lem}{Lemma}
\newtheorem{prop}{Proposition}
\newtheorem{cor}{Corollary}
\newtheorem{defn}{Definition}
\newtheorem{rem}{Remark}
\newtheorem{assump}{Assumption}
\newtheorem{example}{Example}
\newtheorem{cond}{Condition}

\DeclareMathOperator{\W}{\mathcal{W}}
\DeclareMathOperator{\N}{\mathcal{N}}
\DeclareMathOperator{\R}{\mathbb{R}}

\DeclareMathOperator{\E}{\mathcal{E}}
\DeclareMathOperator{\X}{\mathcal{X}}
\DeclareMathOperator{\G}{\mathcal{G}}
\DeclareMathOperator{\Ls}{\mathscr{L}}

\DeclareMathOperator{\U}{\mathcal{U}}
\DeclareMathOperator{\Ex}{\mathbb{E}}
\DeclareMathOperator{\D}{\mathcal{D}}
\newcommand{\Pref}{\ensuremath{\subscr{P}{ref}}}
\DeclareMathOperator{\x}{\mathbf{x}}
\newcommand{\mean}{\operatorname{mean}}

\DeclareMathOperator{\Pp}{\mathbb{P}}
\DeclareMathOperator{\B}{\mathcal{B}}
\DeclareMathOperator{\Y}{\mathcal{Y}}
\DeclareMathOperator{\SSS}{\mathcal{S}}

\DeclareMathOperator{\I}{\mathbb{I}}
\DeclareMathOperator{\Px}{\mathsf{P}}

\newcommand{\nll}{\operatorname{null}}
\newcommand{\aff}{\operatorname{aff}}
\newcommand{\relint}{\operatorname{relint}}

\newcommand{\pu}{\ensuremath{\underline{p}}}
\newcommand{\po}{\ensuremath{\overline{p}}}
\DeclareMathOperator{\Lagr}{\mathcal{L}}

\newcommand{\distnewton}{\textsc{distributed approx-Newton}\xspace}
\newcommand{\dana}{\textsc{DANA}\xspace}
\newcommand{\distnewtondisc}{\textsc{discrete distributed approx-Newton}\xspace}
\newcommand{\danad}{\textsc{DANA-D}\xspace}
\newcommand{\distnewtoncont}{\textsc{continuous distributed approx-Newton}\xspace}
\newcommand{\danac}{\textsc{DANA-C}\xspace}
\newcommand{\distgrad}{\textsc{distributed gradient descent}\xspace}
\newcommand{\dgd}{\textsc{DGD}\xspace}

\newcommand{\nulo}{\operatorname{null}}
\newcommand{\dom}{\operatorname{dom}}

\newcommand{\nodes}{\N}
\DeclareMathOperator{\ones}{\mathbf{1}}
\DeclareMathOperator{\zeros}{\mathbf{0}}

\newcommand{\setdefB}[2]{\Big\{#1 \; | \; #2\Big\}}

\newcommand{\longthmtitle}[1]{\mbox{}{\bf
      \textit{(#1).}}}
  
\newcommand{\binpromc}{\textsc{Newton-like Neural Network}\xspace} 
\newcommand{\binpromd}{\textsc{Newton-like Neural Network}\xspace} 

\newcommand{\binpac}{\textsc{NNN-c}\xspace} 
\newcommand{\binpad}{\textsc{NNN-d}\xspace} 

\newcommand\oprocendsymbol{\hbox{$\bullet$}}
\newcommand\oprocend{\relax\ifmmode\else\unskip\hfill\fi\oprocendsymbol}

\title{Distributed Newton-like Algorithms and Learning for Optimized Power Dispatch}
\author{Tor Anderson}
\degree{Mechanical Engineering}{Doctor of Philosophy}
\chair{Professor Sonia Mart{\'i}nez}
\committee{Professor Jorge Cort{\'e}s}
\committee{Professor Miroslav Krsti{\'c}}
\committee{Professor Jiawang Nie}
\committee{Professor Behrouz Touri}
\degreeyear{2020}

\begin{document}

\frontmatter
\maketitle
\makecopyright
\makesignature

\begin{dedication}
\null\vfil
\begin{center}
To my loving parents, Brian and Karen, my sister, Louise, \\ and to the memory of my dear friend, Marcus.\\\vspace{12pt}
\end{center}
\vfil\null
\end{dedication}

\begin{epigraph}
\null\vfil
\begin{center}
\textit{I think that it is a relatively good approximation to truth -- which is much too complicated to allow anything but approximations -- that mathematical ideas originate in empirics, although the genealogy is sometimes long and obscure. But, once they are so conceived, the subject begins to live a peculiar life of its own and is better compared to a creative one, governed by almost entirely aesthetical motivations, than to anything else and, in particular, to an empirical science. (\dots) In any event, whenever this stage is reached, the only remedy seems to me to be rejuvenating return to the source: the reinjection of more or less directly empirical ideas. I am convinced that this was a necessary condition to conserve the freshness and the vitality of the subject and that this will remain equally true in the future.}
\vspace{12pt} 
\begin{flushright}
-- John von Neumann
\end{flushright}
\end{center}
\vfil\null
\end{epigraph}

\tableofcontents
\listoffigures

\begin{acknowledgements}
I extend the maximum amount of thanks and appreciation to my advisor, Sonia Mart\'{i}nez. On a personal level, Sonia is kind, funny, and understanding, and professionally, she is insightful, diligent, and above all else, patient. I credit her for seeding what would become the broad theme of this thesis in the early stages of my Ph.D. work, and I expect to reap the benefits of the practices, habits, and modes of thinking that she instilled in me for years to come.

Next, my thanks goes out to the other members of my committee: Professors Jorge Cort\'{e}s, Miroslav Krsti\'{c}, Jiawang Nie, and Behrouz Touri; it is a privilege to benefit from their volunteered time and feedback during the development and submission of this work. Further, I would be remiss not to mention those who previously supervised me and molded my professional development in some way. In reverse chronological order, I extend thanks to: Jorge Cort\'{e}s (again) and Jan Kleissl for their supervision and insight during the work that became the content of Chapter~\ref{chap:active_load} of this thesis; Sonja Glavaski, in part for her vision of the NODES project which directly inspired the work of Chapter~\ref{chap:active_load}, but also for supervising and mentoring me at ARPA-e during the summer between my undergraduate and graduate study; James (Mike) Sigler, for being my first boss in engineering and exposing me to how the sauce is made; and, finally, Elias Lemon, my first boss, who constantly trusted and challenged me in my first ``real” job to take on more than I thought I was capable of at a large martial arts studio.

I would now like to expand on the people mentioned in the dedication. My father, Brian, has been a steady source of inspiration and motivation throughout my life to pursue ideas that are interesting and challenging. Without his influence, I likely would not have ever considered attempting graduate study. My mother, Karen, has always been a grounding influence in my life. On the pragmatic side, I credit her for teaching me the writing habits that made this thesis possible, and on the counseling side, she can reliably talk me back into a relaxed and clear state-of-mind when pressure and obligations otherwise seem overwhelming. My sister, Louise, is a frequent reminder to me that problems I face are fairly universal, and that one's accomplishments are always relative to the lens through which one views them. The clich\'{e} ``wise beyond one's years" applies to people like her. Finally, my late friend Marcus was extremely formative; his friendship and shared affinity for math, science, and engineering in high school and undergraduate were irreplaceable. His continued encouragement throughout my graduate study was essential to completing this thesis.

Additional thanks go to current and former members of our group, including (but not limited to!) Ashish, Aaron, Erfan, Eduardo, Yifu, Chin-Yao, Dimitris, Miguel, Priyank, Pio, Dan, and Aamodh, for their friendship and, of course, frequent technical discussions. My late grandparents, Tor and Phyllis, were passionate educators, and my grandmother, Darlene, and my late grandfather, Ken, were successful entrepreneurs, so they also have directly and indirectly helped me get to this point. Lastly, I give thanks to my wonderful partner, Amie, whose companionship over the last year and a half has helped me get through this ``home-run stretch," and to my adorable dog, Teddy, who can inject a quick shot of joy into my day at literally any moment.

The material in Chapter~\ref{chap:DANA}, in full, is a reprint of \textit{Distributed Approximate Newton Algorithms and Weight Design for Constrained Optimization}, T.~Anderson, C.Y.~Chang and S.~Mart\'{i}nez, Automatica, 109, article 108538, November 2019. A preliminary version of the work appeared in the proceedings of the Conference on Control Technology and Applications (CCTA), Mauna Lani, HI, 2017, pp. 632-637, as \textit{Weight Design of Distributed Approximate Newton Algorithms for Constrained Optimization}, T.~Anderson, C.Y.~Chang and S.~Mart\'{i}nez. The dissertation author was the primary investigator and author of these papers.

Chapter~\ref{chap:DiSCRN}, in full, is being revised and prepared for submission to the Systems \& Control Letters. It may appear as \textit{Distributed Stochastic Nested Optimization via Cubic Regularization}, T.~Anderson and S.~Mart{\'i}nez. The dissertation author was the primary investigator and author of this paper.

The content in Chapter~\ref{chap:NNN}, in full, is provisionally accepted in Automatica. It is expected to appear as \textit{Distributed Resource Allocation with Binary Decisions via Newton-like Neural Network Dynamics}, T.~Anderson and S.~Mart\'{i}nez. The dissertation author was the primary investigator and author of this paper.

The material in Chapter~\ref{chap:top}, in full, is a reprint of \textit{Maximizing Algebraic Connectivity of Constrained Graphs in Adversarial Environments}, T.~Anderson, C.Y.~Chang and S.~Mart\'{i}nez, 2018 European Control Conference (ECC), Limassol, 2018, pp. 125-130. The dissertation author was the primary investigator and author of this paper.

Chapter~\ref{chap:active_load}, in full, is under revision for publication in IEEE Transactions on Smart Grid. It may appear as \textit{Frequency Regulation with Heterogeneous Energy Resources: A Realization using Distributed Control}, T.~Anderson, M.~Muralidharan, P.~Srivastava, H.V.~Haghi, J.~Cort\'{e}s, J.~Kleissl, S.~Mart\'{i}nez and B.~Washom. The dissertation author was one of three primary investigators and authors of this paper.

\end{acknowledgements}

\begin{vita}
\noindent
\begin{cv}{}
\begin{cvlist}{}
  \item[2020] Ph.D., University of California San Diego
  \item[2017] M.S., University of California San Diego
  \item[2015] B.S., University of Minnesota Twin Cities
\end{cvlist}
\end{cv}

\publications
\begin{enumerate}
\item \textit{Distributed Stochastic Nested Optimization via Cubic Regularization }, T.~Anderson and S.~Mart{\'i}nez, In preparation for submission to Systems \& Control Letters.

\item \textit{Frequency Regulation with Heterogeneous Energy Resources: A Realization using Distributed Control}, T.~Anderson, M.~Muralidharan, P.~Srivastava, H.V.~Haghi, J.~Cort\'{e}s, J.~Kleissl, S.~Mart\'{i}nez and B.~Washom, IEEE Transactions on Smart Grid, Under revision.

\item \textit{Distributed Resource Allocation with Binary Decisions via Newton-like Neural Network Dynamics}, T.~Anderson and S.~Mart\'{i}nez, Automatica, Provisionally accepted.

\item \textit{Distributed Approximate Newton Algorithms and Weight Design for Constrained Optimization}, T.~Anderson, C.Y.~Chang and S.~Mart\'{i}nez, Automatica, 109, article 108538, November 2019.

\item \textit{Maximizing Algebraic Connectivity of Constrained Graphs in Adversarial Environments}, T.~Anderson, C.Y.~Chang and S.~Mart\'{i}nez, 2018 European Control Conference (ECC), Limassol, 2018, pp. 125-130.

\item \textit{Weight Design of Distributed Approximate Newton Algorithms for Constrained Optimization}, T.~Anderson, C.Y.~Chang and S.~Mart\'{i}nez, 2017 IEEE Conference on Control Technology and Applications (CCTA), Mauna Lani, HI, 2017, pp. 632-637.

\end{enumerate}
\end{vita}

\begin{dissertationabstract}
This thesis explores a particular class of distributed optimization methods for various separable resource allocation problems, which are of high interest in a wide array of multi-agent settings. A distinctly motivating application for this thesis is real-time power dispatch of distributed energy resources for providing frequency control in a distribution grid or microgrid with high renewable energy penetration. In this application, it is paramount that agent data be shared as sparsely as possible in the interest of conserving user privacy, and it is required that algorithms scale gracefully as the network size increases to the order of thousands or millions of resources and devices. Distributed algorithms are naturally well-poised to address these challenges, in contrast to more traditional centralized algorithms which scale poorly and require global access to information.

The class of distributed optimization methods explored here can be broadly described as \emph{Newton-like} or \emph{second-order}, implying utilization of second-derivative information of the cost functions, in contrast to well-studied \emph{gradient-based} or \emph{first-order} methods. We consider three formulations of separable resource-allocation problems and develop a Newton-like algorithm for each. First, the cost function is given by the sum of local agent costs, supplemented with individual linear box constraints and a global matching-constraint in which the sum of agent states must equal a prescribed constant. Second, we consider a stochastic, nested scenario, in which batches of realizations of problems of the first type must be used to gradually learn the optimal value of a parameter which is coupled with the agent costs. Third, we further constrain the agent states to be binary, and we embed the global matching-constraint as a squared penalty term in the cost. The analysis and simulation studies in the subsequent chapters demonstrate the advantages of our approaches over existing methods; most commonly, we note that convergence rates are substantially improved. We supplement our algorithm development for these three problem formulations with a network design technique, in which we can construct a maximally-connected network by adding some edges to the underlying communication graph, and a real demonstration of distributed algorithms on a large set of heterogeneous devices on the UC San Diego microgrid.

\end{dissertationabstract}

\mainmatter

\begin{dissertationintroduction}
\label{chap:intro}

The natural universe is made up of and governed by distributed interactions. This is evidenced on every relevant scale and setting: particle interactions, signal exchanges between neurons, cooperation between biological organs, tight-knit and mass-scale social interactions, the interacting physics of distant Earth biomes, and gravitational forces between planets, stars, and galaxies. These distributed physical dynamics are \emph{prescribed} by the universe (sometimes indirectly, e.g. via human evolution in the social dynamics case); however, with the advent of modern technology, current and future engineers can benefit from the \emph{imposition} of distributed intelligence and algorithms. It may be the case that this technological imposition is not only inevitable, but that we are already deeply in the midst of it.

Let us take a step back by considering a specific analogy to biological evolution. It is estimated that Earth's first prokaryotic life (cells \emph{without} a nucleus) originated about 3.5--3.8 billion years ago, while the first eukaryotic life (cells \emph{with} a nucleus) is estimated to have evolved around 1.7--2.2 billion years ago.\footnote{Source: Carl Woese, J Peter Gogarten, ``When did eukaryotic cells (cells with nuclei and other internal organelles) first evolve? What do we know about how they evolved from earlier life-forms?" Scientific American, October 21, 1999. \url{https://www.scientificamerican.com/article/when-did-eukaryotic-cells/}} The implication is that the time-scale of progression from eukaryotic cells to modern multi-cellular life, in all its richness and complexity, is roughly equal to the time-scale of the comparatively miniscule progression of developing the cell nucleus. Consider, then, the fairly-new (in the scope of human history) technological development of semi-conductor based computers. Most present-day algorithms and computer intelligence are designed for \emph{centralized} architectures, but the technological shift towards decentralized and distributed computing and information systems is undeniable: one need look no further than the recent proliferation of blockchain technologies, advancements in cloud computing and storage architectures, and the abundance of personal cell phones and daily influence from social media networks. Hence, if any kind of analogy can be drawn to biological evolution, computers and intelligence systems seem to find themselves on the precipice of (or already in the midst of) a dramatic shift toward ``multi-cellular" architectures. State-of-the-art computers in the coming decades might barely resemble their centralized technological ancestors, similarly to how modern multi-cellular life barely resembles its eukaryotic ancestors. 

The above discussion can serve as a philosophical motivation for the ideas and algorithms that are developed in this thesis, though more pragmatic and immediate motivations exist.\footnote{For more concrete examples in present-day engineering, the reader can refer to the following subsections and the ``Biological Comments" section of each chapter.} In particular, we study three closely related formulations of what is broadly referred to as the distributed resource allocation problem. Namely, (i) a nominal convex formulation, in which agent states can take continuous values in a convex set; (ii) a nested stochastic formulation, in which problems of type (i) are nested in a broader stochastic, nonconvex optimization which aims to optimization a parameterization or design variable over realizations of (i); and (iii) a further-constrained instance of (i) in which agent states must belong to a binary set. The approaches for each of (i), (ii), (iii) vary significantly, but they each possess the unifying theme of being \emph{distributed} and using \emph{Newton-like} updates, i.e. the updates utilize second-derivative information of the local agent costs.

The aforementioned work for (i), (ii), (iii) is contained in Chapters~\ref{chap:DANA},~\ref{chap:DiSCRN}, and~\ref{chap:NNN}, respectively. In Chapter~\ref{chap:top}, we supplement the results with a design technique for adding edges to a communication graph, and in Chapter~\ref{chap:active_load} we describe a demonstration that we performed with distributed algorithms performing a frequency control application on the UC San Diego microgrid. We give more specific descriptions and motivation tailored to each chapter in the following subsections.

\section*{Nominal Convex Formulation}

Networked systems endowed with distributed,
multi-agent intelligence are becoming pervasive in modern
infrastructure systems such as power, traffic, and large-scale
distribution networks. However, these advancements lead to new
challenges in the coordination of the multiple agents operating the
network, which are mindful of the network dynamics, and subject to
partial information and communication constraints. 
To this end, distributed convex optimization is a rapidly emerging
field which seeks to develop useful algorithms to manage network
resources in a scalable manner. 
Motivated by the rapid emergence of distributed energy resources, a
problem that has recently gained large attention is that of economic
dispatch.  In this problem, a total net load constraint must be
satisfied by a set of generators which each have an associated cost of
producing electricity. 
However, the existing distributed techniques to solve this problem are
often limited by rate of convergence.  Motivated by this, 
we investigate the design of topology
weighting strategies that build on the Newton method and lead to
improved convergence rates.

\section*{Nested Stochastic Formulation}

As applications emerge which are high dimensional and described by
large data sets, the need for powerful optimization tools has never
been greater. In particular, agents in distributed settings are
commonly given a global optimization task where they must sparingly
exchange local information with a small set of neighboring agents for
the sake of privacy and robust scalability. This architecture can,
however, slow down convergence compared to centralized ones, which is
concerning if obtaining the iterative update information is
costly. Gradient-based methods are commonly used due to their
simplicity, but they tend to be vulnerable to slow convergence around
saddle points. Newton-based methods use second-derivative information
to improve convergence, but they are still liable to be slow in areas
where higher order terms dominate the objective function and even
unstable when the Hessian is ill conditioned. A powerful tool for
combating these Newton-based vulnerabilities is imposing a cubic
regularization on the function's second-order Taylor approximation,
but the current work on this technique does not unify
\emph{distributed}, \emph{stochastic}, and \emph{nonconvex}
elements. Motivated by this, we study the adaptation of the Stochastic
Cubic Regularized Newton approach to solve a distributed nested
optimization problem.

\section*{Binary Formulation}

There has been an explosion of literature
surrounding the design of distributed algorithms for convex
optimization problems and how these pertain to the operation of future
power grids.
A common assumption of these algorithms is the property of
\emph{convexity}, which lends itself to provably optimal solutions
which are \emph{scalable} and \emph{fast}. However, some settings give rise to nonconvex decision sets. For example, in an optimal power dispatch setting, devices available
for providing load-side frequency regulation such as HVAC systems, household appliances, and manufacturing
systems are often limited to
discrete on/off operational modes. It is even preferable to charge populations of electric
vehicles in a discrete on/off manner due to nonlinear battery chemistries. The available tools in optimization for these
nonconvex settings are less mature, and when considering a distributed
setting in which devices act as agents that collectively compute a
solution over a sparse communication graph, the available tools are
significantly less developed.  With this in mind, we
are motivated to develop a \emph{scalable}, \emph{fast} approach for
these binary settings which is amenable to a \emph{distributed}
implementation.

\section*{Topology Design}

Multi-agent systems are pervasive in new
technology spaces such as power networks with distributed energy
resources like solar and wind, mobile sensor networks, and
large-scale distribution systems. In such systems, communication
amongst agents is paramount to the propagation of information, which
often lends itself to robustness and stability of the system. Network
connectivity is well studied from a graph-theoretic standpoint, but
the problem of designing topologies when confronted by engineering
constraints or adversarial attacks is not well addressed by current
works. We are motivated to study the NP-hard graph design problem of
adding edges to an initial topology and to develop a method to solve
it which has both improved performance and
allows for direct application to the aforementioned constrained and
adversarial settings.

\section*{Application: Frequency Regulation with Heterogeneous Energy Resources}

Many recent efforts seek to integrate renewable energy resources with the power grid to reduce the carbon footprint. The high variability associated with wind and solar power can be balanced using  distributed energy resources (DERs) providing ancillary services such as frequency regulation. Consequently, there is a growing interest among market operators in DER aggregations with flexible generation and load capabilities to balance fluctuations in grid frequency and minimize area control errors (ACE). The fast ramping rate and minimal marginal standby cost put many DERs at an advantage against conventional generators and make them suitable for participation in the frequency regulation market. 

The fast ramping rates reduce the required power capacity of DERs to only 10\% of an equivalent generator to balance a frequency drop within 30s~\cite{ZAO-LMC-LA-MTM:19}. However, most individual DERs have small capacities, typically on the order of kWs compared to 10~s of MW for conventional frequency control resources. Commanding the required thousands to millions of DERs to replace existing frequency regulation resources over a large balancing area entails aggregating DERs that are distributed at end points all over the grid on customer premises. The dynamic nature, large number, and distributed location of DERs requires coordination. This is in contrast to existing frequency regulation~\cite{MKM:14} implementation with conventional energy resources. For example, CAISO requires all generators to submit their bids once per regulation interval. Then, the setpoints are assigned centrally to all resources every 2-4 sec without any consideration of operational costs~\cite{CAISO:12}. While distributed control has the potential to enable DER participation in the frequency regulation market (e.g.,~\cite{PS-CYC-JC:18-acc}), there is a general lack of large-scale testing to prove its effectiveness for widespread adoption by system operators. The 2017 National Renewable Energy Laboratory Workshop on Autonomous Energy Grids~\cite{NREL:17} concluded that ``A major limitation in developing 
new technologies for autonomous energy systems is that there are no large-scale test cases (...). These test cases serve a critical role in the development, validation, and dissemination of new algorithms''.

\end{dissertationintroduction}

\chapter{Notation and Preliminaries}
\label{chap:prelims}

\section{Notation}
Let $\real$ and $\real_{+}$ denote the set of real and positive real
	numbers, respectively, and let $\natural$ denote the set of natural
	numbers.  For a vector ${x \in \real^n}$, we denote by $x_i$ the
	$\supscr{i}{th}$ entry of $x$.  For a matrix ${A \in\real^{n \times
			m}}$, we write $A_i$ as the $\supscr{i}{th}$ row of $A$ and
	$A_{ij}$ as the element in the $\supscr{i}{th}$ row and
	$\supscr{j}{th}$ column of $A$, and for $A$ square, $A^{\dagger}$ is the
	Moore-Penrose pseudoinverse of $A$. The
	transpose of a vector or matrix is denoted by $x^\top$ and $A^\top$,
	respectively. We use the shorthand notations ${\ones_n = [1,
		\dots, 1]^\top \in \real^n}$,
	${\zeros_n = [0, \dots, 0]^\top \in \real^n}$, $I_n$
	to denote the ${n \times n}$ identity matrix, and define $\I_n\triangleq I_n -
	\dfrac{\mathbf{1}_n\mathbf{1}^\top_n}{n}$. We refer to this matrix as
	a \emph{pseudo-identity matrix}; note that $\nll{(\I_n)} = \spn
	\{\ones_n\}$. The standard inner
	product of two vectors $x,y \in \real^n$ is written $\langle x,
	y\rangle$, and $x \perp y$ indicates $\langle x, y\rangle = 0$. The orthogonal complement to a span of vectors $a_i$ is written
	$\spn\{a_i\}^\perp$, meaning $x\perp y, \forall x\in\spn\{a_i\},\forall y\in\spn\{a_i\}^\perp$. For a
	real-valued function ${f : \real^n \rightarrow \real}$, the gradient
	vector of $f$ with respect to $x$ is denoted by ${\nabla_x f(x)}$ and
	the Hessian matrix with respect to $x$ by either ${\nabla_{xx} f(x)}$ or ${\nabla^2 f(x)}$. When $f:\real^n\times
	\real^m\rightarrow\real$ takes multiple arguments, we specify the
	differentiation variable(s) as a subscript of $\nabla$. Cartesian products of sets are denoted by a
	superscript, for example, $\{0,1\}^n =
	\{0,1\}\times\dots\times\{0,1\}$. 
	The positive (semi) definiteness and negative (semi) definiteness of a
	matrix ${A \in \real^{n \times n}}$ is indicated by ${A \succ 0}$ and
	${A \prec 0}$ (resp. ${A \succeq 0}$ and ${A \preceq 0}$). The same
	symbols are used to indicate componentwise inequalities on vectors of
	equal sizes. The set of eigenvalues of a symmetric matrix ${A \in
		\real^{n \times n}}$ is ordered as $\mu_1(A) \leq \dots \leq \mu_n(A)$
	with
	associated eigenvectors $v_1, \dots , v_n \in \real^n$.  An orthogonal
	matrix ${T \in \real^{n \times n}}$ has the property ${T^\top T = T
		T^\top = I_n}$ and ${T^\top = T^{-1}}$. For a finite set
	$\mathcal{S}$, $\vert \mathcal{S} \vert$ is the cardinality of the
	set. The
	standard Euclidean norm and the Kronecker product are indicated by
	$\| \cdot \|$, $\otimes$,
	respectively. We denote elementwise operations on vectors $x,y\in\real^n$ as $(x_i y_i)_i = (x_1 y_1, \dots, x_n y_n)^\top$, $(x_i)_i^2 = (x_1^2, \dots, x_n^2)^\top$,
	$(c/x_i)_i = (c/x_1, \dots, c/x_n)^\top$, $\log (x_i)_i = (\log(x_1), \dots, \log(x_n))^\top$, and $(e^{x_i})_i = (e^{x_1}, \dots, e^{x_n})^\top$. The notation $\diag{x}$ indicates the diagonal matrix with entries given by elements of $x$, and
	$\B(x,\eta )$ denotes the closed ball of radius $\eta$ centered at $x$. 
	Probabilities and
	expectations are indicated by $\Pp$ and $\Ex$, respectively. The Dirac delta function centered at $a\in\real$ is denoted by $\delta_a$, and the uniform distribution on $[a,b]$ is denoted by $\U[a,b]$.
	We define the projection
	\begin{equation*}
	\left[ u \right]^+_v := \begin{cases}
	u, & v > 0, \\
	\text{max}\{0,u\}, & v \leq 0.
	\end{cases}
	\end{equation*}

	\section{Graph Theory}
	We refer to~\cite{CDG-GFR:01} as a supplement for the concepts we
	describe throughout this section.
	A network of agents is represented by a graph $\G =
	(\nodes,\mathcal{E})$, assumed undirected, with a node set $\nodes =
	\until{n}$ and edge set $\mathcal{E} \subseteq \nodes \times
	\nodes$. The edge set $\mathcal{E}$ has elements $(i,j) \in
	\mathcal{E}$ for $j \in \nodes_i$, where $\nodes_i \subset \nodes$ is
	the set of neighbors of agent $i \in \nodes$.  The union of neighbors
	to each agent $j \in \nodes_i$ are the 2-hop neighbors of agent $i$,
	and denoted by $\nodes_i^2$. More generally, $\nodes_i^p$, or set of
	$p$-hop neighbors of $i$, is the union of neighbors of agents in
	$\nodes_i^{p-1}$. In Chapter~\ref{chap:DANA}, we consider \emph{weighted} edges for the sake of defining the graph Laplacian; the role of edge weightings and the associated design problem is described in Section~\ref{sec:laplacian-design}. The graph $\G$ then has a \emph{weighted} Laplacian ${L
		\in \real^{n \times n}}$ defined as
	\begin{equation*}
	L_{ij} =
	\begin{cases}
	-w_{ij}, & j \in \nodes_i, j \neq i, \\
	w_{ii}, & j = i, \\
	0, & \text{otherwise},
	\end{cases}
	\end{equation*}
	with weights $w_{ij} = w_{ji} > 0, j \in \nodes_i, j \neq i$, and total
	incident weight $w_{ii}$ on $i\in \nodes$, $w_{ii} = \sum_{j \in
	\nodes_i} w_{ij}$. From Chapter~\ref{chap:DiSCRN} onward, $L$ is taken to be \emph{unweighted}, i.e. $w_{ij} = 1, j \in \nodes_i, j \neq i$. Evidently, $L$ has an eigenvector ${v_1 =
	\ones_n}$ with an associated eigenvalue ${\mu_1 = 0}$, and ${L
	= L^\top \succeq 0}$. The graph is connected i.f.f. $0$ is a
	simple eigenvalue, i.e.~$0 = \mu_1 < \mu_2 \leq \dots \leq \mu_n$, and it is well known that the multiplicity of 
	the zero eigenvalue is equal to the number of
	connected components in the graph~\cite{CDG-GFR:01}.
	
	The Laplacian $L$ can be written via its incidence matrix
	${E \in \{-1,0,1\}^{\vert \mathcal{E} \vert \times n}}$ and a diagonal
	matrix ${X \in \real_{+}^{\vert \mathcal{E} \vert \times \vert
	\mathcal{E} \vert}}$ whose entries are weights $w_{ij}$. Each
	row of $E$ is associated with an edge $(i,j)$ whose $\supscr{i}{th}$
	element is $1$, $\supscr{j}{th}$ element is $-1$, and all other
	elements zero.  Then, $L = E^\top X E$.

	\section{Schur Complement}
	The following lemma will be used in the sequel.
	\begin{lem}
		\cite{FZ:05}\longthmtitle{Matrix Definiteness via Schur
			Complement}\label{lem:schur-comp}
		Consider a symmetric matrix $M$ of the form
		\begin{equation*}
		M = \begin{bmatrix}
		A & B \\ B^\top & C
		\end{bmatrix}.
		\end{equation*}
		If $C$ is invertible, then the following properties hold: \\
		(1) $M \succ 0$ if and only if $C \succ 0$ and $A - BC^{-1}B^\top \succ 0$. \\
		(2) If $C \succ 0$, then $M \succeq 0$ if and only if $A - BC^{-1}B^\top \succeq 0$.
	\end{lem}

	\section{Taylor Series Expansion for Matrix Inverses} \label{ssec:Taylor}
	
	A full-rank matrix $A \in \real^{n \times n}$ has a matrix inverse,
	$A^{-1}$, which is characterized by the relation $AA^{-1} = I_n$. In
	principle, it is not straightforward to compute this inverse via a
	distributed algorithm.  However, if the eigenvalues of $A$ satisfy
	$\vert 1 - \mu_i(A) \vert < 1, \forall \,i \in \mathcal{N}$, then we
	can employ the Taylor expansion to compute its inverse~\cite{GS:98}:
	\begin{align*}
	A^{-1} = \sum_{p=0}^\infty (I_n - A)^p.
	\end{align*}
	
	To quickly see this holds, substitute $B = I_n - A$, multiply both sides by
	$I_n - B$ and reason with $\lim_{p\rightarrow\infty}$.
	Note
	that, if the sparsity structure of $A$ represents a network topology,
	then traditional matrix inversion techniques such as Gauss-Jordan
	elimination still necessitate all-to-all communication. However,
	agents can communicate and compute locally to obtain each term in the
	previous expansion.  If $A$ is normal, it can be seen via the diagonalization of $I_n -
	A$ that the terms of the sum become small as $p$ increases due to the
	assumption on the eigenvalues of $A$~\cite{SF-AI-LS:03}. The convergence of these terms
	is exponential and limited by the slowest converging mode,
	i.e.~$\max{\vert 1 - \mu_i(A) \vert}$.
	
	We can compute an approximation of $A^{-1}$ in finite steps by
	computing and summing the terms up to the $\supscr{q}{th}$ power.  We
	refer to this approximation as a \textit{q-approximation} of $A^{-1}$.

	\section{Cubic-Regularized Newton Algorithm}\label{ssec:pre-CRN}
	We now provide a brief background on the Cubic-Regularized Newton
	method, which will be referred to in Chapter~\ref{chap:DiSCRN}. See~\cite{YN-BTP:06} and~\cite{CC-NG-PT:09P1,CC-NG-PT:10P2}
	for more information.  Consider the problem of minimizing a (possibly
	nonconvex) function $f:\real^d\rightarrow \real$:
	\begin{equation}\label{eq:CRN-prob}
	\underset{x\in\real^d}{\text{min}} \ f(x).
	\end{equation}
	As nonconvex optimization is typically intractable in high dimensions,
	a typical objective is to converge to an
	$\epsilon$-second-order stationary point.
	\begin{defn}\longthmtitle{$\epsilon$-Second-Order Stationary Point}\label{def:so-stat-point}
		A point $x^\star$ is an $\epsilon$-second-order stationary point of
		$f$ if
		\begin{equation}
		\|\nabla_{x} f(x^\star)\|\leq \epsilon \quad \text{and} \quad \lambda_{\text{min}} (\nabla^2_{xx}f(x^\star)) \geq -\sqrt{\rho\epsilon}.
		\end{equation}
	\end{defn}
	Here, $\rho$ is commonly taken to be the Lipschitz constant of
	$\nabla^2_{xx}f$, which we will formalize in
	Section~\ref{sec:prob-form}.

	One useful iterative model for minimizing $f(x^k)$ when the function
	is strictly convex at the current iterate $x^k$ {(or, more accurately, if it is strictly convex on some neighborhood of $x^k$)} is descent on a second-order Taylor expansion around
	$x^k$:{\small
		\begin{equation}\label{eq:so-model}
		\begin{aligned}
		\quad x^{k+1} =  \underset{x}{\argmin} \ \bigg{\{} f(x^k) + (x-x^k)^\top
		\nabla f(x^k) + {\frac{1}{2}}(x-x^k)^\top \nabla^2 f(x^k) (x-x^k)\bigg{\}} = x^k - \nabla^{-2}f(x^k)\nabla f(x^k).
		\end{aligned}
		\end{equation}}
	This closed form expression for $x^{k+1}$ breaks down when $f$ is
	nonconvex due to some eigenvalues of $\nabla^2 f(x^k)$ having negative
	sign. Further, when $\nabla^2 f(x^2)$ is nearly-singular, the update
	becomes very large in magnitude and can lead to instability.
	For this reason, consider amending the second-order model with a
	cubic-regularization term, to
	obtain 
	the cubic-regularized, third-order model of $f$ at $x^k$ as:
	\begin{equation}\label{eq:to-model}
		\begin{aligned}
		m_k(x) \triangleq \bigg{\{}f(x^k) + (x-x^k)^\top \nabla f(x^k) + {\frac{1}{2}}(x-x^k)^\top \nabla^2 f(x^k) (x-x^k) + \frac{\rho}{6}\| x-x^k \|^3\bigg{\}}.
	\end{aligned}
	\end{equation}
	The update is naturally given by a minimizer to this model:
	$x^{k+1} \in \underset{x}{\argmin} \ m_k(x).$
	Unfortunately, this model does not beget a closed-form minimizer as
	in~\eqref{eq:so-model}, nor is it convex if $f$ is not convex. The
	model does, however, become convex for $x$ very far from $x^k$, which
	can be seen by computing the Hessian of $m_k$ as $\nabla^2 m_k(x) =
	\nabla^2 f(x^k) + \rho \| x-x^k\| I_n$. Additionally, $m_k$ is an
	\emph{over-estimator} for $f$, {i.e. $m_k (x) \geq f(x), \forall x$. This} is seen by considering the cubic
	term and recalling Lipschitz properties of $\nabla^2 f$; we describe
	this observation in more detail in Chapter~\ref{chap:DiSCRN}.  Therefore, $m_k$
	possesses some advantages over other simpler submodels as it possesses
	properties of a more standard Newton-based, second-order model while
	being sufficiently conservative.
	
	Finally,~\cite{YC-JD:19} recently showed that simply initializing
	$x=x^k-r\nabla f(x^k)/\|\nabla f(x^k){\|}$ for $r\geq 0$ is sufficient to
	show that gradient descent on $m_k$ converges to the global minimizer
	of~\eqref{eq:to-model} (under light conditions on $r$ and the gradient
	step size).
	
	\section{PT-Inverse}
	
	Next, we introduce the Positive-definite Truncated inverse
	(PT-inverse) and its relevance to nonconvex Newton methods.
	
	\begin{defn}[\hspace{1sp}\cite{SP-AM-AR:19}]\longthmtitle{PT-inverse}\label{def:ptinv}
		Let $A\in\real^{n\times n}$ be a symmetric matrix with an
		orthonormal basis of eigenvectors $Q\in\real^{n\times n}$ and
		diagonal matrix of eigenvalues $\Lambda\in\real^{n\times n}$. Consider a constant $m>0$ and define
		$\vert\Lambda\vert_m\in\real^{n\times n}$ by:
		\begin{equation*}
		(\vert \Lambda \vert_m)_{ii} = \begin{cases} \vert\Lambda_{ii}\vert, & \vert\Lambda_{ii}\vert \geq m, \\
		m, & \text{otherwise.}
		\end{cases}
		\end{equation*}
		The PT-inverse of $A$ with parameter $m$ is defined by $(\vert
		A\vert_m)^{-1} = Q^\top (\vert\Lambda\vert_m)^{-1} Q \succ 0$.
	\end{defn}
	
	The PT-inverse operation flips the sign on the negative eigenvalues of
	$A$ and truncates near-zero eigenvalues to a (small) positive value
	$m$ before conducting the inverse. Effectively, this generates a
	positive definite matrix bounded away from zero to be inverted,
	circumventing near-singular cases. In terms of
	computational complexity, it is on the order of standard
	eigendecomposition (or more generally, singular value
	decomposition), which is roughly $O(n^3)$~\cite{VP-ZC:99}. However,
	we note in Section~\ref{sec:dist-hop} that the matrix to be
	PT-inverted is diagonal, which is $O(n)$.
	
	The PT-inverse is useful for nonconvex Newton
	approaches~\cite{SP-AM-AR:19} in the following sense: first, recall
	that the Newton descent direction of $f$ at $x$ is computed as
	$-\left(\nabla_{xx} f(x)\right)^{-1}\nabla_x f(x)$. For $f$ strictly
	convex, it holds that $\nabla_{xx} f(x) \succ 0$ and the Newton
	direction is well defined and decreases the cost. For (non-strictly)
	convex or nonconvex cases, $\nabla_{xx}f(x)$ will be singular,
	indefinite, or negative definite. A PT-inverse operation remedies
	these cases and preserves the descent quality of the
	method. Additionally, saddle points are a primary concern for
	first-order methods in nonconvex settings~\cite{YD-RP-CG-KC-SG-YB:14},
	and the Newton flavor endowed by the PT-inverse effectively performs a
	change of coordinates on saddles with ``slow" unstable manifolds
	compared to the stable manifolds. We discuss this further in
	Section~\ref{sec:all-hop}.
	
	\section{Set Theory}
	{A limit point $p$ of a set $\Px$ is a point such that any neighborhood
		$\mathcal{B}_\epsilon (p)$ contains a point $p'\in\Px$. A set is closed
		if it contains all of its limit points, it is bounded if it is
		contained in a ball of finite radius, and it is compact if it is both
		closed and bounded. } Let $\mathcal{A}_i = \setdef{p}{a_i^\top p \geq
		b_i}$ be a closed half-space and $\Px = \mathcal{A}_1 \cap \dots \cap
	\mathcal{A}_r \subset \R^m$ be a finite intersection of closed
	half-spaces. If $\Px$ is compact, we refer to it as a
	polytope. Consider a set of points $\mathcal{F} =
	\setdef{p\in\Px}{a_i^\top p = b_i, i\in \mathcal{I} \subseteq
		\until{r}; a_j^\top p \geq b_j, j\in
		\until{r}\setminus\mathcal{I}}$. Let $h = \dim(\spn\{a_i\})$ be the
	dimension of the subspace spanned by the vectors
	$\{a_i\}_{i\in\mathcal{I}}$.  Then, we refer to $\mathcal{F}$ as an
	$(m-h)$-dimensional face of $\Px$. Lastly, denote the affine hull of
	$\mathcal{F}$ as $\aff(\mathcal{F}) = \setdef{p + w}{p,w\in\R^m,
		p\in\mathcal{F}, w \perp \spn\{a_i\}_{i\in\mathcal{I}}}$ and define
	the relative interior of $\mathcal{F}$ as $\relint(\mathcal{F}) =
	\setdef{p}{\exists \epsilon > 0 \text{ s. t. }
		\mathcal{B}_\epsilon(p)\cap\aff(\mathcal{F}) \subset
		\mathcal{F}}$. 
\chapter{Distributed Approximate Newton Algorithms and Weight Design for Constrained Optimization}
\label{chap:DANA}
Motivated by economic dispatch and linearly-constrained resource
allocation problems, this chapter proposes a class of novel
$\distnewton$ algorithms that approximate the standard Newton
optimization method. We first develop the notion of an optimal edge
weighting for the communication graph over which agents implement
the second-order algorithm, and propose a convex approximation for
the nonconvex weight design problem.
This weight design formulates to a nonconvex bilinear
optimization, and we propose a convex approximation that
is loosely based on completing the square to compute
adequate solutions.
We next build on the optimal weight design to develop a
\distnewtondisc algorithm which converges linearly
to the optimal solution for economic dispatch problems with unknown
cost functions and relaxed local box constraints. For the full
box-constrained problem, we develop a \distnewtoncont algorithm
which is inspired by first-order saddle-point methods and rigorously
prove its convergence to the primal and dual optimizers. A main
property of each of these distributed algorithms is that they only
require agents to exchange constant-size communication messages,
which lends itself to scalable implementations. Simulations
demonstrate that the $\distnewton$ algorithms with our weight design
have superior convergence properties compared to existing weighting
strategies for first-order saddle-point and gradient descent
methods.

\section{Bibliographical Comments}  The Newton method for minimizing a
real-valued multivariate objective function is well characterized for
centralized contexts in~\cite{SB-LV:04}.  Another centralized method
for solving general constrained convex problems by seeking the
saddle-point of the associated Lagrangian is developed
in~\cite{AC-EM-SHL-JC:18-tac}. This method, which implements a
saddle-point dynamics
is attractive because its convergence properties can be established.
Other first-order or primal-dual based methods for approaching distributed optimization include~\cite{EM-CZ-SL:17,TD-CB:17,RC-GN:13,SHM-MJ:18}.
However, these methods typically do not incorporate second-order
information of the cost function, which compromises convergence speeds.  The notion of computing an
\emph{approximate} Newton direction in distributed contexts has gained
popularity recently, such as~\cite{AM-QL-AR:17}
and~\cite{EW-AO-AJ:13P1,EW-AO-AJ:13P2}.
In the former work, the authors propose a method which uses the Taylor
series expansion for inverting matrices. However, it assumes that each
agent keeps an estimate of the entire decision variable, which does
not scale well in problems where this variable dimension is equal to
the number of agents in the network. Additionally, the optimization is
unconstrained, which helps to keep the problem decoupled but is
narrower in scope.  The latter works pose a separable optimization
with an equality constraint characterized by the incidence matrix. The
proposed method may be not directly applied to networks with
constraints that involve the information of all agents.  The
papers~\cite{DJ-JX-JM-14,FZ-DV-AC-GP-LS:16,RC-GN-LS-DV:15}
incorporate multi-timescaled dynamics together with a dynamic
consensus step to speed up the convergence of the agreement
subroutine.  These works only consider uniform edge weights, while
sophisticated design of the weighting may improve the
convergence. In~\cite{LX-SB:06}, the Laplacian weight design problem
for separable resource allocation is approached from a $\distgrad$
perspective. Solution post-scaling is also presented, which can be
found similarly in~\cite{MM-RT:07} and~\cite{YS:03} for improving the
convergence of the Taylor series expression for matrix inverses.
In~\cite{SYS-MA-LEG:10}, the authors consider edge weight design to
minimize the spectrum of Laplacian matrices. However, in the Newton
descent framework, the weight design problem formulates as a nonconvex
bilinear problem, which is challenging to solve.  Overall, the current
weight-design techniques that are computable in polynomial time are
only mindful of first-order algorithm dynamics. A second-order
approach has its challenges, which manifest themselves in a bilinear
design problem and more demanding communication requirements, but
using second-order information is more heedful of the problem geometry
and leads to faster convergence speeds.

\section*{Statement of Contributions} In this chapter, we propose a
novel framework to design a weighted Laplacian matrix that is used in
the solution to a multi-agent optimization problem via sparse
approximated Newton algorithms. Motivated by economic dispatch, we
start by formulating a separable resource allocation problem subject
to a global linear constraint and local box constraints, and then
derive an equivalent form without the global constraint by means of a
Laplacian matrix, which is well suited for a distributed framework. We
use this to motivate weighting design of the elements of the Laplacian
matrix and formulate this problem as a bilinear optimization. We
develop a convex approximation of this problem whose solution can be
computed offline in polynomial time. A bound on the \emph{best-case}
solution of the original bilinear problem is also given.

We aim to bridge the gap between classic Newton and $\distnewton$
methods. To do this, we first relax the box constraints and develop a
class of constant step-size discrete-time algorithms.  The Newton
step associated with the unconstrained optimization problem do not
inherit the same sparsity as the distributed communication network. To
address this issue, we consider approximations based on a Taylor
series expansion, where the first few terms inherit certain level of
sparsity as prescribed by  the Laplacian matrix.  We analyze
the approximate algorithms and show their convergence for any
truncation of the series expansion.

We next study the original problem with local box constraints, which
has never been considered in the framework of a distributed Newton
method,
and present a novel continuous-time $\distnewton$ algorithm. The
convergence of this algorithm to the optimizer is rigorously studied
and we give an interpretation of the convergence in the Lyapunov
function sense. Furthermore, through a formal statement of the
proposed $\dana$ ($\distnewton$ algorithm), we find several
interesting insights on second-order distributed methods. We compare
the results of our design and algorithm to a generic weighting design
of $\distgrad$ ($\dgd$) implementations in
simulation. Our weighting design shows superior convergence to $\dgd$.

\section{Problem Statement} \label{sec:prob-statement} Motivated by
the economic dispatch problem, in this section we pose the separable
resource allocation problem that we aim to solve distributively. We
reformulate it as an unconstrained optimization problem whose decision
variable is in the span of the graph Laplacian, and motivate the
characterization of a second-order Newton-inspired method.

Consider a group of agents $\nodes$, indexed by $i \in \nodes$, and a
communication topology given by $\G$. Each agent is associated with a
local convex cost function $f_i : \real \rightarrow \real$. These
agents can be thought of as generators in an electricity market, where
each function argument $x_i \in \real$, $i \in \nodes$ represents the
power that agent $i$ produces at a cost characterized by $f_i$. The
economic dispatch problem aims to satisfy a global load-balancing
constraint $\sum_{i=1}^n x_i = d$ for minimal global cost $f : \real^n
\rightarrow \real$, where $d$ is the total demand. In addition, each
agent is subject to a local linear box constraint on its decision
variable given by the interval
$[\underline{x}_i,\overline{x}_i]$. Then, the economic dispatch
optimization problem is stated as:
\begin{subequations} \label{eq:p1-opt}
\begin{align}
\Pc 1: \ & \underset{x}{\text{min}}
& & f(x) = \sum_{i=1}^n f_i(x_i) \label{eq:p1-cost}\\
& \text{subject to}
& & \sum_{i=1}^n x_i = d, \label{eq:p1-resource-const}\\
&&& \underline{x}_i \leq x_i \leq \overline{x}_i, \quad i =
\until{n}. \label{eq:p1-box-const}
\end{align}
\end{subequations}
Distributed optimization algorithms based on a gradient descent
approach to solve $\Pc 1$ are available~\cite{MZ-SM:15-book}.
However, by only taking into account first-order information of the
cost functions, these methods tend to be inherently slow. As for a
Newton (second-order) method, the constraints make the computation of
the descent direction non-distributed. To see this, consider
only~\eqref{eq:p1-cost}--\eqref{eq:p1-resource-const}.  Recall the
unconstrainted Newton step defined as $\subscr{x}{nt} := -\nabla_{xx}
f(x)^{-1}\nabla_x f(x)$, see e.g.~\cite{SB-LV:04}. In this context,
the equality constraint can be eliminated by imposing ${x_n = d -
\sum_{i=1}^{n-1} x_i}$. Then, \eqref{eq:p1-cost} becomes ${f(x) =
\sum_{i=1}^{n-1} f_i(x_i) + f_n(d - \sum_{i=1}^{n-1} x_i)}$. In
general, the resulting Hessian $\nabla_{xx} f(x)$ is fully populated
and its inverse requires all-to-all communication among agents in
order to compute the second-order descent direction. If we
additionally consider~\eqref{eq:p1-box-const}, interior point methods
are often employed, such as introducing a log-barrier function to the
cost in~\eqref{eq:p1-cost}\cite{SB-LV:04}. The value of the
log-barrier parameter is updated online to converge to a
feasible solution, which exacerbates the non-distributed nature of
this approach. This motivates the design of distributed Newton-like
methods which are cognizant
of~\eqref{eq:p1-resource-const}--\eqref{eq:p1-box-const}.

We eliminate \eqref{eq:p1-resource-const} by introducing a 
network topology as encoded by a Laplacian matrix $L$ associated with $\G$ 
and an initial condition $x^0 \in \real^n$ with some
assumptions.
\begin{assump}\longthmtitle{Undirected and Connected Graph} 	\label{ass:conn-graph}
The weighted graph characterized by $L$ is undirected and connected,
i.e. $L = L^\top$ and $0$ is a simple eigenvalue of $L$.
\end{assump}
\begin{assump}\longthmtitle{Feasible Initial
	Condition} \label{ass:initial}
The initial state $x^0$ satisfies~\eqref{eq:p1-resource-const}, i.e.
\begin{equation*}
\sum_{i=1}^n x_i^0 = d.   
\end{equation*}
\end{assump}
If the problem context does not lend itself well to satisfying
Assumption~\ref{ass:initial}, there is a distributed algorithmic
solution to rectify this via dynamic consensus that can be found
in~\cite{AC-JC:16-auto} which could be modified for a Newton-like
method. Given these assumptions, $\Pc 1$ is equivalent to:
\begin{subequations} \label{eq:p2-opt}
\begin{align}
\Pc 2: \
& \underset{z}{\text{min}}
& & f(x^0 + Lz) = \sum_{i=1}^n f_i(x^0_i + L_i z) \label{eq:p2-cost}\\
& \text{subject to}
& & \underline{x} - x^0 - Lz \preceq \zeros_n, \label{eq:p2-box-const1} \\
&&& x^0 + Lz -\overline{x} \preceq \zeros_n . \label{eq:p2-box-const2}
\end{align}
\end{subequations}
Using the property that $\ones_n$ is an eigenvector of $L$
associated with the eigenvalue $0$, we have that $\ones_n^\top
(x^0 + Lz) = d$.  Newton descent for centralized solvers is given
in~\cite{SB-LV:04}; in our distributed framework, the row space of the
Laplacian is a useful property to
address~\eqref{eq:p1-resource-const}.

\begin{rem}\longthmtitle{Relaxing Assumption~\ref{ass:initial}}\label{rem:robust}
The assumption on the initial condition can render the formulation
vulnerable to implementation errors and cannot easily accommodate
packet drops in a distributed algorithm. A potential workaround for
this is outlined here. Consider, instead
of~\eqref{eq:p1-resource-const} in $\Pc 1$, the $n$ linear
constraints:
\begin{equation}\label{eq:dist-equal}
x + Lz = \overline{d},
\end{equation}
where $\overline{d}\in\real^n , \ones_n^{\top} \overline{d} =
d$ and~\eqref{eq:p1-resource-const} can be recovered by
multiplying~\eqref{eq:dist-equal} from the left by
$\ones_n$. (As an aside, it may be desirable to impose
sparsity on $\overline{d}$ so that only some agents need
access to global problem data). Both $x\in\real^n$ and
$z\in\real^n$ become decision variables, and agent $i$ can
verify the $\supscr{i}{th}$ component of~\eqref{eq:dist-equal}
with one-hop neighbor information. Further, a distributed
saddle-point algorithm can be obtained by assigning a dual
variable to~\eqref{eq:dist-equal} and proceeding as
in~\cite{AC-EM-SHL-JC:18-tac}.

We provide a simulation justification for this approach in
Section~\ref{ssec:sims-robust-dana}, although the analysis of
robustness to perturbations and packet drops is ongoing and
outside the scope of this chapter. For now we strictly impose
Assumption~\ref{ass:initial}. 
\end{rem}

We aim to leverage the freedom given by the elements of $L$ in order
to compute an approximate Newton direction to $\Pc 2$. To this end, we
adopt the following assumption.
\begin{assump}\longthmtitle{Cost Functions}\label{ass:cost}
The local costs $f_i$ are twice continuously differentiable and
strongly convex with bounded second-derivatives given by
\begin{equation*}
0 < \delta_i \leq \dfrac{\partial^2 f_i}{\partial x_i^2} \leq \Delta_i,
\end{equation*}
for every $i\in\nodes$ with given $\delta_i, \Delta_i\in\real_+$.
\end{assump}
This assumption is common in other distributed Newton or Newton-like
methods, e.g.~\cite{AM-QL-AR:17,DJ-JX-JM-14} and in classical convex
optimization~\cite{SB-LV:04,YN:13}. Assumption~\ref{ass:cost} is
necessary to attain convergence in our computation of the Newton
step/direction and to construct the notion of an optimal edge
weighting $L$.
We adopt the shorthands $H(x) := \nabla_{xx}
f(x)$,
$H_\delta := \diag{\delta}$, and $H_\Delta := \diag{\Delta}$ as the
diagonal matrices with elements given by $\partial^2
f_i(x_i)/ \partial x_i^2$, $\delta_i$, and $\Delta_i$,
respectively. 

Next, for the purpose of developing a distributed Newton-like method,
we must slightly rethink the idea of inverting a Hessian matrix. By
application of the chain rule, we have that $\nabla_{zz}f(x^0 + Lz) =
LH(x^0 + Lz)L$. Clearly, $\nabla_{zz}f$ is non-invertible due to the
smallest eigenvalue of $L$ fixed at zero, a manifestation of the
equality constraint in the original problem $\Pc 1$. We instead focus
on the $n-1$ nonfixed eigenvalues of $\nabla_{zz}f$ to employ the Taylor
expansion outlined in Section~\ref{ssec:Taylor}. To this end,
we project $LH(x^0 + Lz)L$ to the $\real^{(n-1)\times (n-1)}$ space with a
coordinate transformation; the justification for this and relation to
the traditional Newton method are made explicitly clear in
Section~\ref{sec:discrete-red}. We seek a matrix $T\in\real^{n\times n}$ satisfying $T^\top T = I_n - \ones_n\ones_n^\top /n$~\cite{SF-AI-LS:03}; the particular matrix $T$ we employ is given as
\begin{equation*}
T\hspace{-1mm} = \hspace{-1mm}\begin{bmatrix}
n\hspace{-1mm} -\hspace{-1mm} 1 + \hspace{-1mm}\sqrt{n} 
& -1 & \cdots & -1 & \dfrac{1}{\sqrt{n}} \\
-1 & \ddots & \cdots & \vdots & \\
\vdots & & \ddots & -1 & \vdots \\
-1 & \cdots & -1 & n\hspace{-1mm} -\hspace{-1mm} 1 + \hspace{-1mm}\sqrt{n} & \\
-1 - \sqrt{n} & \cdots & \cdots & -1 - \sqrt{n} & \dfrac{1}{\sqrt{n}}
\end{bmatrix}\hspace{-1mm}\diag{\hspace{-.5mm}\begin{bmatrix}
\rho \\ 1
\end{bmatrix}\hspace{-.5mm}},
\end{equation*}
where $\rho = \sqrt{n(n+1+2\sqrt{n})}^{-1}\ones_{n-1}$. This choice of $T$ has the effect of projecting the null-space of the Hessian onto the $\supscr{n}{th}$ row and $\supscr{n}{th}$ column, which is demonstrated by defining
$M(x):= JT^\top LH(x)L TJ^\top \in \real^{(n-1)\times (n-1)}$, where $J
= \begin{bmatrix} I_{n-1} & \zeros_{n-1}\end{bmatrix}$. The matrix
$M(x)$ shares its $n-1$ eigenvalues with the $n-1$ nonzero eigenvalues
of $LH(x)L$ at each $x$, and $M(x)^{-1}$ is well defined, which provides
us with a concrete notion of an inverse Hessian. We now adopt the
following assumption.
\begin{assump}\longthmtitle{Convergent Eigenvalues}\label{ass:e-vals}
For any $x$, the eigenvalues of $I_{n-1} - M(x)$, corresponding to
the $n-1$ smallest eigenvalues of $I_{n} - LH(x)L$, are contained in
the unit ball, i.e. $\exists \ \varepsilon < 1$ such
that
\begin{equation*}
-\varepsilon I_{n-1} \preceq I_{n-1} - M(x) \preceq \varepsilon I_{n-1}.
\end{equation*}
\end{assump}
Technically speaking, we are only concerned with arguments of $M$
belonging to the $n-1$ dimensional hyperplane $\setdef{x^0 +
Lz}{z\in\real^n}$, although we consider all $x\in\real^n$ for
simplicity. In the following section, we address
Assumption~\ref{ass:e-vals} (Convergent Eigenvalues) 
by minimizing $\epsilon$ via weight design of the Laplacian. By doing
this, we aim to obtain a good approximation of $M^{-1}$ from the
Taylor expansion with small $q$, which lends itself well to the
convergence of the distributed algorithms in
Sections~\ref{sec:discrete-red} and~\ref{sec:cont-time-dana}.
\section{Weight Design of the Laplacian}\label{sec:laplacian-design}
In this section, we pose the nonconvex weight design
problem on the elements of $L$, which formulates as a bilinear optimization to be solved by a central authority.
To make this problem tractable, we develop a convex approximation and
demonstrate that the solution is guaranteed to satisfy
Assumption~\ref{ass:e-vals}. Next, we provide a lower bound on the
solution to the nonconvex
problem. This gives a measure of performance for
evaluating our approximation.
\subsection{Formulation and Convex Approximation}\label{sec:topology-design}
Our approach hearkens to the intuition on the rate of convergence of
the $q$-approximation of $M(x)^{-1}$. We design a weighting scheme for a
communication topology characterized by $L$ which lends itself to a
scalable, fast approximation of a Newton-like direction. To this end,
we minimize $\underset{i,x}{\max} \vert 1 - \mu_i(M(x)) \vert$:
\begin{subequations}		
\begin{align}
\Pc 3: \quad
& \underset{\varepsilon, L}{\text{min}}
&&\varepsilon \\
& \text{s.t.} 
&&\hspace*{-0.3cm}-\varepsilon I_{n-1} \preceq I_{n-1} - M(x) \preceq
\varepsilon I_{n-1}, \forall x,\label{eq:eval-prob-BMI} \\
&&& L \ones_n = \zeros_n, \ L \succeq 0, \ L = L^\top, \\
&&& L_{ij} \leq 0, j\in \nodes_i, \ L_{ij} = 0,\, j \notin \nodes_i.
\end{align}
\end{subequations}

Naturally, $\Pc 3$ must be solved offline by a central authority because it requires complete information about the local Hessians embedded in $M(x)$, in addition to being a semidefinite program for which distributed solvers are not mature. Even for a centralized solver $\Pc 3$ is hard for a
few reasons, the first being that~\eqref{eq:eval-prob-BMI}
is a function over all possible $x \in \real^n$. To reconcile with
this, we invoke Assumption~\ref{ass:cost} on the cost functions and
write $M_\delta = JT^\top LH_\delta L TJ^\top$ and $M_\Delta = JT^\top
LH_\Delta L TJ^\top$. Then,~\eqref{eq:eval-prob-BMI} is equivalent to
\begin{subequations}
\begin{align}
-(\varepsilon_-  + 1)I_{n-1} + M_\delta &\preceq 0, \label{eq:BMI-epminus} \\
(1 - \varepsilon_+)I_{n-1} - M_\Delta &\preceq 0, \label{eq:BMI-epplus} \\
\varepsilon_- = \varepsilon_+, \label{eq:BMI-equal}
\end{align}
\end{subequations}
where the purpose of introducing $\varepsilon_-$ and $\varepsilon_+$
will become clear in the discussion that follows.

The other difficult element of $\Pc 3$ is the nonconvexity stemming
from~\eqref{eq:BMI-epminus}--\eqref{eq:BMI-epplus} being bilinear in
$L$. 
There are path-following techniques available to solve bilinear
problems of this form~\cite{AH-JH-SB:99}, but simulation results do
not produce satisfactory solutions for problems of the form $\Pc
3$. Instead, we aim to develop a convex approximation of $\Pc 3$ which
exploits its structure. Consider~\eqref{eq:BMI-epminus}
and~\eqref{eq:BMI-epplus} separately by
relaxing~\eqref{eq:BMI-equal}. In fact,~\eqref{eq:BMI-epminus} may be
rewritten in a convex manner. To do this, write $L$ as a weighted
product of its incidence matrix, $L = E^\top X E$. Applying
Lemma~\ref{lem:schur-comp} makes the constraint become
\begin{equation}\label{ineq:convex_constraint}
\begin{bmatrix}
(\varepsilon_- + 1)I_{n-1} & JT^\top E^\top X E 			\\
E^\top X E T J^\top & H_\delta^{-1}
\end{bmatrix} \succeq 0.
\end{equation}
As for~\eqref{eq:BMI-epplus}, consider the approximation $LH_\Delta L
\approx \left(\dfrac{\sqrt{H_\Delta }L + L\sqrt{H_\Delta
}}{2}\right)^2$. This approximation can be thought of as
a rough completion of squares, which lends itself well to our approach of convexifying~\eqref{eq:BMI-epplus}. One should not expect the approximation to be reliably ``better" or ``worse" than the BMI; rather, it is only intended to reflect the original constraint more than a simple linearization. To this end,
substitute this in $M_\Delta$ to get
\begin{equation*}
\begin{aligned}
&\dfrac{1}{4}J T^\top (\sqrt{H_\Delta }L + L\sqrt{H_\Delta})^2
T J^\top \succeq (1 - \varepsilon_+)I_{n-1} \\
&\dfrac{1}{2}J T^\top (\sqrt{H_\Delta }L + L\sqrt{H_\Delta}) T
J^\top
\succeq \sqrt{(1 - \varepsilon_+)}I_{n-1} \\
&\dfrac{1}{2}J T^\top
(\sqrt{H_\Delta }L + L\sqrt{H_\Delta }) T J^\top \succeq (1 - 
\dfrac{\varepsilon_+}{2}
+\dfrac{\varepsilon_+^2}{8} +
O(\varepsilon_+^3))I_{n-1},
\end{aligned}
\end{equation*}
where the second line uses the property that $T J^\top J T^\top = I_n
- \ones_n\ones_n^\top / n$ is idempotent and that
\begin{equation*}
\begin{aligned}
\left(\dfrac{1}{2}J T^\top (\sqrt{H_\Delta}L + L\sqrt{H_\Delta}) T
J^\top\right)^2 &\succeq (1 - \varepsilon_+)I_{n-1} \succeq 0	\\
\Leftrightarrow \dfrac{1}{2}J T^\top (\sqrt{H_\Delta}L +
L\sqrt{H_\Delta}) T J^\top &\succeq \sqrt{1-\varepsilon_+}I_{n-1}
\succeq 0,
\end{aligned}
\end{equation*}
see~\cite{JV-RB:00}. The third line expresses the right-hand
side as a Taylor expansion about $\varepsilon_+ = 0$. Neglecting the
higher order terms $O(\varepsilon_+^3)$ and applying
Lemma~\ref{lem:schur-comp} gives
\begin{equation}\label{ineq:nonconvex_constraint}
\begin{bmatrix}
\dfrac{1}{2}J T^\top (\sqrt{H_\Delta}L + L\sqrt{H_\Delta}) T J^\top - {(1 -
\dfrac{1}{2}\varepsilon_+)I_{n-1}} & \dfrac{1}{\sqrt{8}}\varepsilon_+ I_{n-1} 			\\
\dfrac{1}{\sqrt{8}}\varepsilon_+ I_{n-1} & I_{n-1}
\end{bmatrix} \succeq 0.
\end{equation}

Returning to $\Pc 3$, note that the latter three constraints are
satisfied by $L = E^\top X E$. Then, the approximate reformulation of
$\Pc 3$ can be written as
\begin{equation*}		
\begin{aligned}
\Pc 4: \quad &\underset{\varepsilon_-, \varepsilon_+,
X}{\text{min}}
&&\max(\varepsilon_-, \varepsilon_+) \\
& \text{s.t.}
&& \varepsilon_- \geq 0, \varepsilon_+ \geq 0,			\\
&&& X \succeq 0, \eqref{ineq:convex_constraint}, \eqref{ineq:nonconvex_constraint}.	\\
\end{aligned}
\end{equation*}
This is a convex problem in $X$ and solvable in polynomial time. To
improve the solution, we perform some post-scaling. Take $L^\star_0 =
E^\top X_0^\star E$, where $X_0^\star$ is the solution to $\Pc 4$, and
let $M_{\Delta 0}^\star = JT^\top L^\star_0H_\Delta L^\star_0TJ^\top,
M_{\delta 0}^\star = JT^\top L^\star_0H_\delta
L^\star_0TJ^\top$. Then, consider
\begin{equation*}
\beta = \sqrt{\dfrac{2}{\mu_1(M_{\delta 0}^\star) + \ \mu_{n-1}(M_{\Delta 0}^\star)}},
\end{equation*}
and take $L^\star = \beta L^\star_0$. This shifts the eigenvalues of
$M^\star_0(x)$ to $M^\star(x)$ (defined similarly via $L^\star$) such
that $1 - \mu_1(M_\delta^\star) = -(1 - \mu_{n-1}(M_\Delta^\star))$,
which shrinks $\underset{i,x}{\max}(\vert 1 -
\mu_i(M^\star(x))\vert)$. We refer to this metric as
$\varepsilon_{L^\star} := \underset{i,x}{\max}(\vert 1 -
\mu_i(M^\star(x))\vert)$, and it can be verified that this
post-scaling satisfies
Assumption~\ref{ass:e-vals} with regard to $\varepsilon_{L^\star}$. To see this, first consider scaling $L$ by an arbitrarily small constant, which places the eigenvalues of $I_{n-1} - M(x)$ very close to $1$ and satisfies Assumption~\ref{ass:e-vals}. Then, consider gradually increasing this constant until the lower bound on the minimum eigenvalue and upper bound on the maximum eigenvalue of $I_{n-1} - M(x)$ are equal in magnitude. This is precisely the scaling produced by $\beta$. Then, the solution to $\Pc 4$ followed by
a post scaling by $\beta$ given by $L^\star$ is an approximation of
the solution to the nonconvex problem $\Pc 3$ with the sparsity
structure preserved.

\begin{rem}\longthmtitle{Unknown Local Hessian Bounds}\label{rem:glob-hess-bound}
It may be the case that a central entity tasked with computing some $L^\star$ does not have access to the local bounds $\delta_i, \Delta_i, \forall i$. In this case, globally known bounds $\delta \leq \delta_i, \Delta_i \leq \Delta, \forall i$ can be substituted in place of the local values in the formulation of $\Pc 4$. It can be verified that this will result in a more conservative formulation, and that the resulting $L^\star$ will still satisfy Assumption~\ref{ass:e-vals} at the expense of possibly larger $\varepsilon$.
\end{rem}

\subsection{A Bound on Performance}
We are motivated to find a
``best-case scenario'' 
for our solution given the structural constraints of the
network. Instead of solving
$\Pc 3$ for $L$, we solve it for some $A$ where $A_{ij} = 0$ for
$j\notin \nodes_i^2$, i.e. the two-hop neighbor structure of the
network and sparsity structure of $LH(x)L$. Define $M_A := JT^\top A T
J^\top$. This problem is:
\begin{equation*}\label{eq:lower bound}
\begin{aligned}
\Pc 5: \quad &\underset{\varepsilon, A}{\text{min}}
&&\varepsilon \\
& \text{s.t.}
&&- \varepsilon I_{n-1} \preceq I_{n-1} - M_A \preceq \varepsilon I_{n-1}, \\
&&& A \ones_n = \zeros_n, \ A \succeq 0, \\
&&& A_{ij} = 0, j \notin \nodes^2_i.
\end{aligned}
\end{equation*}
This problem is convex in $A$ and produces a solution $\varepsilon_A$,
which serves as a lower bound for the solution to $\Pc 3$. 
It should not be expected that this lower bound is tight or achievable by ``reverse engineering" an $L^\star$ with the desired sparsity from the solution $A^\star$ to $\Pc 5$, rather, $\varepsilon_{L^\star} - \varepsilon_A$ gives just a rough indication of how close
$\varepsilon_{L^\star}$ is to the conservative lower bound of
$\Pc 3$.
\section{Discrete Time Algorithm for Relaxed Economic Dispatch}
\label{sec:discrete-red}
In this section, we focus on a \emph{relaxed}
version of $\Pc 2$ to develop a direct relation between traditional
discrete-time Newton descent and our distributed, approximate
method. First, we state the relaxed problem and define the approximate
Newton step. We then state the $\distnewtondisc$ algorithm 
and provide a rigorous study of its convergence properties.
\subsection{Characterization of the Approximate Newton
Step}\label{ssec:char_approx_newton}
Even the traditional centralized Newton method is not well-suited to
solve $\Pc 1$ due to the box constraints~\eqref{eq:p1-box-const}. For
this reason, for now we focus on the relaxed problem
\begin{subequations} \label{eq:p6-opt}
\begin{align}
\Pc 6: \
& \underset{x}{\text{min}}
&& \hspace{-12mm} f(x) = \sum_{i=1}^n f_i(x_i), \label{eq:p6-cost}\\
& \text{subject to}
&& \hspace{-12mm} \sum_{i=1}^n x_i = d. \label{eq:p6-resource-const}
\end{align}
\end{subequations}
The equivalent unconstrained problem in $z$ is
\begin{equation} \label{eq:p7-opt} \Pc 7: \
\underset{z}{\text{min}} \quad g(z) := f(x^0 + Lz) =
\sum_{i=1}^n f_i(x^0_i + L_i z).
\end{equation}

\begin{rem}\longthmtitle{Nonuniqueness of Solution}
Given a $z^\star$ which
solves $\Pc 7$, the set of solutions can be characterized by
$\setdef{z^{\star\prime}}{z^{\star\prime} = z^\star + \gamma
	\ones_n\, \ \gamma\in\real}$. The fact that $z^{\star\prime}$ is
a solution is due to $\nulo{(L)} = \spn{(\ones_n)}$, and the fact
that this characterizes the entire set of solutions is due to
$\nulo{(\nabla_{zz}g(z))} = \spn{(\ones_n)}$.
\end{rem}
To solve $\Pc 6,$ we aim to implement a descent method in $x$ via the
dynamics
\begin{equation} \label{eq:short algorithm} x^{+} = x +
\alpha L \subscr{\tilde{z}}{nt},
\end{equation}
where $\subscr{\tilde{z}}{nt}$ is the \emph{approximate} Newton step that we seek to compute distributively, and $\alpha > 0$ is a
fixed step size.

It is true that $\Pc 7$ is unconstrained with respect to $z$, although
we have already alluded to the fact that the Hessian matrix
$\nabla_{zz} g(z) = LH(x+Lz)L$ is rank-deficient stemming
from~\eqref{eq:p6-resource-const}. We now reconcile this by deriving a
well defined Newton step in a reduced variable
$\hat{z}\in\real^{n-1}$.  Consider a change of coordinates by the
orthogonal matrix $T$ defined in Section~\ref{sec:prob-statement} and
write $z=TJ^\top\hat{z}$. Taking the gradient and Hessian of $g(z)$
with respect to $\hat{z}$ gives
\begin{equation*}
\begin{aligned}
\nabla_{\hat{z}} g(z) &=
JT^\top \nabla_z g(z) = JT^\top L \nabla_x f(x+LTJ^\top \hat{z}) \\
\nabla_{\hat{z}\hat{z}} g(z) &= JT^\top LH(x + LTJ^\top \hat{z})LTJ^\top = M(x + LTJ^\top \hat{z}).
\end{aligned}
\end{equation*}
Notice that the zero eigenvalue of $\nabla_{zz} g(z)$ is eliminated by
this projection and the other eigenvalues are preserved. Evaluating at
$x + LTJ^\top \hat{z}\big|_{\hat{z}=0}$, the Newton step in $\hat{z}$
is now well defined as $\subscr{\hat{z}}{nt}
:=-\nabla_{\hat{z}\hat{z}} g(0)^{-1} \nabla_{\hat{z}} g(0) =
-M(x)^{-1}JT^\top L \nabla_x f(x)$.

Consider now a $q$-approximation of $M(x)^{-1}$ given by
$\sum_{p=0}^q (I_{n-1} - M(x))^p$ and return to the original
coordinates to obtain the \emph{approximate} Newton direction
$L\tilde{z}_{nt}$:
\begin{equation*}
L\tilde{z}_{nt} = -LTJ^\top \sum_{p=0}^{q} (I_{n-1} - M(x))^p JT^\top L \nabla_x f(x).
\end{equation*}
With the property that $LTJ^\top JT^\top L = L^2$, rewrite
$L\subscr{\tilde{z}}{nt}$:
\begin{equation}
L\subscr{\tilde{z}}{nt} = -L \sum_{p=0}^q (I_n - LH(x)L)^p L \nabla_x f(x).
\label{eq:znewton-step}
\end{equation}
It can be seen via eigendecomposition of $I_n - LHL$, which is normal, and application
of Assumption~\ref{ass:e-vals} that the terms $L(I_n - LHL)^p$ become
small with $p\rightarrow\infty$ at a rate dictated by
$\varepsilon$. Note that there is a nonconverging mode
of the sum corresponding to the eigenspace spanned by $\ones_n$,
but this is mapped to zero by left multiplication by $L$.
This expression can be computed distributively: each
multiplication by $L$ encodes a communication with the neighbor set of
each agent, and we utilize recursion to perform the computation
efficiently, which is formally described in
Algorithm~\ref{alg:approx-newton}.
\subsection{The $\distnewton$ Algorithm}
We now have the tools to introduce the \\
$\distnewtondisc$ algorithm, or
$\danad$.

\begin{algorithm}
\caption{$\danad_i$}
\label{alg:approx-newton}
\begin{algorithmic}[1]
	\Require $L_{ij}$ for $j \in \{i\} \cup \nodes_i$ and
	communication with nodes $j \in \nodes_i \cup \nodes_i^2$
	\Procedure{Newton$_i$}{$x_i^0,L_i,f_i,q$}
	\State Initialize $x_i\gets x_i^0$
	\Loop 
	\State Compute 
	$\dfrac{\partial f_i}{\partial x_i}$, $\dfrac{\partial^2
		f_i}{\partial x_i^2}$; send to $j\in\nodes_i$, $\nodes^2_i$ \vspace{1mm} \label{alg1:loop first term}
	\State $y_i \gets L_{ii}\dfrac{\partial f_i}{\partial x_i} +
	\sum_{j \in\nodes_i} L_{ij} \dfrac{\partial f_j}{\partial x_j}$
	\State $z_i \gets -y_i$
	\State $p_i \gets 1$
	\While{$p_i \leq q$}		\label{alg1:loop pth term}
	\State 	Acquire $y_j$ from $j \in \nodes_i^{2}$\label{alg1:y-info}
	\State 	$w_i = (I_n - LH(x)L)_i y$\label{alg1:w-comp}
	\State 	$y_i \gets w_i$
	\State 	$z_i \gets z_i - y_i$	\label{alg1:sum z}
	\State $p_i \gets p_i + 1$
	\EndWhile				
	\State Acquire $z_j$ for $j \in \nodes_i$ \label{alg1:loop step}
	\State $x_i \gets x_i + \alpha\left( L_{ii}z_i + \sum_{j \in\nodes_i}L_{ij}z_j\right)$
	\EndLoop
	\State \textbf{return} $x_i$
	\EndProcedure
\end{algorithmic}
\end{algorithm}

The algorithm is constructed directly from~\eqref{eq:short algorithm}
and~\eqref{eq:znewton-step}. The $L\nabla_xf(x^k)$ factor
of~\eqref{eq:znewton-step} is computed first in the loop starting on
line~\ref{alg1:loop first term}. Then, each additional term of the sum
is computed recursively in the loop starting on line~\ref{alg1:loop
pth term}, where $y$ implicitly embeds the
exponentiation by $p$ indicated in~\eqref{eq:znewton-step}, $z$
accumulates each term of the summation of~\eqref{eq:znewton-step}, $w$ is used as an intermediate variable, and $p_i$ is used as a simple counter. We
introduce some abuse of notation by switching to vector and matrix
representations of local variables in line~\ref{alg1:w-comp}; this
is done for compactness and to avoid undue clutter. Note that the
diagonal elements of $H(x)$ are given by $\partial^2 f_j / \partial
x_j^2$ and the matrix and vector operations can be implemented
locally for each agent using the corresponding elements $y_j$,
$L_{ij}$, and $L^2_{ij}$. The one-hop and two-hop communications of
the algorithm are contained in lines~\ref{alg1:loop first term}
and~\ref{alg1:y-info}, where line~\ref{alg1:loop first term} calls
upon local evaluations of the gradient and Hessian. (In principle,
Hessian information could be acquired along with $y_j$ in the first
iteration of the inner loop to utilize one fewer two-hop
communication, but it need only be acquired once per outer loop.)
The information is utilized in local computations indicated the next
line in each case. It is understood that agents perform
communications and computations synchronously.

The outer loop of
the algorithm corresponding to~\eqref{eq:short
algorithm} is performed starting on line~\ref{alg1:loop step}. If
only one-hop communications are available, each outer loop of the
algorithm requires $2q+1$ communications. The process repeats until
desired accuracy is achieved. If $q$ is increased, it requires
additional communications, but the step approximation gains accuracy.
\subsection{Convergence Analysis}~\label{ssec:disc-conv-anl}
This section establishes convergence properties of the $\danad$
algorithm for problems of the form $\Pc 6$. For the sake of cleaner
analysis, we will reframe the algorithm as solving $\Pc 7$ via
\begin{equation} \label{eq:short algorithm-z} z^{+} = z -
\alpha A_q(z) \nabla_z g(z),
\end{equation}
where $A_q(z):= \sum_{p=0}^q (I_n - LH(x^0+Lz)L)^p$. Then, note that
the solution $z^\star$ to $\Pc 7$ solves $\Pc 6$ by $x^\star = x^0 +
Lz^\star$ and that~\eqref{eq:short algorithm-z} is equivalent
to~\eqref{eq:short algorithm}-\eqref{eq:znewton-step} and Algorithm~\ref{alg:approx-newton}.
\begin{rem}\longthmtitle{Initial Condition, 
	Trajectories, \& Solution}\label{rem:init-traj-sol}
Consider an initial condition $z(0)\in\real^n$ with $\ones_n^\top
z(0) = \omega$. Due to $A_q(z)\nabla_z g(z)\perp\ones_n$, the trajectories under~\eqref{eq:short algorithm-z} are contained
in the set $\setdef{z}{z = \tilde{z} + (\omega/n) \ones_n, \
	\tilde{z} \perp \ones_n}$. The solution $x^\star = x^0 +
Lz^\star$ to $\Pc 6$ is agnostic to $(\omega/n) \ones_n$ due to
$\nulo{(L)} = \spn{(\ones_n)}$, so we consider the solution $z^\star$ uniquely satisfying $\ones_n^\top z^\star = \omega$.
\end{rem}
\begin{thm}\longthmtitle{Convergence of  $\danad$}\label{thm:cvg-dana-d}
{\rm Given an initial condition $z(0)\in\real^n$, if
	Assumption~\ref{ass:conn-graph}, on the bidirectional connected
	graph, Assumption~\ref{ass:initial}, on the feasibility of the
	initial condition, Assumption~\ref{ass:cost}, on bounded Hessians,
	and Assumption~\ref{ass:e-vals}, on convergent eigenvalues, hold,
	then the $\danad$ dynamics~\eqref{eq:short algorithm-z} converge
	asymptotically to an optimal solution $z^\star$ of $\Pc 7$
	uniquely satisfying $\ones_n^\top z^\star = \ones_n^\top z(0)$
	for any $q \in \natural$ and $\alpha <
	\dfrac{2(1-\epsilon)}{(n-1)(1+\epsilon)(1-\epsilon^{q+1})}$.  }
\end{thm}

\begin{proof} Consider the discrete-time Lyapunov function
\begin{equation*}
V(z) = g(z) - g(z^\star)
\end{equation*}
defined on the domain $\dom{(V)} = \setdef{z}{\ones_n^\top z =
	\ones_n^\top z(0)}$. From the theorem statement and in
consideration of Remark~\ref{rem:init-traj-sol}, the trajectories of
$z$ under~\eqref{eq:short algorithm-z} are contained in the domain
of $V$, and $V(z) > 0, \forall z\in\dom{(V)}, z\neq z^\star$.  To
prove convergence to $z^\star$, we must show negativity of
\begin{equation}\label{eq:disc-lyap-diff}
V(z^+) - V(z) = g(z^+) - g(z).
\end{equation}
From the weight design of $L$ (Assumption~\ref{ass:e-vals}),
we have $\nabla_{zz}g(z) \preceq (1+\epsilon)I_n$,
$\epsilon\in[0,1)$. This implies
\begin{equation*}
\begin{aligned}
g(z^+) &= g(z) + \nabla_z g(z)^\top (z^+ - z) 
+ \dfrac{1}{2}(z^+ - z)^\top \nabla_{zz}g(z^\prime)(z^+ - z) \\
&\leq g(z) + \nabla_z g(z)^\top (z^+ - z) + 
\dfrac{1+\epsilon}{2}\| z^+ - z\|^2_2,
\end{aligned}
\end{equation*}
which employs the standard quadratic expansion of convex
functions via some $z^\prime$ in the segment
extending from $z$ to $z^+$ (see e.g. $\S 9.1.2$ of~\cite{SB-LV:04}).
Substituting~\eqref{eq:short algorithm-z} gives
\begin{equation}~\label{eq:gz-ineq}
\begin{aligned}
g(z^+) \leq g(z) - \alpha \nabla_z g(z)^\top A_q(z) \nabla_z g(z) 
+ \dfrac{(1+\epsilon)\alpha^2}{2}\| A_q(z) \nabla_z g(z)\|^2_2.
\end{aligned}
\end{equation}
We now show $A_q(z) \succ 0$ by computing its eigenvalues. Note
$\mu_i(I_n -
LH(x^0+Lz)L)\in[-\varepsilon,\varepsilon]\cup\{1\}$. Let $\mu_i(I_n
- LH(x^0+Lz)L) = \eta_i(z)$ for $i\in\until{n-1}$. The terms of
$A_q(z)$ commute and it is normal, so it can be diagonalized as
\begin{equation*}
\begin{aligned}
A_q(z) = W(z) \begin{bmatrix}
\ddots & & & \\
& \dfrac{1 - \eta_i(z)^{q+1}}{1 - \eta_i(z)} & & \\
& & \ddots \\
& & & q+1
\end{bmatrix} W(z)^\top,
\end{aligned}
\end{equation*}
where the columns of $W(z)$ are the eigenvectors of $A_q(z) \succ
0$, the last column being $\ones_n$, and the terms of the diagonal
matrix are its eigenvalues computed by a geometric series.

For now, we only use the fact that $A_q(z)\succ 0$ to justify the
existence of $A_q(z)^{1/2}$. Returning to~\eqref{eq:gz-ineq},
\begin{equation}\label{eq:alpha-bound-0}
\begin{aligned}
g(z^+) &\leq g(z) - \alpha \biggl(\| A_q(z)^{1/2} \nabla_z g(z)\|^2_2  - \dfrac{(1+\epsilon)\alpha}{2}\| A_q(z) \nabla_z g(z)\|^2_2\biggr).
\end{aligned}
\end{equation}
Recall $\nabla_z g(z)\perp\ones_n$ and that $\ones_n$ is an
eigenvector of $A_q(z)$ associated with the eigenvalue
$q+1$. Consider a matrix $\widetilde{A}_q(z)$ whose rows are
projected onto the subspace spanning the orthogonal complement of
$\ones_n$. More precisely, writing $\widetilde{A}_q(z)$ via its
diagonalization gives
\begin{subequations}\label{eq:tilde-Aq}
	\begin{align}
	&\widetilde{A}_q(z) = W(z) \begin{bmatrix}
	\ddots & & & \\
	& \dfrac{1 - \eta_i(z)^{q+1}}{1 - \eta_i(z)} & & \\
	& & \ddots \\
	& & & 0
	\end{bmatrix} W(z)^\top,\label{eq:tilde-Aq-def}	\\
	&\widetilde{A}_q(z) \nabla_z g(z) = A_q(z) \nabla_z g(z),
	\label{eq:tilde-Aq-lin} \\
	&\widetilde{A}_q(z)^{1/2} \nabla_z g(z) = A_q(z)^{1/2} \nabla_z
	g(z).
	\label{eq:tilde-Aq-sqrt}
	\end{align}
\end{subequations}
Combining~\eqref{eq:alpha-bound-0}--\eqref{eq:tilde-Aq} gives the
sufficient condition on $\alpha$:
\begin{equation}\label{eq:alpha-bound-2}
\alpha < \dfrac{2\| \widetilde{A}_q(z)^{1/2} 
	\nabla_z g(z)\|^2_2}{(1+\epsilon)\| \widetilde{A}_q(z) 
	\nabla_z g(z)\|^2_2}.
\end{equation}
Multiply the top and bottom of the righthand side
of~\eqref{eq:alpha-bound-2} by $\| \widetilde{A}_q(z)^{1/2}\|
^2_2$ and apply submultiplicativity of $\| \cdot \|^2_2$:
\begin{equation}\label{eq:Aq-ineq}
\dfrac{2}{(1+\epsilon)\| \widetilde{A}_q(z)^{1/2}\|^2_2} 
\leq \dfrac{2\| \widetilde{A}_q(z)^{1/2} \nabla_z g(z)
	\|^2_2}{(1+\epsilon)\| \widetilde{A}_q(z) \nabla_z g(z)\|^2_2}.
\end{equation}
Finally, we bound the lefthand side of~\eqref{eq:Aq-ineq} from below
by substituting $\eta_i(z)$ with $\epsilon$:
\begin{equation}\label{eq:bound-Aq-half}
\begin{aligned}
\| \widetilde{A}_q(z)^{1/2}\|^2_2 = \sum_i^{n-1} \dfrac{1-\eta_i(z)^{q+1}}{1-\eta_i(z)} \leq (n-1)\dfrac{1-\epsilon^{q+1}}{1-\epsilon}, \quad \forall
z\in\real^n.
\end{aligned}
\end{equation}
Combining~\eqref{eq:bound-Aq-half} with~\eqref{eq:Aq-ineq} gives the
condition on $\alpha$ in the theorem statement and completes the
proof.		
\end{proof}
In practice, we find this to be a very conservative bound on $\alpha$
due to the employment of many inequalities which simplify the
analysis. We note that designing $L$ effectively such that $\epsilon$
is close to zero allows for more flexibility in choosing $\alpha$
large, which intuitively indicates the Taylor approximation of the
Hessian inverse converging with greater accuracy in fewer terms $q$.
\begin{thm}\longthmtitle{Linear Convergence of  $\danad$}\label{thm:cvg-dana-d-exp}
{\rm Given an initial condition $z(0)\in\real^n$ and step size
	$\alpha =
	\dfrac{(1-\epsilon)}{(n-1)(1+\epsilon)(1-\epsilon^{q+1})}$, if
	Assumption~\ref{ass:conn-graph}, on the bidirectional connected
	graph, Assumption~\ref{ass:initial}, on the feasibility of the
	initial condition, Assumption~\ref{ass:cost}, on bounded Hessians,
	and Assumption~\ref{ass:e-vals}, on convergent eigenvalues, hold,
	the $\danad$ dynamics~\eqref{eq:short algorithm-z} converge
	\emph{linearly} to an optimal solution $z^\star$ of $\Pc 7$
	uniquely satisfying $\ones_n^\top z^\star = \ones_n^\top z(0)$
	in the sense that $g(z^+) - g(z) \leq
	-\dfrac{(1-\varepsilon)^4(1+\varepsilon(-\varepsilon)^q)^2\|
		z-z^\star\|_2^2}{2(n-1)^2(1+\varepsilon)^3(1-\varepsilon^{2(q+1)})}$
	for any $q \in \natural$.  }
\end{thm}
\begin{proof}
Define
\begin{equation*}
\begin{aligned}
c_1(z) = \| \widetilde{A}_q(z)^{1/2}\nabla_z g(z)\|_2^2, \quad c_2(z) = \dfrac{(1+\varepsilon)}{2}\|
\widetilde{A}_q(z)\nabla_z g(z)\|_2^2,
\end{aligned}
\end{equation*}
with $\widetilde{A}_q(z)$ defined as
in~\eqref{eq:tilde-Aq-def}.
Recalling~\eqref{eq:alpha-bound-0}--\eqref{eq:tilde-Aq}, consider $\bar{\alpha} = 2\alpha$ as the smallest step size such
that $- \bar{\alpha} c_1(z) + \bar{\alpha}^2c_2(z)$ is not strictly
negative for all $z$, which is obtained from the result of Theorem~\ref{thm:cvg-dana-d}. 
Then,
\begin{equation}\label{eq:c-bounds-1}
\begin{aligned}
-\bar{\alpha} c_1(z) + \bar{\alpha}^2c_2(z) \leq 0 \Rightarrow -\alpha c_1(z) + \alpha^2 c_2 (z) \leq -\alpha^2 c_2(z).
\end{aligned}
\end{equation}
The implication is obtained from the first by substituting
$\bar{\alpha} = 2\alpha$. We now consider an
implementation of $\danad$ with $\alpha$. From~\eqref{eq:alpha-bound-0}
and substituting
via~\eqref{eq:tilde-Aq-lin}--\eqref{eq:tilde-Aq-sqrt}, we obtain
$g(z^+) - g(z) \leq -\alpha c_1(z) + \alpha^2 c_2 (z)$. Combining
this with the second line of~\eqref{eq:c-bounds-1},
\begin{equation}\label{eq:c-bounds-2}
g(z^+) - g(z) \leq -\alpha^2 c_2 (z).
\end{equation}
We seek a lower bound for $\widetilde{A}_q(z)$. Consider its
definition~\eqref{eq:tilde-Aq-def}, where a lower bound can be
obtained by substituting each $\eta_i(z)$ by $-\varepsilon$. Then,
\begin{equation*}
\widetilde{A}_q(z) \succeq \dfrac{1+\varepsilon 
	(-\varepsilon)^{q}}{1 + \varepsilon} \left(I_n - 
\dfrac{\ones_n\ones_n^\top}{n}\right).
\end{equation*}
Returning to~\eqref{eq:c-bounds-2} and applying the definition of $c_2(z)$,
\begin{equation}\label{eq:c-bounds-3}
g(z^+) - g(z) \leq -\dfrac{\alpha^2(1+
	\varepsilon(-\varepsilon)^q)^2}{2(1+\varepsilon)}\| \nabla_z g(z)\|_2^2,
\end{equation}
due to $\nulo(I_n - \ones_n\ones_n^\top/n) = \spn(\ones_n)$ and
$\nabla_z g(z) \perp \ones_n$.

Next, we bound $\| \nabla_z g(z)\|_2^2$. Apply the Fundamental
Theorem of Calculus to compute $\nabla_z g(z)$ via a line integral. Let $z(s) = sz + (1-s)z^\star$. Then,
\begin{equation}\label{eq:c-bounds-4}
\begin{aligned}
\nabla_z g(z) &= \int_0^1 \nabla_{zz} g(z(s))
(z-z^\star) ds .
\end{aligned}
\end{equation}
Applying Assumption~\ref{ass:e-vals} (convergent eigenvalues) gives a
lower bound on the Hessian of $g$, implying a lower bound on its
line integral:
\begin{equation}\label{eq:c-bounds-5}
\begin{aligned}
&\nabla_{zz}g(z)\succeq
(1-\varepsilon)(I-\ones_n\ones_n^\top/n)
\Rightarrow	\\
\int_0^1 &\nabla_{zz} g(z(s))ds \succeq
(1-\varepsilon)(I-\ones_n\ones_n^\top/n).
\end{aligned}
\end{equation}
Factoring out $z-z^\star$ from~\eqref{eq:c-bounds-4} and applying
the second line of~\eqref{eq:c-bounds-5} gives the lower bound
\begin{equation}\label{eq:c-bounds-6}
\begin{aligned}
\| \nabla_z g(z)\|_2^2 \geq (1-\varepsilon)^2\|
z-z^\star\|_2^2,
\end{aligned}
\end{equation}
due to $\nulo(I_n - \ones_n\ones_n^\top/n) = \spn(\ones_n)$ and
$z-z^\star \perp \ones_n$. Combining~\eqref{eq:c-bounds-6}
with~\eqref{eq:c-bounds-3} and substituting $\alpha$:
\begin{equation*}
g(z^+) - g(z) \leq -\dfrac{(1-\varepsilon)^4(1+
	\varepsilon(-\varepsilon)^q)^2\|
	z-z^\star\|_2^2}{2(n-1)^2(1+
	\varepsilon)^3(1-\varepsilon^{2(q+1)})}.
\end{equation*}
\end{proof}
In principle, this result can be extended to any $\alpha$ which is
compliant with Theorem~\ref{thm:cvg-dana-d}; we have chosen this
particular $\alpha$ for simplicity. The methods we employ to arrive at
the results of Theorems~\ref{thm:cvg-dana-d}
and~\ref{thm:cvg-dana-d-exp} are necessarily conservative. However, in
practice, we find that choosing substantially larger $\alpha$
generally converges to the solution faster. Additionally, we find
clear-cut improved convergence properties for larger $q$ (more
accurate step approximation) and smaller $\varepsilon$ (more effective
weight design). Simulations confirm this in Section~\ref{sec:sims-discuss}. 
\section{Continuous Time Distributed Approximate Newton
Algorithm}\label{sec:cont-time-dana}

In this section, we develop a continuous-time Newton-like algorithm to
distributively solve $\Pc 2$ for quadratic cost functions. Our method
borrows from and expands upon known results of gradient-based
saddle-point dynamics~\cite{AC-EM-SHL-JC:18-tac}. We provide a rigorous
proof of convergence and an interpretation of the convergence result
for various parameters of the proposed algorithm.
\subsection{Formulation of Continuous Time Dynamics}
First, we adopt a stronger version of
Assumption~\ref{ass:cost}:
\begin{assump}\longthmtitle{Quadratic Cost Functions}\label{ass:quad-cost}
The local costs $f_i$ are strongly convex and quadratic, i.e. they take the form
\begin{equation*}
f_i(x_i) = \dfrac{1}{2}a_i x_i^2 + b_i x_i, \quad i\in\until{n}.
\end{equation*}
\end{assump}
Note that the Hessian of $f$ with respect to $x$ is now constant, so
we omit the arguments of $H$ and $A_q$ for the remainder of this
section. The dynamics we intend to use to solve $\Pc 2$ are
substantially more complex than those for the problem with no box
constraints, which makes this simplification necessary. In fact, the
quadratic model is very commonly used for generator costs in power
grid operation~\cite{AW-BW-GS:12}.

We aim to solve $\Pc 2$ by finding a saddle point of the associated
Lagrangian $\Ls$.  Introduce the dual variable $\lambda\in\real^{2n}$
corresponding to~\eqref{eq:p2-box-const1}--\eqref{eq:p2-box-const2},
and define $P(z)$ as
\begin{equation*}
P(z) = \begin{bmatrix}
\underline{P}(z) \\ \overline{P}(z)
\end{bmatrix} =
\begin{bmatrix}
\underline{x} - x^0 - Lz \\ x^0 + Lz - \overline{x}
\end{bmatrix} \in\real^{2n}.
\end{equation*}
The Lagrangian of $\Pc 2$ is given by
\begin{equation}
\Ls (z,\lambda) = g(z) + \lambda^\top P(z). \label{eq:Lagrangian}
\end{equation}
We aim to design distributed dynamics which converge to a saddle point
$(z^\star ,\lambda^\star )$ of~\eqref{eq:Lagrangian}, which solves
$\Pc 2$. A saddle point has the property
\begin{equation*}
\Ls (z^\star,\lambda) \leq \Ls (z^\star,\lambda^\star) 
\leq \Ls (z,\lambda^\star), \quad \forall z\in\real^n , \lambda\in\real^n_{\geq 0}.
\end{equation*}
To solve this, consider Newton-like descent dynamics in the primal
variable $z$ and gradient ascent dynamics in the dual variable
$\lambda$ (Newton dynamics are not well defined for linear
functions). First, we state some equivalencies:
\begin{equation}\label{eq:equiv}
\begin{aligned}
&\nabla_z \Ls (z,\lambda) = \nabla_z g(z) + \begin{bmatrix}-L &
L\end{bmatrix}\lambda, \qquad
\nabla_\lambda \Ls (z,\lambda) = P(z), \\
&\nabla_{zz}\Ls (z,\lambda) = LHL, \quad
\nabla_{\lambda\lambda}\Ls (z,\lambda) = \zeros_{2n\times 2n}, \quad
\nabla_{\lambda z}\Ls (z,\lambda) = \nabla_{z \lambda}\Ls
(z,\lambda)^\top
= \begin{bmatrix}-L & L \end{bmatrix}.
\end{aligned}
\end{equation}
The $\distnewtoncont$, or $\danac$, dynamics are given by
\begin{equation}\label{eq:approx-newt-dynamics}
\begin{aligned}
\dot{z} &= -A_q \nabla_z \Ls (z,\lambda), \\
\dot{\lambda} &= \left[ \nabla_\lambda \Ls (z,\lambda)  \right]^+_\lambda .
\end{aligned}
\end{equation}
The descent in the primal variable $z$ is the approximate Newton
direction as~\eqref{eq:short algorithm}, augmented with dual ascent
dynamics in $\lambda$ (one-hop communication) and implemented in
continuous time. The projection on the dynamics in $\lambda$ ensures
that if $\lambda_i(t_0) \geq 0 $ then $ \lambda_i(t) \geq 0$ for all
$t \geq t_0$.

Define $\mathscr{Z}_q: \real^n \times \real^{2n}_{\geq 0} \rightarrow \real^n
\times \real^{2n}$ as the map in~\eqref{eq:approx-newt-dynamics}
implemented by $\danac$. We now make the following assumptions on
initial conditions and the feasibility set.
\begin{assump}\longthmtitle{Initial Dual Feasibility}\label{ass:init-cond-cont-time}
The initial condition $\lambda(0)$ is dual feasible,
i.e. $\lambda(0) \succeq 0$.
\end{assump}
\begin{assump}\longthmtitle{Nontrivial Primal Feasibility}\label{ass:nontrivial-sol}
The feasibility set of $\Pc 2$ is such that $\exists z$ with $P(z) \prec 0$.
\end{assump}
The dynamics $\mathscr{Z}_q$ are not well suited to handle $\lambda$
infeasible, so Assumption~\ref{ass:init-cond-cont-time} is necessary. As for Assumption~\ref{ass:nontrivial-sol}, if it does
not hold, then either $d = \sum \underline{x}$ or $d = \sum
\overline{x}$ or $\Pc 1$ is infeasible, which are trivial
cases. Assuming it does hold, Slater's condition is satisfied and KKT
conditions are necessary and sufficient for solving $\Pc 2$.

Due to the structure of $L$, $\dot{z}$ is computed using only
$(2q+1)$-hop neighbor information. In practice, the quantity $A_q
\nabla_z\mathscr{L}(z,\lambda)$ may be computed recursively over
multiple one-hop or two-hop rounds of communication, with a discrete
step taken in the direction indicated by
$(\dot{z},\dot{\lambda})$. Note that a table statement of this discretized algorithm would be quite similar to Algorithm~\ref{alg:approx-newton} (with the addition of one-hop dynamics in $\lambda$), so we omit it here for brevity.
Discrete-time algorithms to solve this problem do exist, see e.g.~\cite{ER-SM:16-ocam} in which the authors achieve convergence to a ball around the optimizer whose radius is a function of the step size. However, the analysis of discrete-time algorithms to solve $\Pc 2$ via a Newton-like method is outside the scope of this work.
\subsection{Convergence Analysis}
This section provides a rigorous proof of convergence of the
distributed dynamics $\mathscr{Z}_q$ to the optimizer $(z^\star
,\lambda^\star)$ of $\Pc 2$. The solution $x^\star$ to $\Pc 1$ may
then be computed via a one-hop neighbor communication by $x^\star =
x^0 + Lz^\star$.
\begin{thm}\longthmtitle{Convergence of Continuous Dynamics
	$\mathscr{Z}_q$} \label{thm:cvg-dana-c}\rm
If Assumption~\ref{ass:conn-graph}, on the undirected and connected
graph, Assumption~\ref{ass:initial}, on the feasible initial
condition, Assumption~\ref{ass:e-vals}, on convergent eigenvalues,
Assumption~\ref{ass:quad-cost}, on quadratic cost functions,
Assumption~\ref{ass:init-cond-cont-time}, on the feasible dual
initial condition, and Assumption~\ref{ass:nontrivial-sol}, on
nontrivial primal feasibility,
hold, then the solution trajectories under $\mathscr{Z}_q$
assymptotically converge to an optimal point $(z^\star ,
\lambda^\star)$ of $\Pc 2$, where $z^\star$ uniquely satisfies
$\ones_n^\top z^\star = \ones_n^\top z(0)$.
\end{thm}

\begin{proof}
Consider $Q = \begin{bmatrix} A_q^{-1} & 0 \\ 0 &
I_{2n} \end{bmatrix} \succ 0$ and define the Lyapunov function
\begin{equation}\label{eq:vq}
\begin{aligned}
V_Q(z,\lambda) := \dfrac{1}{2}\begin{bmatrix}
z - z^\star \\
\lambda - \lambda^\star
\end{bmatrix}^\top Q \begin{bmatrix}
z - z^\star \\ 
\lambda - \lambda^\star
\end{bmatrix} 
= \dfrac{1}{2}\Big(\|
A_q^{-1/2}(z-z^\star)\|_2^2 + \| (\lambda -
\lambda^\star)\|_2^2\Big).
\end{aligned}
\end{equation}
The time derivative of $V_Q$ along the trajectories of $\mathscr{Z}_q$ is 
\begin{equation}\label{eq:vq-dot}
\begin{aligned}
\dot{V}_Q(z,\lambda) &= \begin{bmatrix}
z - z^\star \\
\lambda - \lambda^\star
\end{bmatrix}^\top Q \begin{bmatrix}
\dot{z} \\
\dot{\lambda}
\end{bmatrix} 
= -(z-z^\star)^\top A_q^{-1}A_q \nabla_z \Ls(z,\lambda)
+ (\lambda-\lambda^\star)^\top \left[ \nabla_\lambda \Ls(z,\lambda) \right]^+_\lambda \\
&\overset{(a)}{\leq} -(z-z^\star)^\top \nabla_z \Ls(z,\lambda) + (\lambda-\lambda^\star)^\top \nabla_\lambda \Ls(z,\lambda) \\
&\overset{(b)}{=} -(z-z^\star)^\top LHL (z-z^\star) - (z-z^\star)^\top \begin{bmatrix} -L & L \end{bmatrix} (\lambda - \lambda^\star) \\ 
&\qquad + (\lambda - \lambda^\star)^\top \begin{bmatrix} -L & L \end{bmatrix}^\top (z-z^\star) = -\| H^{1/2}L (z-z^\star)\|_2^2 \overset{(c)}{<} 0, \quad z\neq z^\star.
\end{aligned}
\end{equation}
The inequality (a) follows from the componentwise relation
$(\lambda_i - \lambda_i^\star)(\left[ \nabla_{\lambda_i} \Ls
\right]^+_{\lambda_i} - \nabla_{\lambda_i} \Ls ) \leq 0$. To see
this, if $\lambda_i > 0$, the projection is inactive and this term
equals zero. If $\lambda_i = 0$, then the inequality follows from
$\lambda_i^\star \geq 0$ and $\left[ \nabla_{\lambda_i} \Ls
\right]^+_{\lambda_i} - \nabla_{\lambda_i} \Ls \geq 0$. The equality
(b) is obtained from an application of the Fundamental Theorem of
Calculus and computing the line integral along the line
$(z(s),\lambda(s)) = s(z,\lambda) + (1-s)(z^\star,\lambda^\star)$ as
follows:
\begin{equation*}
\begin{aligned}
\nabla_z \Ls(z,\lambda) &= \int_0^1 
\Big(\nabla_{zz}\Ls(z(s),\lambda(s))(z-z^\star) + \nabla_{\lambda z}\Ls(z(s),\lambda(s))
(\lambda-\lambda^\star)\Big)ds \\
&= \nabla_{zz}\Ls(z,\lambda)(z-z^\star) + 
\nabla_{\lambda z}\Ls(z,\lambda)(\lambda-\lambda^\star), \\
\nabla_\lambda \Ls(z,\lambda)
&= \int_0^1 
\Big(\nabla_{\lambda\lambda}\Ls(z(s),\lambda(s)(\lambda-\lambda^\star) + \nabla_{z\lambda}\Ls(z(s),\lambda(s))(z-z^\star)\Big)ds \\
&= \nabla_{z\lambda}\Ls(z,\lambda)(z-z^\star),
\end{aligned}
\end{equation*}
where the integrals can be simplified due to $\nabla_{zz}\Ls$ and $\nabla_{\lambda z}\Ls$ constant, as per~\eqref{eq:equiv}. Recalling Remark~\ref{rem:init-traj-sol}, which applies similarly
here, and noticing $\dot{z}\perp\ones_n$, it follows from the theorem
statement that $(z-z^\star)\perp \ones_n$. Additionally, zero is
a simple eigenvalue of $H^{1/2}L$ with a corresponding right
eigenvector $\ones_n$, implying that (c), the last line of~\eqref{eq:vq-dot}, is strict for $z\neq
z^\star$.

Let $\mathcal{S} := \setdefB{(z,\lambda)}{z = z^\star ,
	\lambda \succeq 0}$ be
an asymptotically stable set under the dynamics
$\mathscr{Z}_q$ defined
in~\eqref{eq:approx-newt-dynamics}. 
We aim to show the largest invariant set contained in $\mathcal{S}$ is the optimizer $\{(z^\star,\lambda^\star)\}$,
so we reason with KKT conditions to
complete the convergence argument for $\lambda$. For
$(z,\lambda)\in\mathcal{S}$, clearly primal feasibility is
satisfied. Assumption~\ref{ass:init-cond-cont-time} gives
feasibility of $\lambda (0)$, which is maintained along the
trajectories of $\mathscr{Z}_q$. The stationarity condition
$\nabla_z \Ls(z^\star,\lambda^\star) = 0$ is also satisfied
for $(z,\lambda)\in\mathcal{S}$: 
examine the dynamics $\dot{z}(t) =
-A_q \nabla_z \Ls(z,\lambda) \equiv 0$.
It follows that $\nabla_z \Ls(z,\lambda)_{(z,\lambda)\in\mathcal{S}} = 0$ due to $A_q$ being full rank. Then, each KKT condition has been
satisfied for $(z,\lambda)\in\mathcal{S}$ except
complementary slackness: $P_i(z) \lambda_i = 0$ for $i \in \until{2n}$. We now address this.

Notice the relation $\dot{z} \equiv 0$
implies
\begin{equation}\label{eq:lambda-all-ones}
\lambda(t) = \hat{\lambda} + \phi_{\underline{\lambda}} (t)
\begin{bmatrix}
\ones_n \\ \zeros_n
\end{bmatrix}
+ \phi_{\overline{\lambda}} (t)
\begin{bmatrix}
\zeros_n \\ \ones_n
\end{bmatrix}
\end{equation}
for some constant $\hat{\lambda}\in\real^{2n}$ and possibly time
varying $\phi_{\underline{\lambda}} (t), \phi_{\overline{\lambda}}
(t)\in\real$. This is due to $\nulo{L} = \spn{\{\ones_n\}}$ and
inferring from $\dot{z} \equiv 0$ that $\begin{bmatrix} -L & L
\end{bmatrix}\lambda (t)$ must be constant. 
Additionally, we may infer from the map $\mathscr{Z}_q$ that
$\phi_{\underline{\lambda}} (t), \phi_{\overline{\lambda}}
(t)$ are continuous and piecewise smooth. The
dynamics $\dot{\lambda}$ and
differentiating~\eqref{eq:lambda-all-ones} in time gives
\begin{equation}\label{eq:lambda-dot-all-ones}
\begin{aligned}
\dot{\lambda} = \left[ \nabla_\lambda \Ls(z^\star,\lambda)
\right]^+_\lambda &= \left[ P(z^\star)  \right]^+_\lambda \in \partial \phi_{\underline{\lambda}} (t)
\begin{bmatrix}
\mathbf{1}_{n} \\ \zeros_n
\end{bmatrix} 
+ \partial \phi_{\overline{\lambda}} (t)
\begin{bmatrix}
\mathbf{0}_{n} \\ \mathbf{1}_n
\end{bmatrix},
\end{aligned}
\end{equation}
where $\partial \phi_{\underline{\lambda}} (t)$ and $\partial
\phi_{\overline{\lambda}} (t)$ are subdifferentials with respect to time
of
$\phi_{\underline{\lambda}} (t)$ and $\phi_{\overline{\lambda}} (t)$,
respectively. Then, $\phi_{\underline{\lambda}} (t)$ and
$\phi_{\overline{\lambda}} (t)$ are additionally piecewise linear due
to $P(z^\star)$ constant.
We now state two cases for $\underline{P}(z^\star)$ to prove
$\underline{\lambda}(t)\rightarrow\underline{\lambda}^\star$.

\textbf{Case 1:} $\underline{P}_i(z^\star) = 0$ for at least one
$i\in\until{n}$. Then, $\dot{\underline{\lambda}}_i = 0$ and
from~\eqref{eq:lambda-dot-all-ones} this implies
$\dot{\underline{\lambda}} = \mathbf{0}_{n}$. Reasoning from the
projection dynamics, this implies either $\underline{\lambda}_j = 0$
or $\underline{P}_j(z^\star) = 0$ for each $j$, which satisfies the
complementary slackness condition $\underline{\lambda}_j^\star
\underline{P}_j(z^\star) = 0$ for every $j\in\until{n}$, and we
conclude that $\underline{\lambda} = \underline{\lambda}^\star$ for
$(z,\lambda)\in\mathcal{S}$.

\textbf{Case 2:} $\underline{P}(z^\star)\prec
0$.
Complementary slackness states $\underline{\lambda}_i^\star
\underline{P}_i(z^\star) = 0$ for each $i\in\until{n}$, implying
$\underline{\lambda}^\star = \zeros_n$. The dynamics preserve
$\lambda(t) \succeq 0$, so the quantity $\underline{\lambda}_i -
\underline{\lambda}_i^\star$ is strictly positive for any
$\underline{\lambda}_i \neq \underline{\lambda}_i^\star$. Applying this to the term $(\lambda-\lambda^\star)^\top[\nabla_\lambda\Ls(z,\lambda)]_\lambda^+$ obtained from the second equality (third line) of~\eqref{eq:vq-dot}, and also applying
$\underline{P}(z^\star)=\nabla_\lambda\Ls(z^\star,\lambda)\prec 0$, we obtain $\dot{V}_Q < 0$ for $z =
z^\star, \underline{\lambda} \neq
\underline{\lambda}^\star$.

The inferences of Case 1 (satisfying complementary slackness)
and Case 2 (reasoning with $\dot{V}_Q$) hold similarly for
$\overline{\lambda}$. Then, we have shown that
$\dot{V}_Q(z,\lambda) < 0, \forall (z,\lambda)\in\mathcal{S}\setminus\{(z^\star,\lambda^\star)\}$.
Asymptotic convergence to the primal and dual optimizers of
$\Pc 2$ follows from the LaSalle Invariance Principle~\cite{HKK:02}.
\end{proof}
\subsection{Interpretation of the Convergence Result}\label{ssec:interpretation}
For fast convergence, it is desirable for the ratio $\dot{V}_Q / V_Q <
0$ to be large in magnitude for any $(z,\lambda)\in\real^n \times
\real_{+}^{2n}$. Recall the diagonalization of $A_q$ and use this to
compute $A_q^{-1}$:
\begin{align*}
A_q &= W \begin{bmatrix}
\dfrac{1 - \eta_1^{q+1}}{1 - \eta_1} & & & \\
& \ddots & & \\
& & \dfrac{1 - \eta_{n-1}^{q+1}}{1 - \eta_{n-1}} \\
& & & q+1
\end{bmatrix} W^\top ,\displaybreak[0]\\
A_q^{-1} &= W \begin{bmatrix}
\dfrac{1 - \eta_1}{1 - \eta_1^{q+1}} & & & \\
& \ddots & & \\
& & \dfrac{1 - \eta_{n-1}}{1 - \eta_{n-1}^{q+1}} \\
& & & (q+1)^{-1}
\end{bmatrix} W^\top .
\end{align*}
Next, write $z - z^\star = \zeta_1w_1 + \dots + \zeta_{n-1}w_{n-1}$ as
a weighted sum of the eigenvectors $w_i$ of $I_n - LHL$. Note that we
do not need $w_n = \mathbf{1}_{n}$ for this representation due to $z -
z^\star \perp w_n$. Then, $V_Q = \sum_{i=1}^{n-1} \zeta_i^2
(1-\eta_i)/(1-\eta_i^{q+1}) + V_\lambda$, where $V_\lambda :=
\vert\vert \lambda - \lambda^\star \vert\vert_2^2$. Additionally, note
that $LHL$ and $A_q^{-1}$ share eigenvectors, so $\dot{V}_Q \leq
-\sum_{i=1}^{n-1} \zeta_i^2 (1-\eta_i)$. 
Toward this end, we can write
\begin{equation*}
\dfrac{\dot{V}_Q}{V_Q} \leq \dfrac{-\sum_{i=1}^{n-1} 
\zeta_i^2(1- \eta_i)}{\sum_{i=1}^{n-1} \zeta_i^2\left(\dfrac{1- 
	\eta_i}{1-\eta_i^{q+1}}\right) + V_\lambda}.
\end{equation*}
To interpret this, first reason
with the values of $q$. Consider $q=0$, which is analogous to a
gradient-based method. Then, the rational in the sum contained in the
denominator is equal to one and there is no \emph{weighting}, in a
sense, to the step direction. In other words, if the value of
$\zeta_i$ happens to be large in magnitude corresponding to the
eigenvector $w_i$ of $\nabla_{zz}\mathscr{L}$ whose corresponding
eigenvalue $(1-\eta_i)$ is small in magnitude, then that term does not
appropriately dominate the numerator relative to each other term and
the quantity $\dot{V}_Q / V_Q$ is small in magnitude. On the other hand,
if
$q$ is large, then the quantity $1-\eta_i^{q+1}$ is close to $1$, and
the terms of the sums in the numerator and denominator have the effect
of ``cancelling'' one another, which provides more uniform convergence
on the trajectories of $z$. 
In addition, if the values of $\eta_i$ are small in magnitude,
i.e. our weight design on $L$ was relatively successful, the quantity $1-\eta_i^{q+1}$ approaches $1$ more quickly
and the effect of a particular $\zeta_i$ being large relative to the
other terms in the sum is diminished for any particular $q$.

Note that, although we have framed this argument as an improvement
over the gradient technique, it may be the case that for a particular
time $t$ the decomposition on $z(t)$ may have a large $\zeta_i$
corresponding to $1-\eta_i$ large. This actually provides superior
momentary convergence compared to a Newton-like method. However, we
contend that the oscillatory nature of the trajectories over the
entire time horizon gives way to improved convergence from the Newton
flavor of our algorithm. This is confirmed in simulation.

Finally, it is apparent that choosing $q$ even is (generally speaking)
superior to $q$ odd: the quantity $1-\eta_i^{q+1}$ may take values in
$\left[ 1-\varepsilon^{q+1}, 1+\varepsilon^{q+1}\right]$, as opposed
to odd $q$ for which $1-\eta_i^{q+1}$ takes values in $\left[
1-\varepsilon^{q+1}, 1\right]$. We would like this quantity to be
large so the magnitude of $\dot{V}_Q /V_Q$ is large. This observation
of choosing even $q$ to prompt superior convergence is confirmed in
simulation.

This discussion neglects the $V_\lambda$ term which may be large
for arbitrarily "bad" initial conditions $\lambda(0) \succeq
0$. However, the ascent direction in $\lambda$
is clearly more effective for $z$ nearly optimal, so this term is
``cooperative'' in the sense that its decay roughly corresponds to
the decay of the Lyapunov term in $z$.

To summarize, gradient methods neglect the curvature of the underlying
cost function, which dictates the convergence properties of descent
algorithms. By weighting the descent direction by $A_q$, we elegantly
capture this curvature in a distributed fashion and the solution
trajectory reflects this property. We now provide a remark on
convergence of the algorithm for nonquadratic costs that are well approximated by quadratic functions.

\begin{rem}\longthmtitle{Convergence of $\danac$ for Approximately Quadratic Costs}
Instead of Assumption~\ref{ass:quad-cost} (quadratic costs), let
Assumption~\ref{ass:cost} (general costs) hold and consider the
dynamics
\begin{equation}\label{eq:cont-dyn-nonquadr}
\begin{aligned}
\dot{z} &= -A_q(z) \nabla_z \Ls (z,\lambda), \\
\dot{\lambda} &= \left[ \nabla_\lambda \Ls (z,\lambda)  \right]^+_\lambda .
\end{aligned}
\end{equation}
Let $H^\prime:= \dfrac{H_\Delta + H_\delta}{2}$ and
$A_q^\prime:=\sum_{p=0}^q (I_n - LH^\prime L)^p$. In a sense, these
matrices are obtained from quadratic approximations of the
nonquadratic costs $f_i$, i.e. $\left| \dfrac{\partial^2
	f_i}{\partial x_i^2} - H^\prime_{ii} \right| \leq
\dfrac{\Delta_i - \delta_i}{2}$. Use $Q = \begin{bmatrix}
A_q^{\prime -1} & 0 \\ 0 & I_{2n} \end{bmatrix}$ to define the
quadratic Lyapunov function $V_Q(z,\lambda)$ as
in~\eqref{eq:vq}. Differentiating along the trajectories
of~\eqref{eq:cont-dyn-nonquadr} now gives
\begin{equation*}
\dot{V}_Q(z,\lambda) = \dot{V}^\prime_Q(z,\lambda) + U(e,z,\lambda),
\end{equation*}
where $e$ gives some measure of how much the functions deviate from
quadratic and $U(0,z,\lambda) = 0$. The
$\dot{V}^\prime_Q(z,\lambda)$ is obtained by decomposing the
dynamics~\eqref{eq:cont-dyn-nonquadr} as
\begin{equation*}
\begin{aligned}
\dot{z} &= -A_q^\prime \nabla_z \Ls (z,\lambda) + u(e,z,\lambda), \\
\dot{\lambda} &= \left[ \nabla_\lambda \Ls (z,\lambda)
\right]^+_\lambda .
\end{aligned}
\end{equation*}
and including only the terms without $u(e,z,\lambda)$, where the
remaining terms are captured by $U(e,z,\lambda)$. $U$ and $u$
are continuous functions of $e$, and $u(0,z,\lambda) = 0$. Applying
the convergence argument of Theorem~\ref{thm:cvg-dana-c} to
$V^\prime_Q(z,\lambda)$, the continuity of $U$ and $u$ imply
$\dot{V}^\prime_Q(z,\lambda) < -U(\bar{e},z,\lambda)$ for
sufficiently small $\bar{e}$. Therefore, $\dot{V}_Q(z,\lambda) < 0$
for functions that are well approximated by quadratic functions.
\end{rem}
\section{Simulations and Discussion}		\label{sec:sims-discuss}
In this section, we implement our weight design and verify the
convergence of the $\distnewton$ algorithm in each of the discrete-time (relaxed) and continuous time (box-constrained) settings.
\subsection{Weight Design}
To evaluate the weight design posed in
Section~\ref{sec:laplacian-design} we use quadratic costs in
accordance with Assumption~\ref{ass:quad-cost}, i.e. $\delta_i = \Delta_i = a_i, \forall i$. We do this in order to isolate the
other parameters for this part of the study. Consider the following metrics: the
solution to $\Pc 4$ followed by the post-scaling by $\beta$ gives
$\varepsilon_{L^\star} := \max(\vert 1 - \mu_i(M^\star)\vert )$; this
metric represents the convergence speed of $\distnewton$ when applying
our proposed weight design of $L$.  Using the same topology $(\nodes,
\mathcal{E})$, the solution to $\Pc 5$ gives the metric
$\varepsilon_A$. Note that $\varepsilon_A$ is a \emph{best-case}
estimate of the weight design problem; however, ``reverse engineering" an $L^\star$ from the solution $A^\star$ to $\Pc 5$ is both intractable and generally likely to be infeasible. With this in mind, the metric $\varepsilon_A$ is a very conservative lower bound, whereas $\varepsilon_{L^\star}$ is the metric for which we can
compute a feasible $L^\star$. The objective of each problem is to minimize the
associated $\varepsilon$; to this end, we aim to characterize the
relationship between network parameters and these metrics. We ran 100
trials on each of 16 test cases which encapsulate a variety of
parameter cases: two cases for the cost coefficients, a \textit{tight}
distribution $a_i \in \mathcal{U}\left[ 0.8, 1.2\right]$ and a
\textit{wide} distribution $a_i \in \mathcal{U}\left[ 0.2,
5\right]$. For topologies, we randomly generated connected graphs
with network size $n \in \{10, 20, 30, 40, 50\}$, a \textit{linearly}
scaled number of edges $\vert\mathcal{E}\vert = 3n$, and a
\textit{quadratically} scaled number of edges $\vert\mathcal{E}\vert =
0.16n^2$ for $n \in \{30, 40, 50\}$. The linearly scaled
connectivity case corresponds to keeping the average degree of a node
constant for increasing network sizes, while the quadratically scaled
case roughly preserves the proportion of connected edges to total
possible edges, which is a quadratic function of $n$ and equal to
$n(n-1)/2$ for an undirected network. The results are depicted in
Table~\ref{L-design-table}, where the quadratically scaled cases are indicated by boldface. This gives the mean $\Sigma$ and standard
deviation $\sigma$ of the distributions for \emph{performance}
$\varepsilon_{L^\star}$ and \emph{performance gap}
$\varepsilon_{L^\star} - \varepsilon_A$.
\begin{table}[]
\centering
\caption{Laplacian Design. Quadratically-scaled number-of-edge cases are indicated by boldface.}
\label{L-design-table}
\begin{tabular}{|c||c|c|c|c|}
	\hline
	\begin{tabular}[c]{@{}c@{}} $a_i \in \mathcal{U} \left[ 0.8, 1.2\right]$ \\ $b_i \in \mathcal{U} \left[ 0, 1\right]$\end{tabular} & $\Sigma ( \varepsilon_{L^\star})$ & $\sigma (\varepsilon_{L^\star})$ & $\Sigma ( \varepsilon_{L^\star} - \varepsilon_A)$ & $\sigma ( \varepsilon_{L^\star} - \varepsilon_A)$ \\ \hline
	\begin{tabular}[c]{@{}c@{}}$n = 10$ \\ $\vert \mathcal{E} \vert = 30$ \end{tabular}  & 0.6343 & 0.0599 & 0.2767 & 0.0186 \\ \hline
	\begin{tabular}[c]{@{}c@{}}$n = 20$ \\ $\vert \mathcal{E} \vert = 60$ \end{tabular}  & 0.8655 & 0.0383 & 0.2879 & 0.0217 \\ \hline
	\begin{tabular}[c]{@{}c@{}}$n = 30$ \\ $\vert \mathcal{E} \vert = 90$ \end{tabular}  & 0.9100 & 0.0250 & 0.2666 & 0.0233 \\ \hline
	\begin{tabular}[c]{@{}c@{}}$n = 40$ \\ $\vert \mathcal{E} \vert = 120$ \end{tabular}  & 0.9303 & 0.0201 & 0.2501 & 0.0264 \\ \hline
	\begin{tabular}[c]{@{}c@{}}$n = 50$ \\ $\vert \mathcal{E} \vert = 150$ \end{tabular}  & 0.9422 & 0.0175 & 0.2375 & 0.0264 \\ \hline
	\begin{tabular}[c]{@{}c@{}}$n = 30$ \\ $\mathbf{\vert \mathcal{E} \vert = 144}$ \end{tabular}  & 0.7266 & 0.0324 & 0.2973 & 0.0070 \\ \hline
	\begin{tabular}[c]{@{}c@{}}$n = 40$ \\ $\mathbf{\vert \mathcal{E} \vert = 256}$ \end{tabular}  & 0.6528 & 0.0366 & 0.2829 & 0.0091 \\ \hline
	\begin{tabular}[c]{@{}c@{}}$n = 50$ \\ $\mathbf{\vert \mathcal{E} \vert = 400}$ \end{tabular}  & 0.5840 & 0.0281 & 0.2641 & 0.0101 \\ \hline \hline
	\begin{tabular}[c]{@{}c@{}} $a_i \in \mathcal{U} \left[ 0.2, 5\right]$ \\ $b_i \in \mathcal{U} \left[ 0, 1\right]$\end{tabular} & $\Sigma ( \varepsilon_{L^\star})$ & $\sigma (\varepsilon_{L^\star})$ & $\Sigma ( \varepsilon_{L^\star} - \varepsilon_A)$ & $\sigma ( \varepsilon_{L^\star} - \varepsilon_A)$ \\ \hline
	\begin{tabular}[c]{@{}c@{}}$n = 10$ \\ $\vert \mathcal{E} \vert = 30$ \end{tabular}  & 0.6885 & 0.0831 & 0.3288 & 0.0769 \\ \hline
	\begin{tabular}[c]{@{}c@{}}$n = 20$ \\ $\vert \mathcal{E} \vert = 60$ \end{tabular}  & 0.8965 & 0.0410 & 0.3241 & 0.0437 \\ \hline
	\begin{tabular}[c]{@{}c@{}}$n = 30$ \\ $\vert \mathcal{E} \vert = 90$ \end{tabular}  & 0.9389 & 0.0254 & 0.2878 & 0.0395 \\ \hline
	\begin{tabular}[c]{@{}c@{}}$n = 40$ \\ $\vert \mathcal{E} \vert = 120$ \end{tabular} & 0.9539 & 0.0189 & 0.2830 & 0.0355 \\ \hline
	\begin{tabular}[c]{@{}c@{}}$n = 50$ \\ $\vert \mathcal{E} \vert = 150$ \end{tabular} & 0.9628 & 0.0168 & 0.2590 & 0.0335 \\ \hline
	\begin{tabular}[c]{@{}c@{}}$n = 30$ \\ $\mathbf{\vert \mathcal{E} \vert = 144}$ \end{tabular} & 0.7997 & 0.0520 & 0.3587 & 0.0524 \\ \hline
	\begin{tabular}[c]{@{}c@{}}$n = 40$ \\ $\mathbf{\vert \mathcal{E} \vert = 256}$ \end{tabular} & 0.7339 & 0.0550 & 0.3688 & 0.0569 \\ \hline
	\begin{tabular}[c]{@{}c@{}}$n = 50$ \\ $\mathbf{\vert \mathcal{E} \vert = 400}$ \end{tabular} & 0.6741 & 0.0487 & 0.3543 & 0.0425 \\ \hline
\end{tabular}
\end{table}

From these results, first note that the \emph{tightly} distributed
coefficients $a_i$ result in improved $\varepsilon_{L^\star}$ across
the board compared to the \emph{widely} distributed coefficients. We
attribute this to the approximation $LHL \approx
\left(\dfrac{\sqrt{H}L + L\sqrt{H}}{2}\right)^2$ being more accurate
for roughly homogeneous $H = \diag{a_i}$. Next, it is clear that in
the cases with \emph{linearly} scaled edges, $\varepsilon_{L^\star}$
worsens as network size increases.  This is intuitive: the
\textit{proportion} of connected edges in the graph decreases as
network size increases in these cases. This also manifests itself in
the performance gap $\varepsilon_{L^\star} - \varepsilon_A$ shrinking,
indicating the \emph{best-case} solution $\varepsilon_A$ (for which a
valid $L$ does not necessarily exist) degrades even quicker as a
function of network size than our solution $\varepsilon_{L^\star}$. On
the other hand, $\varepsilon_{L^\star}$ substantially improves as
network size increases in the \emph{quadratically} scaled cases, with
a roughly constant performance gap $\varepsilon_{L^\star} -
\varepsilon_A$. Considering this relationship between the linear and
quadratic scalings on $\vert \mathcal{E}\vert$ and the metrics
$\varepsilon_{L^\star}$ and $\varepsilon_A$, we get the impression
that both proportion of connectedness and average node degree play a
role in both the effectiveness of our weight-designed solution
$L^\star$ and the best-case solution.  For this reason, we postulate
that $\varepsilon_{L^\star}$ remains roughly constant in large-scale
applications if the number of edges is scaled subquadratically as a
function of network size; equivalently, the convergence properties of
$\distnewton$ algorithm remain relatively unchanged when using our
proposed weight design and growing the number of communications per
agent sublinearly as a function of $n$.

\subsection{Discrete-Time Distributed Approx-Newton}

Consider solving $\Pc 6$ with $\danad$ for a network of $n = 100$ generators and $\vert \mathcal{E}
\vert = 250$ communication links. The local computations required of each generator
are simple vector operations whose dimension scales linearly with the
network size, which can be implemented on a microprocessor. The graph
topology is plotted in Figure~\ref{fig:graph}. The problem parameters are given by
\begin{equation*}
\begin{aligned}
&f_i(x_i) = \dfrac{1}{2}a_i x_i^2 + b_i x_i + c_i \sin{(x_i + \theta_i)}, \quad &&a_i\in \mathcal{U} [2,4], \quad b_i\in \mathcal{U} [-1,1], \\ 
&c_i\in \mathcal{U} [0,1], \quad \theta_i\in \mathcal{U} [0,2\pi], \quad &&d = 200, \quad x^0 = (d/n)\ones_n.
\end{aligned}
\end{equation*}
Note that $0 < a_i - c_i \leq \dfrac{\partial^2 f_i}{\partial x_i^2}
\leq a_i + c_i$ satisfies Assumption~\ref{ass:cost}. We compare to
the $\dgd$ and weight design policies for resource allocation described in~\cite{LX-SB:06}, along with 
an ``unweighted'' version of~\cite{LX-SB:06} in the sense that
$L$ is taken to be the degree matrix minus the adjacency matrix of the
graph, followed by the post-scaling described in
Section~\ref{sec:topology-design} to guarantee convergence. The
results are given in Figure~\ref{fig:qcomp}, which show linear
convergence to the optimal value as the number of iterations
increases, with fewer iterations needed for larger $q$. We note a
substantially improved convergence over the $\dgd$ methods, even for the
$q=0$ case which utilizes an equal number of agent-to-agent
communications as $\dgd$. This can be attributed in-part to the superior weight design of our method, which is cognizant of second-order information.

In addition, in Figure~\ref{fig:qcomp} we plot convergence of DGD, weighted by the one-sided design scheme in~\cite{LX-SB:06}, compared to our two-sided design with $q=0$, for cases in which only a universal bound on $\delta_i$, $\Delta_i$ is known (namely, using $\underline{\delta} \leq \delta_i, \Delta_i \leq \overline{\Delta}, \forall i$, as in Remark~\ref{rem:glob-hess-bound}). We note an improved convergence in each case for the locally known bounds versus the universal bound, while the locally weighted DGD method outperforms our $q=0$ two-sided globally weighted method by a slight margin.

\begin{figure}[h]
\centering
\includegraphics[scale = 0.5]{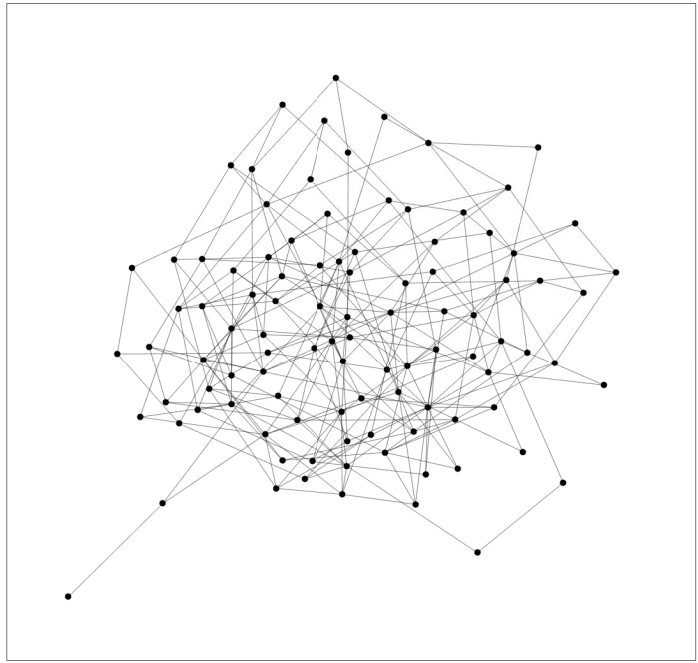}
\caption{Communication topology used for discrete-time numerical study; $n=100,\vert\E\vert = 250$.}
\label{fig:graph}
\end{figure}
\begin{figure}[h]
\centering
\includegraphics[scale = 0.6]{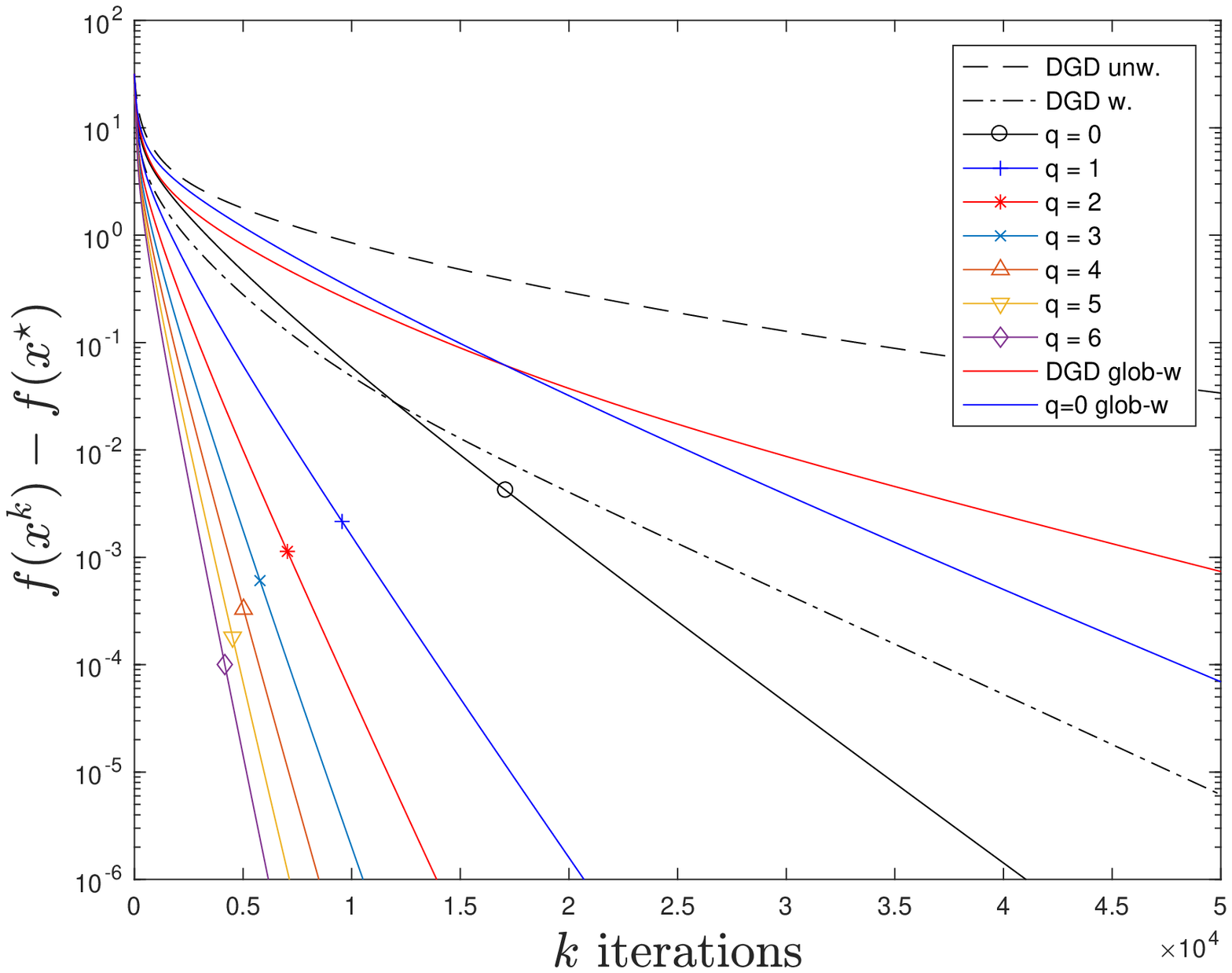}
\caption{Comparison of weighted and unweighted $\dgd$ versus
	$\danad$ with various $q$ for solving $\Pc 6$; $n=100,\vert\E\vert = 250$.}
\label{fig:qcomp}
\end{figure}
\subsection{Continuous-Time Distributed Approx-Newton}
We now study $\danac$ for solving $\Pc 1$ for
a simple $3$ node network with two edges $\E = \{\{1,2\},
\{2,3\}\}$ for the sake of visualizing trajectories. The problem parameters are given by
\begin{equation*}
\begin{aligned}
&f_1(x_1) = \dfrac{1}{4}x_1^2 + \dfrac{1}{2}x_1, \quad f_2(x_2) = \dfrac{3}{4}x_2^2 + \dfrac{1}{2}x_2, \quad f_3(x_3) = 2x_3^2 + \dfrac{1}{2}x_3, \\
&\underline{x} = \begin{bmatrix}0.2 & 2.5 & 1.5\end{bmatrix}^\top,
\quad
\overline{x} = \begin{bmatrix}1 & 6 & 4\end{bmatrix}^\top, \quad d = 6, \\
&x^0 = \begin{bmatrix}5 & -1 & 2\end{bmatrix}^\top, \quad z(0) = \zeros_3, \quad
\underline{\lambda}(0) = \begin{bmatrix} 1.5 & .5 & 0\end{bmatrix},
\quad \overline{\lambda}(0) = \begin{bmatrix}0 & 2 & 1\end{bmatrix}
\end{aligned}
\end{equation*}
Note that $x^0$ is infeasible with respect to
$\underline{x},\overline{x}$; all that we require is it satisfies
Assumption~\ref{ass:initial} (feasible with respect to $d$). We plot the
trajectories of the $3$-dimensional state projected onto the plane
orthogonal to $\ones_3$ under various $q$. Figure~\ref{fig:newton-traj} shows this, with a zoomed look at the optimizer in Figure~\ref{fig:newton-traj-zoomed}.

\begin{figure}[h]
\centering
\includegraphics[scale = 0.6]{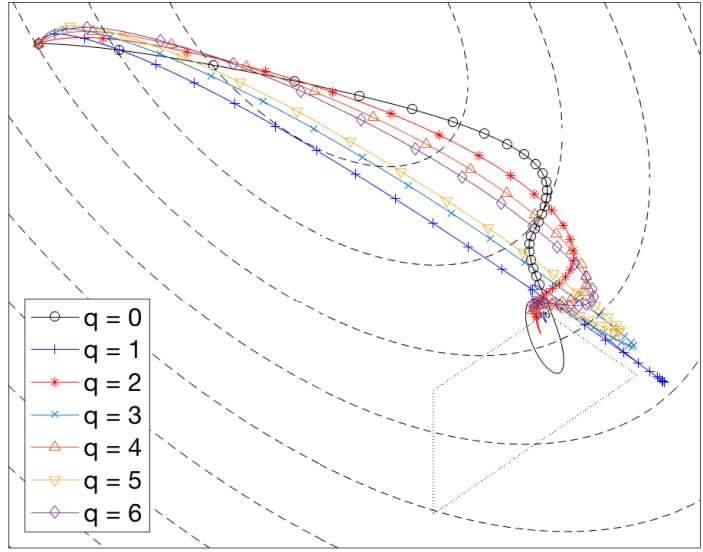}
\caption{Three node case: projection of $x^0 + Lz(t)\in\real^3$ onto the
	$2$-dimensional plane $\setdef{x}{\sum_i x_i = d}$. Markers
	plotted for $t = 0,0.2,0.4,\dots,5$ seconds. Dashed line
	ellipses indicate intersection of ellipsoid level sets with
	the plane; dotted lines indicate intersection of box
	constraints with the plane.}
\label{fig:newton-traj}
\end{figure}

\begin{figure}[h]
\centering
\includegraphics[scale = 0.6]{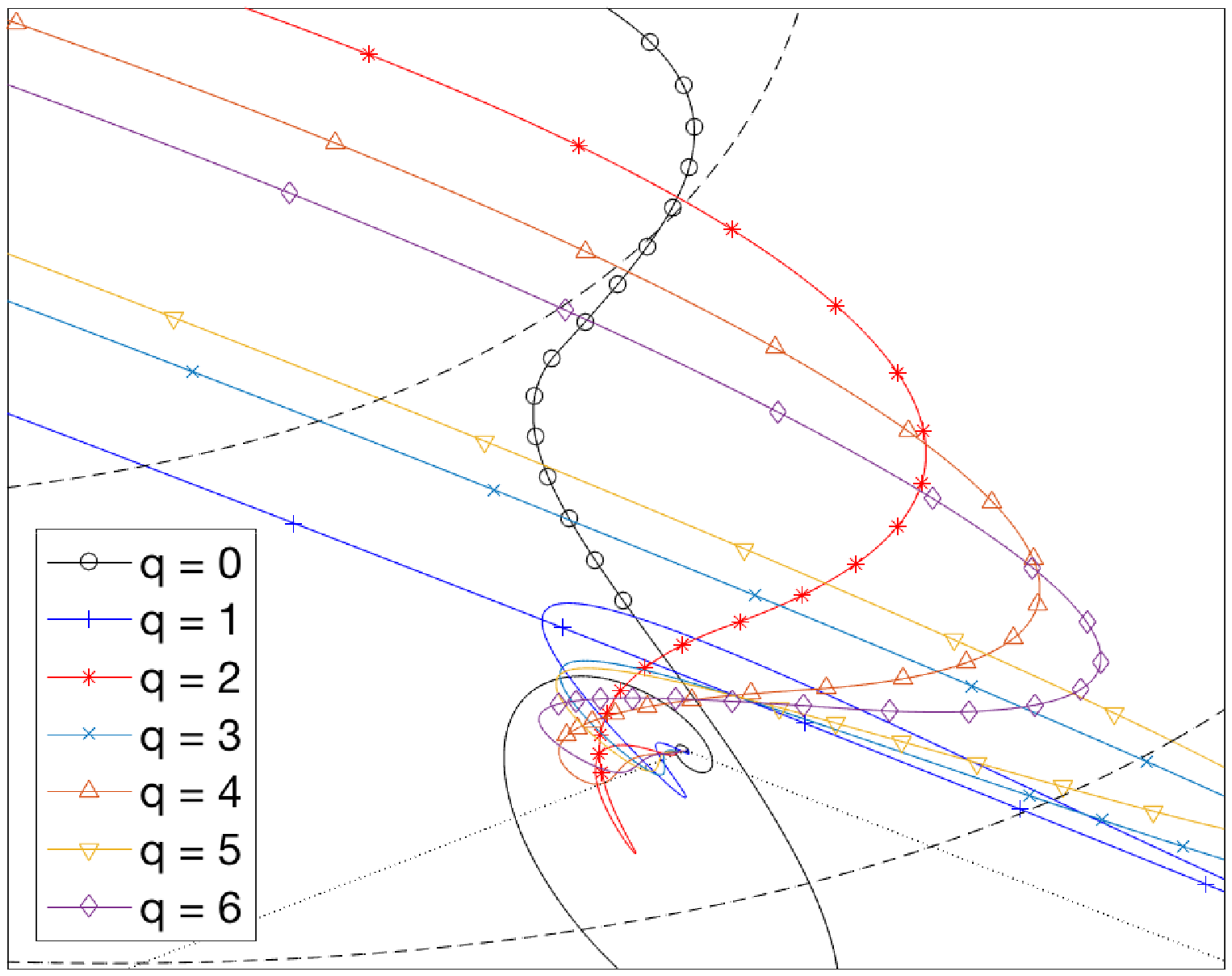}
\caption{Three node case: trajectories zoomed closer to the optimizer. Markers
	plotted in $0.2s$ increments up to $t = 5s$.}
\label{fig:newton-traj-zoomed}
\end{figure}

It is clear that choosing $q$ even versus $q$ odd has a qualitative
effect on the shape of the trajectories, as noted in
Section~\ref{ssec:interpretation}. Looking at
Figure~\ref{fig:newton-traj}, it seems the trajectories are intially
pulled toward the unconstrained optimizer (center of the level sets)
with some bias due to $\lambda(0)\neq \zeros_6$. As $\lambda$ is
given time to evolve, these trajectories are pulled back toward
satisfying the box constraints indicated by the dotted quadrilateral,
i.e. the intersection of the box constraints and the plane defined by
$\setdef{x}{\sum_i x_i = d}$.

For a quantitative comparison, we consider $n = 40$ generators with $\vert\E\vert=156$ communication links whose graph is given by Figure~\ref{fig:graph40} and
the following parameters.
\begin{equation*}
\begin{aligned}
&f_i(x_i) = \dfrac{1}{2}a_i x_i^2 + b_i x_i, \ a_i\in\mathcal{U}[0.5,3], \ b_i\in\mathcal{U}[-2,2],\\
&\underline{x}_i \in\mathcal{U}[1.5,3], \quad
\overline{x}_i \in\mathcal{U}[3,4.5], \quad i\in\until{100}, \\
&d = 120, \ x^0 = 3*\ones_{40}, \ z(0) = \zeros_{40}, \ \lambda(0) = \zeros_{80}.
\end{aligned}
\end{equation*}
\begin{figure}[h]
\centering
\includegraphics[scale = 0.6]{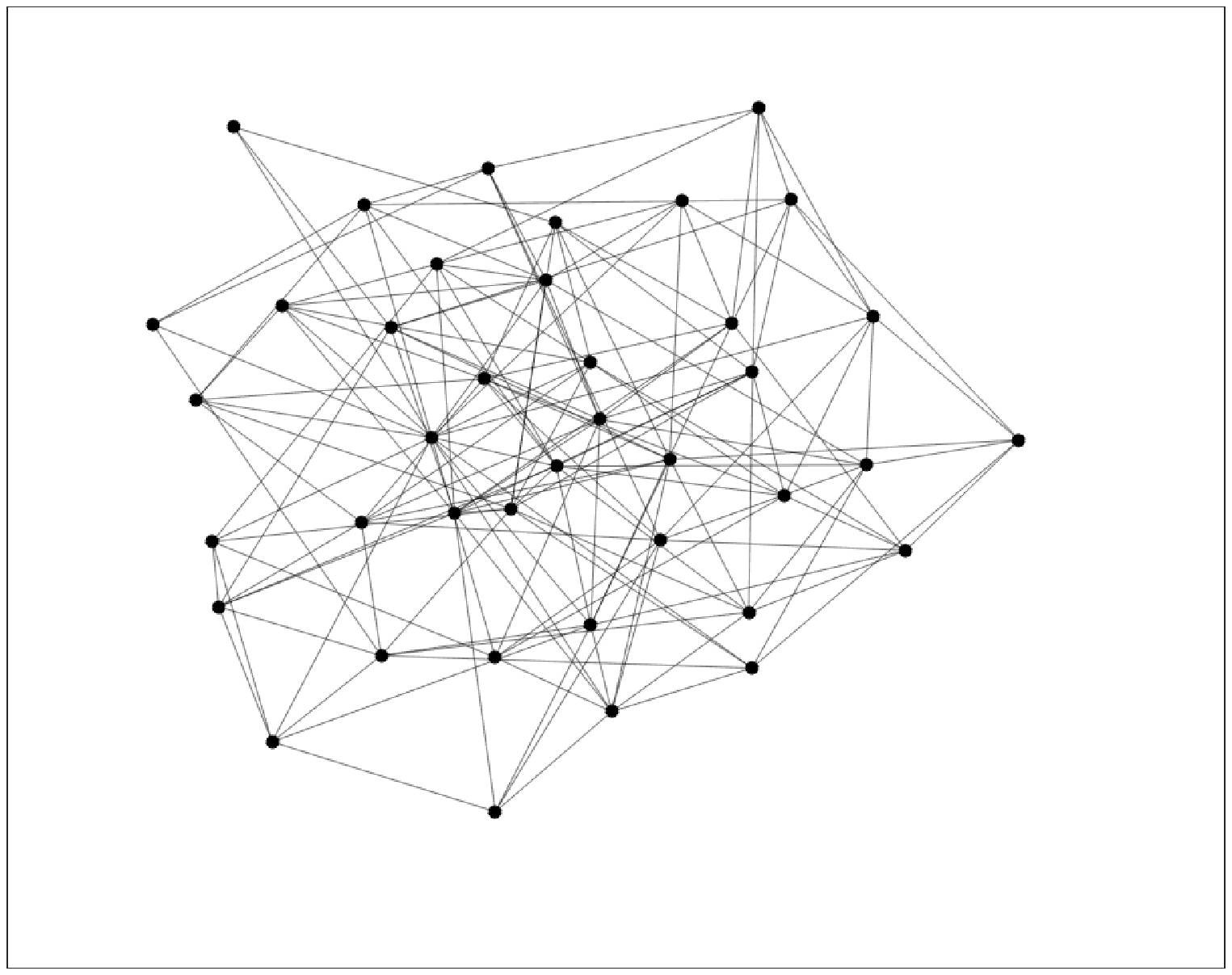}
\caption{Communication graph for continuous-time numerical study: 40 nodes and 156 edges.}
\label{fig:graph40}
\end{figure}
\begin{figure}[h]
\centering
\includegraphics[scale = 0.6]{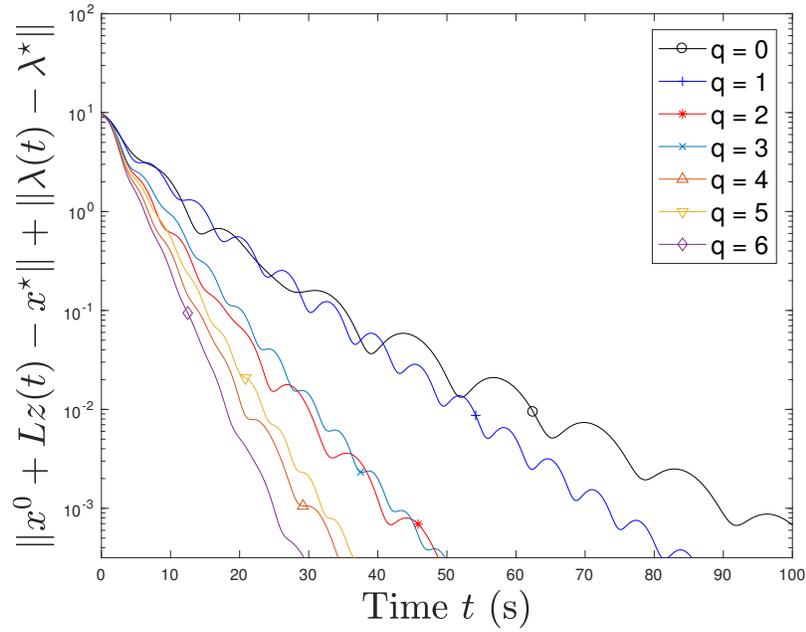}
\caption{Error in the primal and dual state variables versus time for various $q$; $n=40,\vert\E\vert=156$.}
\label{fig:errstate_cont}
\end{figure}
\begin{figure}[h]
\centering
\includegraphics[scale = 0.6]{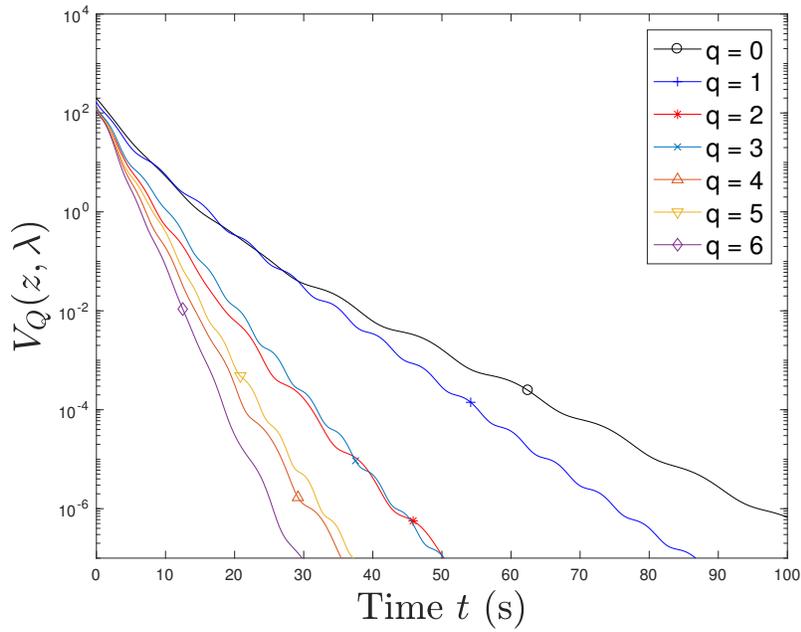}
\caption{Value of the Lyapunov function $V_Q$ versus time for various $q$; $n=40,\vert\E\vert=156$.}
\label{fig:Lyap_cont}
\end{figure}
\begin{figure}[h]
\centering
\includegraphics[scale = 0.6]{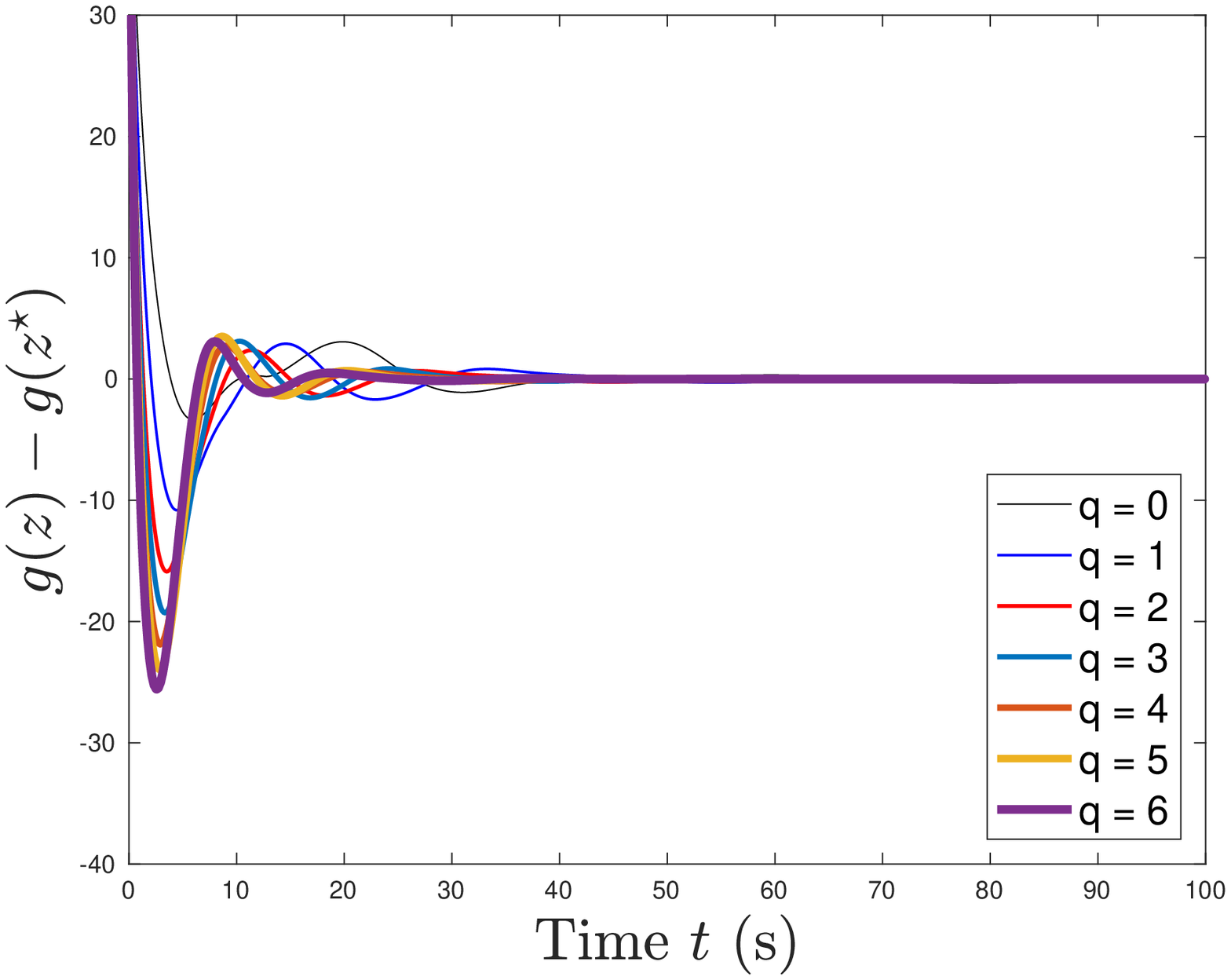}
\caption{Value of the objective function versus time for various $q$; $n=40,\vert\E\vert=156$.}
\label{fig:errfn_cont}
\end{figure}

Note from Figure~\ref{fig:errstate_cont} that convergence with respect
to $\Vert x^0 + Lz(t) - x^\star\Vert + \Vert\lambda(t) -
\lambda^\star\Vert$ is not monotonic for some $q$. This is resolved in
Figure~\ref{fig:Lyap_cont} by examining $V_Q$ as defined
by~\eqref{eq:vq}. We also note the phenomenon of faster convergence
for even $q$ over odd $q+1$; the reason for this is related to the
modes of $I_n - LHL$ and was discussed in
Section~\ref{ssec:interpretation}. However, increasing $q$ on a whole
lends itself to superior convergence compared to smaller $q$. As for
the metric $g(z) - g(z^\star)$ in Figure~\ref{fig:errfn_cont}, note
that these values become significantly negative before eventually
stabilizing around zero. The reason for this is simple: in order for
the $\mathscr{Z}_q$ dynamics~\eqref{eq:approx-newt-dynamics} in
$\lambda$ to ``activate," the primal variable must become infeasible
with respect to the box constraints. In this sense, the stabilization
to zero of the plots in Figure~\ref{fig:errfn_cont} represents the
trajectories converging to feasible points of $\Pc 2$.

\subsection{Robust DANA Implementation}\label{ssec:sims-robust-dana}
Lastly, we provide a simulation justification for relaxing
Assumption~\ref{ass:initial} via the method described in
Remark~\ref{rem:robust}. Figure~\ref{fig:robust-state-err} plots the
error in the primal and dual states over time of the modified
``robust" method, which tends to approach zero for all observed values
of $q$, and Figure~\ref{fig:robust-sum-err} demonstrates that the
violation of the equality constraint stablizes to zero very
quickly. Noisy state perturbations are injected at $t=25,50,75$, and
we observe a rapid re-approach to the plane satisfying the equality
constraint. However, even though the algorithm presents a faster
convergence than gradient methods, here do not observe as clear of a
relationship between performance and increased $q$ as in previous
settings. The investigation of the properties of this algorithm is
left as future work.

\begin{figure}[h]
\centering
\includegraphics[scale = 0.45]{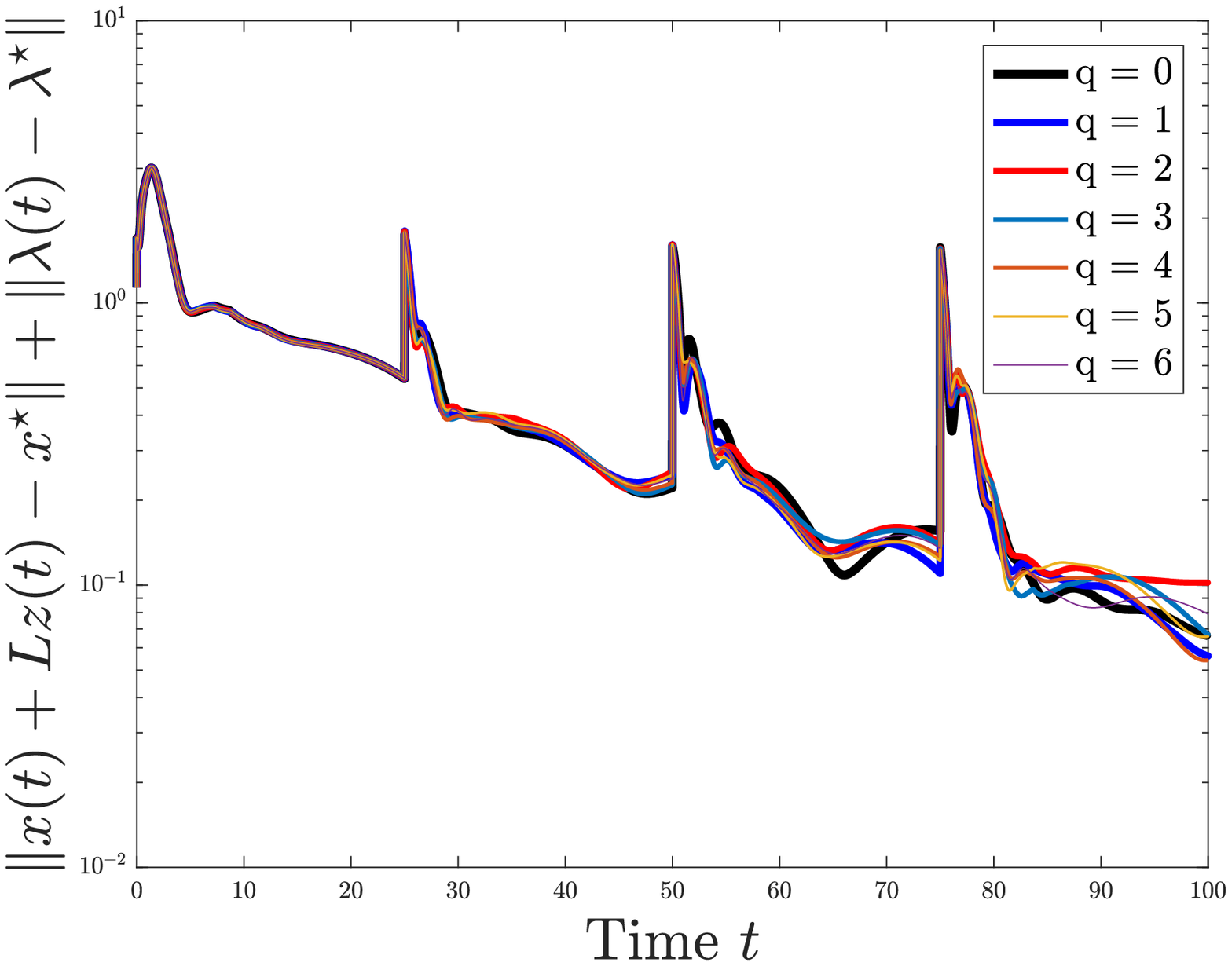}
\caption{ Error in the primal and dual states for a robust implementation of DANA; $n=20,\vert\E\vert=40$. Initialization does not satisfy Assumption~\ref{ass:initial}, and perturbations are injected at $t=25,50,75$.}
\label{fig:robust-state-err}
\end{figure}
\begin{figure}[h]
\centering
\includegraphics[scale = 0.45]{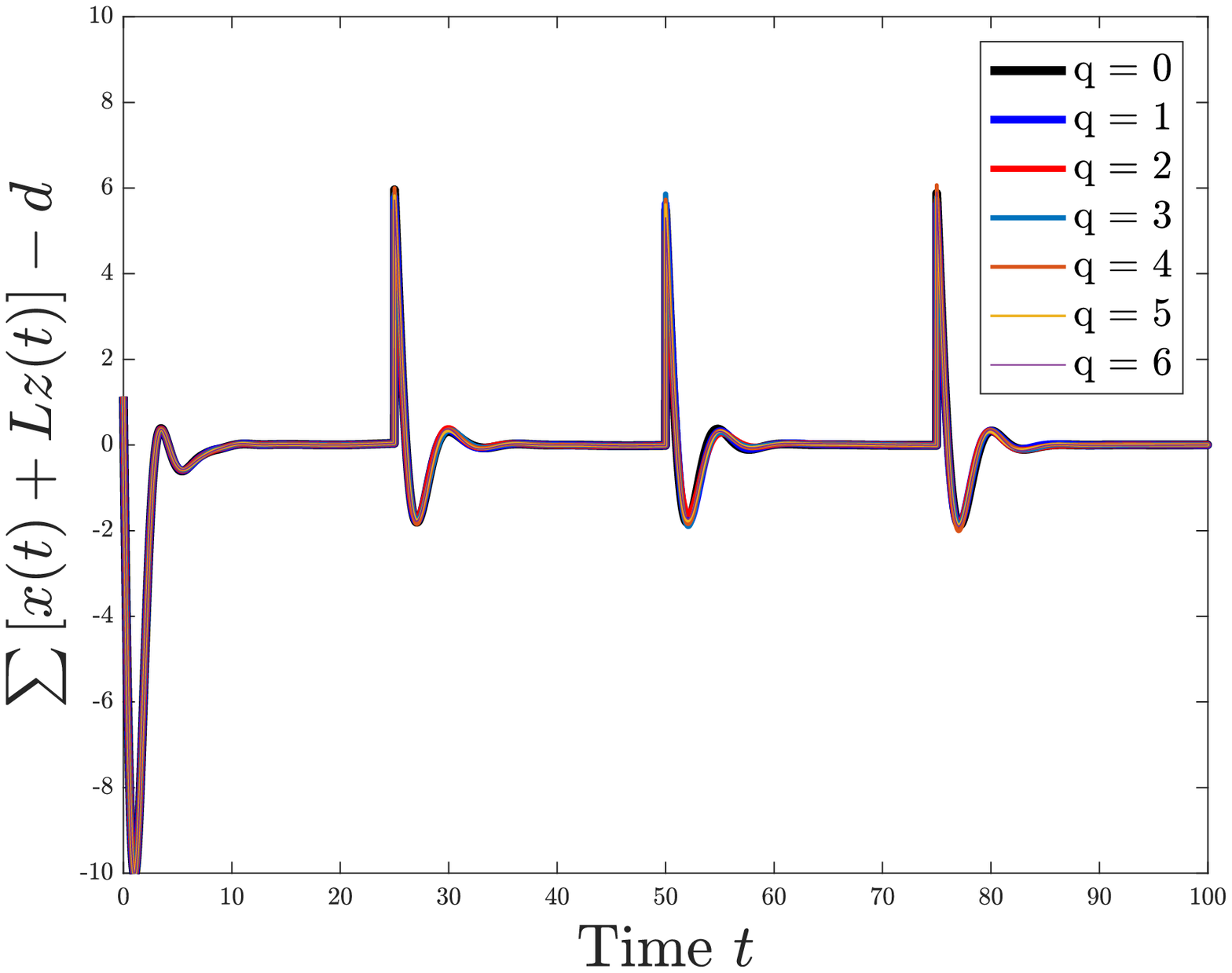}
\caption{ Violation of the resource constraint over time for robust DANA; $n=20,\vert\E\vert=40$. Perturbations are injected at $t=25,50,75$.}
\label{fig:robust-sum-err}
\end{figure}

\section*{Acknowledgements}

The material in this chapter, in full, is a reprint of \textit{Distributed Approximate Newton Algorithms and Weight Design for Constrained Optimization}, T.~Anderson, C.Y.~Chang and S.~Mart\'{i}nez, Automatica, 109, article 108538, November 2019. A preliminary version of the work appeared in the proceedings of the Conference on Control Technology and Applications (CCTA), Mauna Lani, HI, 2017, pp. 632-637, as \textit{Weight Design of Distributed Approximate Newton Algorithms for Constrained Optimization}, T.~Anderson, C.Y.~Chang and S.~Mart\'{i}nez. The dissertation author was the primary investigator and author of these papers.
\chapter{Distributed Stochastic Nested Optimization via Cubic Regularization}
\label{chap:DiSCRN}
This chapter considers a nested stochastic distributed optimization
problem. In it, approximate solutions to realizations of the
inner-problem are leveraged to obtain a Distributed Stochastic Cubic
Regularized Newton (DiSCRN) update to the decision variable of the
outer problem. We provide an example involving electric vehicle
users with various preferences which demonstrates that this model is
appropriate and sufficiently complex for a variety of data-driven
multi-agent settings, in contrast to non-nested models. The main two
contributions of the chapter are: (i) development of local stopping
criterion for solving the inner optimization problem which
guarantees sufficient accuracy for the outer-problem update, and
(ii) development of the novel DiSCRN algorithm for solving the
outer-problem and a theoretical justification of its
efficacy. Simulations demonstrate that this approach is more stable
and converges faster than standard gradient and Newton outer-problem
updates in a highly nonconvex scenario.

\section{Bibliographical Comments}

One of the most 
widely used 
stochastic optimization method is 
stochastic gradient-based (first-order) methods,
see~\cite{WAG:84,LB:10,LB-FEC-JN:18} as broad references. These
methods are powerful because they necessitate only a small sampling of
the data set 
to compute an update direction at each iterate. However, these
first-order algorithms suffer from slow convergence around
saddle-points~\cite{SD-CJ-JL-MJ-BP-AS:17}, which are
disproportionately more present in higher-dimensional nonconvex
problems~\cite{YD-RP-CG-KC-SG-YB:14}. By contrast, higher-order
Newton-based methods tend to perform more strongly across applications
in terms of number of calls to an oracle or total iterations,
see~\cite{FY-AN-US:16,XW-SM-DG-WL:17} for examples in stochastic
non-strongly convex and nonconvex settings, respectively,
and~\cite{
	AM-QL-AR:17,TA-CYC-SM:18-auto,RT-HBA-AJ:19} for various multi-agent
examples.

An issue with many of the aforementioned algorithms is they are
vulnerable to slow convergence or instability in the presence of
saddle-points and/or an ill-conditioned Hessian matrix. 
A
growing body of works thus focuses on using a cubic-regularization
term in the second-order Taylor approximation of the objective
function. 
Nesterov and Polyak laid significant groundwork for this method
in~\cite{YN-BTP:06}, and substantial follow-ups are contained
in~\cite{CC-NG-PT:09P1,CC-NG-PT:10P2}, which study adaptive batch
sizes and the effect of inexactness in the cubic submodel on
convergence. Excitement about this topic has grown substantially in
the last few years, with~\cite{YC-JD:19} showing how the global
optimizer of the nonconvex cubic submodel can be obtained under
certain initializations of gradient descent,
and~\cite{NT-MS-CJ-JR-MJ:18} being one of the first thorough analyses
of the algorithm in the traditional stochastic optimization
setting. In~\cite{XC-BJ-TL-SZ:18}, the authors consider the stochastic
setting from an adaptive batch-size perspective and~\cite{CU-AJ:20}
is, to our knowledge, the only existing work in a distributed
application, with an alternative approach that allows for a
communication complexity analysis. 
Both~\cite{XC-BJ-TL-SZ:18} and~\cite{CU-AJ:20} assume convexity,
and~\cite{CU-AJ:20} is nonstochastic. As far as we know, no current
work has unified \emph{distributed}, \emph{stochastic}, and
\emph{nonconvex} elements, particularly in a nested optimization
scenario.

\section*{Statement of Contributions} We begin the chapter by formulating
a nested distributed stochastic optimization problem, where
approximate solutions to realizations of the inner-problem are needed
to obtain iterative updates to the outer problem, and we motivate this
model with an example based on electric vehicle charging
preferences. The contributions of this chapter are then twofold. First,
we develop a stopping criterion for a Laplacian-gradient subsolver of
the inner-problem. The stopping criterion can be validated locally by
each agent in the network, and the relationship to solution accuracy
aids the synthesis with the outer-problem update. Second, to that end,
we formulate a distributed optimization model of the stochastic
outer problem and develop a cubic regularization of its second-order
approximation. This formulation lends itself to obtaining a
Distributed Stochastic Cubic-Regularized Newton (DiSCRN) algorithm, and
we provide theoretical justification of its convergence.

\section{Problem Formulation}\label{sec:prob-form}

This section details the two problem formulations which are of
interest, where the first problem $\Pc 1$ takes the form of a
stochastic approximation whose cost is a parameterization of the cost
of the second problem $\Pc 2$. Problem $\Pc 2$ is a separable
resource allocation problem in which $n$ agents $i\in \N$ must
collectively obtain a solution that satisfies a linear equality
constraint while minimizing the sum of their local costs. (This
problem commonly appears in real-time optimal dispatch for electric
grids with flexible loads and distributed generators, see
e.g.~\cite{CAISO-BPM:18}.) Thus, $\Pc 1$ can be treated as a nested
optimization, with an objective $F$ that takes stochastic arguments,
and is not necessarily available in closed form if $\Pc 2$ cannot be
solved directly and/or the distribution $\D$ being unknown. These
problems are stated as
\begin{equation*} 
\Pc 1: \
\underset{x\in\real^d}{\text{min}} \
F(x) = \Ex_{\chi\sim\D} \left[F_\chi (x)\right].
\end{equation*}
\vspace{-.25cm}
\begin{equation*} 
\begin{aligned}
\Pc 2: \
\underset{p\in\real^n}{\text{min}} \
f(x,p) &= \sum_{i=1}^n f_i(x,p_i), \\
\text{subject to} \ \sum_{i=1}^n p_i &= \Pref + \hat{\chi} = \Pref + \sum_{i=1}^n \hat{\chi}_i.
\end{aligned}
\end{equation*}
In $\Pc 1$, each $f_i:\real^d \times \real \rightarrow \real$, and $F_\chi(x) \equiv f(x,p^\star)$, where $p^\star$ is the
solution to $\Pc 2$ for particular realizations $\hat{\chi}_i,$ where $\chi_i\sim\D_i$, i.e. $\chi\sim\D = \D_1 \times \dots \times \D_n$. The elements
$p_i\in\real$ of $p\in\real^n$ and terms $\hat{\chi}_i$ are each associated
with and locally known by agents $i\in\N$, and $\Pref\in\real$ is a
given constant known by a subset of agents (we discuss its
interpretation shortly with an example). First, for $F_\chi$ to be
well defined, it helps if solutions $p^\star$ to $\Pc 2$ are unique for
fixed $x$ and $\hat{\chi}$, which we now justify with convexity
assumptions for
$f_i$. 
\begin{assump}\longthmtitle{Function Properties: Inner-Problem Argument}\label{assump:fun-inner}
	The local cost functions $f_i$ are twice differentiable and
	$\omega_i$-strongly convex in $p_i$ for any fixed $x$. Further, the
	second derivatives are lower and upper bounded:
	\begin{equation*} 0 < \omega_i\leq \nabla^2_{p_i}f_i(x,p_i) \leq
		\theta_i, \qquad \forall x\in\real^d, p_i\in\real, \text{ and } i\in\N.
	\end{equation*}
	This implies $\forall x\in\real^d, p_i, \hat{p}_i\in\real
	\text{ and } i\in\N$:
	\begin{equation*}\omega_i \|p_i - \hat{p}_i \| \leq \|
	\nabla_{p_i} f_i(x,p_i) - \nabla_{p_i} f_i(x,\hat{p}_i) \|
	\leq \theta_i \| p_i -
	\hat{p}_i\|. 
	\end{equation*}
	We also use the shorthands $\omega \triangleq \min_i{\omega_i}$ and $\theta \triangleq \max_i{\theta_i}$.
\end{assump}
This assumption will be required of our analysis in
Section~\ref{ssec:inner-loop}.
We now state some additional assumptions.
\begin{assump}
	\longthmtitle{Function Properties: Lipschitz Outer-Problem
		Argument}\label{assump:fun-outer-lips} The functions $f_i$ have
	$l_i$-Lipschitz gradients and $\rho_i$-Lipschitz Hessians:
	\begin{equation*}
		\begin{aligned} &\| \nabla f_i(x,p_i) - \nabla f_i(y,p_i) \| \leq l_i \| x - y\|, &&\forall x, y\in\real^d, \forall p_i\in\real, \\ &\| \nabla^2 f_i(x,p_i) - \nabla^2 f_i(y,p_i) \| \leq \rho_i \| x - y\|, &&\forall x, y\in\real^d, \forall p_i\in\real.
		\end{aligned}
	\end{equation*}
	We also use the shorthands $l\triangleq \max_i l_i$ and $\rho\triangleq \max_i \rho_i$.
\end{assump}
\begin{assump}\longthmtitle{Function Properties: 
		Bounded Variance Outer-Problem Argument}\label{assump:fun-outer-var}
	The function $F_\chi $ possesses the following bounded variance
	properties: 
	\begin{equation*}
	\begin{aligned} &\Ex\left[\|\nabla F_\chi (x) - \nabla F (x)\|^2 \right] \leq \sigma^2_1, 
	\quad &&\Ex\left[\|\nabla^2 F_\chi (x) - \nabla^2 F (x)\|^2 \right] \leq \sigma^2_2, \\
	&\|\nabla F_\chi (x) - \nabla F (x)\|^2 \leq M_1 \text{ almost surely}, \quad &&\|\nabla^2 F_\chi (x) - \nabla^2 F (x)\|^2 \leq M_2 \text{ almost surely}.
	\end{aligned}
	\end{equation*}
\end{assump}
\begin{assump}\longthmtitle{Function 
		Properties: Lipschitz Interconnection of Variables}\label{assump:inter-lipsch}
	The gradient and Hessian of the function $f$ with respect to $x$ are
	Lipschitz in $p$; that is, there exists constants $\psi_g, \psi_H > 0$
	such that
	\begin{equation*}
		\begin{aligned}
		&\| \nabla_{x} f(x,p) - \nabla_{x} f(x,\hat{p}) \| \leq \psi_g \| p - \hat{p}\|, \\
		&\| \nabla^2_{xx} f(x,p) - \nabla^2_{xx} f(x,\hat{p}) \| \leq \psi_H \| p - \hat{p}\|, \\ 
		&\qquad \forall x\in\real^d, p, \hat{p}\in\real^n.
		\end{aligned}
	\end{equation*}	
\end{assump}
Assumption~\ref{assump:fun-inner} is relatively common in the convex
optimization literature, and it lends itself to obtaining approximate
solutions to $\Pc 2$ very quickly with stopping criterion
guarantees. Assumption~\ref{assump:fun-outer-lips} is unanimously
leveraged in literature on Cubic-Regularized Newton methods, as the
constant $\rho$ pertains directly to the cubic submodel, while
Assumption~\ref{assump:fun-outer-var} is a common assumption in the
stochastic optimization literature~\cite{NT-MS-CJ-JR-MJ:18}. We note
that, although Assumptions~\ref{assump:fun-outer-lips}
and~\ref{assump:fun-outer-var} do not give a direct relationship with
the local functions $f_i(x,p_i)$, they do imply an implicit
relationship between $x, p, \text{ and } \D$ in the sense that
solutions $p^\star$ to $\Pc 2$ (and therefore the distributions $\D_i$)
must be ``well-behaved" in some sense. 
This relationship, along with a broader interpretation of the model
$\Pc 1$ and $\Pc 2$, is illustrated more concretely in the following
real-world power distribution example.


\begin{example}\longthmtitle{EV Drivers with PV Generators}\label{example:model}
	Consider two EV drivers who each have an EV charging station and a
	PV generator. 
	The goal of this small grid system is to consume net zero power from
	the perspective of the tie line to the bulk grid, thus $\Pref =
	0$. The distributions $\D_1,\D_2$ represent the power output
	distributions of the PVs, and we consider two scenarios for these in
	this example: (1) a ``sunny day" scenario, where the realizations
	$\chi_1,\chi_2\sim\D_1,\D_2$ of PVs 1 and 2 are deterministic, and
	(2) a ``cloudy day" scenario, where intermittent cloud cover induces
	some uncertainty in the moment-to-moment PV generation.
	
	Let $A\in\{\text{sunny},\text{cloudy}\}$ indicate the weather forecast. The model is then fully described as
	\begin{equation*}
	\D_i = \begin{cases}
	\delta_{1.5}, & A = \text{sunny}, \\
	\U[0,1.5], & A = \text{cloudy}
	\end{cases} \quad \text{for both i=1,2},
	\end{equation*}
	\begin{equation*}
	f_1(x,p_1) = (2x + p_1 - 1)^2, \;
	f_2(x,p_2) = (x + p_2 - 2)^2.
	\end{equation*}
	
	For $x=0$, these quadratic functions\footnote{See~\cite{SB-HF:13} for an example where quadratic costs to EV users are induced by resistive energy losses in the battery model and~\cite{AW-BW-GS:12} for a broad reference on modeling generator dispatch.} have local minima at
	$p^\star = (p_1^\star,p_2^\star ) = (1,2)$, which is
	interpreted as drivers 1 and 2 preferring to charge at rates
	of 1 unit and 2 units, respectively, if there are no external
	incentives. On a sunny day, both PVs 
	deterministically produce $\hat{\chi}_1,\hat{\chi}_2 = 1.5$, which
	effectively balances the unconstrained $p^\star$ and both
	drivers can charge at their preference to maintain $\sum_i p_i
	= \Pref + \sum_i \hat{\chi}_i$.
	
	However, on cloudy days the generation of the PVs is no longer deterministic. Thus, the variable $x$
	comes in to play, which can represent a government credit that
	the drivers value differently. The role of $x$ is to shift the cost
	functions such that the unconstrained minima are near lower
	charging values in consideration of the lower expected
	generation from PVs 1 and 2. The optimal $x^\star$ to $\Pc 1$
	is the value which gives the lowest expected cost of an
	instance of $\Pc 2$ given $\hat{\chi}_1,\hat{\chi}_2$ realizations from the
	$A = $ cloudy distributions $\D_1, \D_2$. A more complete model of $\Pc 2$ could include power flow constraints; in this work, we relax these for simplicity.
	
\end{example}

\section{Distributed Formulation and Algorithm}

In this section, we develop the inner-loop
algorithm used to solve $\Pc 2$. We then synthesize inexact solutions to $\Pc 2$ with the DiSCRN algorithm for $\Pc 1$.

\subsection{Inner Loop Gradient Solver}\label{ssec:inner-loop}

For this section, consider $x$ to be fixed and known by all
agents. Further, let $\hat{\chi}_i$ be fixed (presumably from a realization
of $\D_i$) and known only to agent $i$. We adopt the following
assumption on the initial condition $p^0$.

\begin{assump}\longthmtitle{Feasibility of Inner-Problem Initial Condition}\label{assump:init-cond-p2}
	The agents are endowed with an initial condition which is feasible
	with respect to the constraint of $\Pc 2$; that is, they each possess
	elements $p_i^0$ of a $p^0$ satisfying $\ones_n^\top p^0 = \Pref + \hat{\chi}.$
\end{assump}

The assumption is easily satisfied in practice by communicating
$\Pref$ to one agent $i$ and setting $p_i^0 = \Pref + \hat{\chi}_i$, with
all other agents $j$ using $p_j^0 = \hat{\chi}_j$. 
An alternative to this assumption consists of reformulating
$\Pc 2$ with distributed constraints and using a dynamic consensus
algorithm as in~\cite{AC-EM-SHL-JC:18-tac}, which would still retain
exponential convergence. We impose
Assumption~\ref{assump:init-cond-p2} for simplicity. Finally, we
assume connectedness of the communication graph:
\begin{assump}\longthmtitle{Graph Properties}\label{assump:graph-conn}
	The communication graph $\G$ is connected and undirected; that is, a
	path exists between any pair of nodes and, equivalently, its
	Laplacian matrix $L = L^\top \succeq 0$ has rank $n-1$ with
	eigenvalues $0 = \lambda_1 < \lambda_2 \leq \dots \leq \lambda_n$.
\end{assump}

The discretized Laplacian-flow dynamics are given by:
\begin{equation}\label{eq:disc-lap-flow}
p^+ = p - \eta L\nabla_p f(x,p).
\end{equation}
Note that these dynamics are distributed, as the sparsity of $L$
implies each agent need only know $\nabla_{p_i} f_i(x,p_i)$ and
$\nabla_{p_j} f(x,p_j)$ for $j\in\N_i$ to compute $p_i^+$. We now
justify convergence of~\eqref{eq:disc-lap-flow} to the solution
$p^\star$ of $\Pc 2$:
\begin{prop}\longthmtitle{Convergence of Discretized Laplacian Flow}\label{prop:grad-cvg}
	Let $p^\star \in \real^n$ be the unique minimizer of $\Pc 2$. Given
	Assumption~\ref{assump:init-cond-p2} on the feasibility of the
	initial condition, Assumption~\ref{assump:graph-conn} on
	connectivity of the communication graph, and
	Assumption~\ref{assump:fun-inner} on the Lipschitz gradient
	condition of the function gradients, then, under the
	dynamics~\eqref{eq:disc-lap-flow} with $0 < \eta < \frac{2}{\theta
		\lambda_n}$, $p$ converges asymptotically to $p^\star$.
\end{prop}
\begin{proof}
	Using a standard quadratic expansion around the current iterate $p$
	(see e.g. $\S 9.3$ of~\cite{SB-LV:04}) and Lipschitz bounds yields $f(p^+) - f(p) \leq
	\theta\eta^2 / 2 \| L \nabla_p f(x,p)\|^2 \\ - \eta \nabla_p
	f(x,p)^\top L \nabla_p f(x,p)$. Careful treatment of the eigenspace of $L$ and some algebraic manipulation shows that $f(p^+) - f(p)$ is strictly negative for $\eta$ as in the statement. A more detailed proof can be found in~\cite{TA-SM:20-extended}.
\end{proof}

We now provide an additional result on exponential convergence of the state
error with a further-constrained step size as compared to the
statement in Proposition~\ref{prop:grad-cvg}.
\begin{prop}\longthmtitle{Exponential Convergence with Bounded Error}\label{prop:exp-cvg}
	Let Assumptions~\ref{assump:init-cond-p2},~\ref{assump:graph-conn},
	and~\ref{assump:fun-inner} hold as before. For $0 < \eta <
	2\omega\lambda_2/\theta^2\lambda_n^2$, the quantity $\|p - p^\star
	\|$ converges exponentially to zero under the
	dynamics~\eqref{eq:disc-lap-flow}. For $\eta = \omega\lambda_2 / \theta^2 \lambda_n^2$, the rate is $\| p^+ - p^\star \| \leq
	\sqrt{1-\omega^2 \lambda_2^2 / \lambda_n^2 \theta^2}\|p-p^\star\|$,
	and $\|p^K-p^\star\| \leq \Delta$ for \\ $K \geq \log(\Delta/\|p^0 -
	p^\star\|)/\log(\sqrt{1-\omega^2 \lambda_2^2 / \lambda_n^2
		\theta^2})$. 
\end{prop}
\begin{proof}
	Consider $V(p) =
	\|p-p^\star\|^2$. Substituting~\eqref{eq:disc-lap-flow} and applying
	bounds via eigenvalues of $L$, using
	Assumption~\ref{assump:fun-inner}, and $\nu$-strong function
	convexity, we get $V(p^+)\le (\eta^2\lambda_n^2\theta^2 - 2\eta
	\lambda_2 \omega + 1)V(p)$, with $0 < \eta <
	2\lambda_2\omega/\lambda_n^2\theta^2$. The choice of $\eta =
	\omega\lambda_2 / \lambda_n^2 \theta^2$ implies the exponential
	convergence as in the statement. See~\cite{TA-SM:20-extended}
	for more details.
\end{proof}

We note that the results of
Propositions~\ref{prop:grad-cvg} and~\ref{prop:exp-cvg} simply build on a
Laplacian-projected version of vanilla gradient descent. However, it
lays some basic groundwork and supplements our main results in the
next subsection.

With this, we are ready to transition to the discussion on obtaining a DiSCRN update to $\Pc 1$.

\subsection{Outer-Loop Cubic-Newton Update}

We endow each agent with a local copy $x_i$ of the  variable
$x$, and we let $\x\in\real^{nd}$ be the stacked vector of these local
copies. 
Thus, a distributed reformulation of $\Pc 1$ is
\begin{equation*}
\begin{aligned}
\overline{\Pc 1}: \
\underset{\x\in\real^{nd}}{\text{min}} \quad
& \bar{F}(\x) = \Ex_{\chi\sim\D} \left[\bar{F}_{\chi} (\x)\right],  \\
\text{subject to} \quad & (L\otimes I_d)\x = \zeros_{nd},
\end{aligned}
\end{equation*}
where $\bar{F}_{\chi}:\real^{nd}\rightarrow\real$ is analagous to
$F_{\chi}:\real^{d}\rightarrow\real$ in the sense that each agent evaluates $f_i(x_i,p^\star_i)$ with its local copy of
$x_i$. 
Note that the constraint $(L\otimes I_d)\x = \zeros_{nd}$ imposes $x_i
= x_j, \forall i,j$ (Assumption~\ref{assump:graph-conn}), so
$\bar{F}_{\chi}$ and $F_{\chi}$ are equivalent in the agreement
subspace (and $\overline{\Pc 1}$ is equivalent to $\Pc 1$). 
Since our problem is nested and stochastic, there is a lack of access to
a closed form expression for $\bar{F}$ and 
$\bar{F}_{\chi}$. Thus, we introduce an
\emph{empirical-risk}, \emph{approximate} objective function.
To this end, let $F^S(\x) = 1/S \sum_{s=1}^S F^\Delta_{\chi^s}(\x)$ 
be the approximation of $\bar{F}$ for $S$ samples of
$\chi^s\sim\D$, where $F^\Delta_{\chi^s}\equiv \sum
f_i(x_i,\tilde{p}^s_i)$ and $\|\tilde{p}^s - p^\star\| \leq \Delta$ for
realization $\chi^s$. In this sense, $F^\Delta_{\chi^s}$ implicitly
depends on $\tilde{p}^s$, and the $\Delta$ superscript is a slight abuse
of notation. For now, the reader can consider $\Delta$ to be a sufficiently small design parameter describing the inexactness of the obtained solutions to $\Pc 2$; we build on this later. Ultimately, we intend to use batches of $F^S$ rather than the exact
$\overline{F}$ to implement DiSCRN. Consider then the cubic regularized submodel of $F^S$ at some $\x^k$:
\begin{equation}\label{eq:cub-sub-primal}
m_S^k (\x) = F^S(\x^k) +
(\x-\x^k)^\top g^k + \frac{1}{2}
(\x-\x^k)^\top H^k (\x-\x^k) + \sum_{i=1}^n \frac{\rho_i}{6}\| x_i - x_i^k\|^3,
\end{equation}
where $g^k = \nabla F^S (\x^k), H^k = \nabla^2 F^S (\x^k)$. Note that there is a slight difference between~\eqref{eq:cub-sub-primal} and
the more standard cubic submodel~\eqref{eq:to-model} in that the
regularization terms are \emph{directly separable}; this is crucial
for a distributed implementation, and our forthcoming analysis justifies that convergence can still be established.
We are interested in finding $\x^+$ which minimizes~\eqref{eq:cub-sub-primal} in the agreement subspace:
\begin{equation}\label{eq:cub-sub-constr}
\Pc 3: \ \underset{\x\in\real^{nd}}{\min} \quad m_S^k(\x), \quad \text{subject to } (L \otimes I_d)\x = \zeros_{nd}.
\end{equation}
Therefore, we prescribe the Decentralized Gradient Descent dynamics from~\cite{JZ-WY:18}:
\begin{equation}\label{eq:saddle-submodel}
\x^{+,t+1}
= 
W\x^{+,t} -\alpha_t \nabla_{\x}m_{S}^k(\x^{+,t}),
\end{equation}
where $W = I_{nd} - (1/\lambda_n L\otimes I_d)$ and $\alpha_t \sim 1/t$. Per Proposition 3 and Theorem 2 of~\cite{JZ-WY:18}, $\x^{+,t}$ under the dynamics~\eqref{eq:cub-sub-primal} converges asymptotically to a stationary point of $\Pc 3$ with $O(1/k)$ convergence in the agreement subspace, i.e. $\| x_i-\overline{x} \|$ approaches zero at a rate $O(1/k)$, where $\overline{x} = \mean{(x_i)}$.

We remark that one could formulate the Lagrangian of $\Pc 3$ and use a saddle-point method to obtain a useful update $\x^+$. This is more parallel to the work of~\cite{YC-JD:19}, which achieves the global solution via gradient descent in the centralized setting. However, even the existence of a Lagrangian saddle-point is in question when the duality gap is nonzero, so further study is required on that approach.

Our aim is to obtain an $\epsilon$-second-order stationary point of $\Pc 1$,
as in Definition~\ref{def:so-stat-point}. The above discussion serves to set up the following condition on $\x^{k+1}$:
\begin{cond}\longthmtitle{Subsolver Output}\label{cond:subsolver}
	Let
	$\x^{k+1}$ be the output of a subsolver for $\Pc 3$. Then,
	\begin{enumerate}[(i)]
		\item $\x^{k+1}$ satisfies $(L\otimes I_d)\x^{k+1} = \zeros_{nd}$.\label{item:cond-i}
		\item For an arbitrarily small constant $c>0$ and some $\epsilon > 0$, $\x^{k+1}$ satisfies $m_S^k (\x^{k+1}) - m_S^k(\x^k) < -c\epsilon\|\x^{k+1} - \x^k \| - c\sqrt{\rho\epsilon} \|\x^{k+1} - \x^k \|^2$. \label{item:cond-ii}
	\end{enumerate}
\end{cond}
Part~(\ref{item:cond-i}) is implied in a linear convergence sense by the result of~\cite{JZ-WY:18} for the subsolver~\eqref{eq:saddle-submodel}. 
The~(\ref{item:cond-ii}) condition is straightforwardly implied by any subsolver that is guaranteed to strictly decrease $m_S^k$, e.g.~\eqref{eq:saddle-submodel}, because $c$ can be taken arbitrarily small. However, it can be seen in the statement of Theorem~\ref{thm:discrn} that small $c$ implies a direct tradeoff with $\Delta$ (becomes small) and/or $S$ (becomes large).

We now give a brief outline of the entire algorithm. \newpage
{\begin{center}\underline{DiSCRN Algorithm}\end{center}}
\begin{enumerate}
	\item Initialize $\x^0$ s.t. $(L\otimes I_d)\x^0 = \zeros_{nd}$
	\item Realize $\chi^s$ and initialize $p^0$ per Assumption~\ref{assump:init-cond-p2}\label{item:realize-init}
	\item Implement~\eqref{eq:disc-lap-flow} until $\vert p_i^+ -
	p_i\vert \leq \Delta \eta \lambda_2 \omega /\sqrt{n},
	\forall i$\label{item:stop-crit}
	\item Repeat from step~\ref{item:realize-init} $S$ times, storing $\tilde{p}^s\gets p^+$ at each $s$
	\item Compute locally required elements of $g^k,H^k$
	\item Compute an $\x^{k+1}$ satisfying
	Condition~\ref{cond:subsolver},
	e.g. via~\eqref{eq:saddle-submodel}; repeat from
	step~\ref{item:realize-init}\label{item:outer-loop}
\end{enumerate}
The DiSCRN Algorithm describes a fully distributed algorithm, as each
step can be performed with only local 
information. Ostensibly, $\x^0$ could be initialized arbitrarily, but
the first outer-loop would be a ``garbage" update until agreement is
obtained in step~\ref{item:outer-loop}. Note that Step~\ref{item:stop-crit} relates to a distributed stopping criterion for the subsolver of $\Pc 2$; this condition produces a solution $p^+$ in finite iterations which is sufficiently close to $p^\star$ for the sake of our analysis. This is detailed more in Theorem~\ref{thm:discrn} and its proof.

\begin{cond}\longthmtitle{Assumptions and Conditions for Theorem~\ref{thm:discrn}}\label{cond:thm1}
	Let $F$ satisfy Assumption~\ref{assump:fun-outer-lips}, on Lipschitz
	gradients and Hessians, and Assumption~\ref{assump:fun-outer-var},
	on variance conditions, and let $f$ satisfy
	Assumption~\ref{assump:inter-lipsch}, on Lipschitz interconnection
	of $x$ and $p$, and Assumption~\ref{assump:fun-inner}, on the Lipschitz condition of the
	function gradients with respect to $p$. Further, let Assumption~\ref{assump:init-cond-p2}, on the feasibility of the
	initial condition for $\Pc 2$, and Assumption~\ref{assump:graph-conn} on
	connectivity of the communication graph, each hold. Let $\x^{k+1}$ be the output of a subsolver for $\Pc 3$ that
	satisfies Condition~\ref{cond:subsolver} with $c$, and let $\tilde{p}^s\gets p^+$, where $p^+$ is the returned value under the dynamics~\eqref{eq:disc-lap-flow} satisfying $\vert p_i^+ -
	p_i\vert \leq \Delta \eta \lambda_2 \omega /\sqrt{n},
	\forall i$.
\end{cond}

\begin{thm}\longthmtitle{Convergence of DiSCRN}\label{thm:discrn}
	Let the circumstances of Condition~\ref{cond:thm1} apply here. For \\ $S \geq
	\max\{\frac{M_1}{\bar{c}\epsilon},\frac{\sigma_1^2}{\bar{c}^2\epsilon^2},\frac{M_2}{\bar{c}\sqrt{\rho\epsilon}},\frac{\sigma_2^2}{\bar{c}^2\rho\epsilon}\}O(\log{((\epsilon^{1.5}\zeta\bar{c})^{-1})})$
	with $\bar{c}\epsilon + \psi_g \Delta \leq c \epsilon$ and
	$\bar{c}\sqrt{\rho\epsilon} + \psi_H \Delta \leq
	c\sqrt{\rho\epsilon}$, then for all $\zeta > 0$ each $x_i$ asymptotically approaches a common $\epsilon$-second-order stationary point $\tilde{x}$ of $F$ with
	probability $\geq 1 - \zeta$ under the DiSCRN algorithm dynamics.
\end{thm}
\begin{proof}
	First, we aim to obtain the bound $\| \tilde{p}^s - p^\star\| \leq \Delta$ for each instance $s$ of $\Pc 2$. The Lipschitz condition of Assumption~\ref{assump:fun-inner} implies
	\begin{equation*}
	{\small
		\begin{aligned}
		\omega \| p - p^\star\| &\leq \| \nabla_p f(x,p) - \nabla_p f(x,p^\star)\|, \\
		\lambda_2 \omega \| p - p^\star\| &\leq \| L(\nabla_p f(x,p) - \nabla_p f(x,p^\star))\| \\
		&= \| L \nabla_p f(x,p) \| = 1/\eta \|p^+ - p\| \leq \Delta \lambda_2 \omega.
		\end{aligned}
	}
	\end{equation*}
	Finally,  $1/\sqrt{n}$  comes from breaking $p^+ - p$
	into components and since, for $v\in\real^n$, if 
	$\vert v_i \vert \leq c/\sqrt{n}$ implies $ \| v \| \leq
	c$.
	
	Turning to $\Pc 1$, let $g_\star^k = \frac{1}{S} \sum_{s=1}^S \sum_i \nabla_{x_i}
	f_i(x_i^k, p_i^\star)$ and $H_\star^k = \frac{1}{S}
	\sum_{s=1}^S \sum_i \nabla^2_{x_i x_i} f_i(x_i^k, p_i^\star)$.
	\\ Lemma~4 of~\cite{NT-MS-CJ-JR-MJ:18} justifies that for
	arbitrary $\bar{c} > 0$, choosing \\ $S \geq
	\max\{\frac{M_1}{\bar{c}\epsilon},\frac{\sigma_1^2}{\bar{c}^2\epsilon^2},\frac{M_2}{\bar{c}\sqrt{\rho\epsilon}},\frac{\sigma_2^2}{\bar{c}^2\rho\epsilon}\}O(\log{((\epsilon^{1.5}\zeta\bar{c})^{-1})})$
	implies that $\| g_\star^k - \nabla \bar{F}(\x^k) \| \leq
	\bar{c}\epsilon$ and $\|(H_\star^k - \nabla_{\x\x}^2
	\bar{F}(\x^k))v\| \leq
	\bar{c}\epsilon\sqrt{\rho\epsilon}\|v\|, \forall v$ with
	probability $1-\zeta$.
	
	Let $\phi^k_g = g^k - g_\star^k, \phi^k_H = H^k - H_\star^k$, where $g^k$ and $H^k$ use the inexact estimates $\tilde{p}^s$ satisfying $\| \tilde{p}^s - p^\star\| \leq \Delta$. Substitutions and applying Assumption~\ref{assump:inter-lipsch} gives:
	\begin{equation*}
	\begin{aligned}
	&\| g^k - \nabla_{\x} \bar{F}(\x^k) \| \leq \| g_\star^k - \nabla_{\x} \bar{F}(\x^k) \| + \| \phi^k_g\| \leq \bar{c}\epsilon + \psi_g \Delta, \\
	&\| (H^k - \nabla^2_{\x\x} \bar{F}(\x^k))v \| \leq \| (H_\star^k - \nabla^2_{\x\x} \bar{F}(\x^k))v \| + \| \phi^k_H\| \leq \bar{c}\sqrt{\rho\epsilon} + \psi_H \Delta, \forall v.
	\end{aligned}
	\end{equation*}
	Next, let $\xi^k := \x^{k+1} - \x^k$ for notational convenience. The separable cubic regularized terms of $m_S^k$ can be used to bound the true function value:
	\begin{equation*}
	\begin{aligned}
	\bar{F}(\x^{k+1}) &\leq \bar{F}(\x^{k}) + \nabla \bar{F}(\x^{k})^\top \xi^k + \xi^{k\top} \nabla^2 \bar{F} (\x^k) \xi^k + \sum_i \rho_i/6 \| x_i^{k+1} - x_i^k \|^3 \Rightarrow \\
	\bar{F}(\x^{k+1}) - \bar{F}(\x^{k}) &\leq m_S^k(\x^{k+1}) - m_S^k(\x^k) + (\nabla \bar{F} (x^k) - g^k )^\top \xi^k + 1/2 \xi^{k\top} (\nabla^2 \bar{F}(x^k) - H^k ) \xi^k \\
	& \leq m_S^k(\x^{k+1}) - m_S^k(\x^k) + (\bar{c} \epsilon  + \psi_g \Delta) \| \xi^k \| + (\bar{c} \sqrt{\rho\epsilon} + \psi_H \Delta )\| \xi^k\|^2 \\
	& \leq m_S^k(\x^{k+1}) - m_S^k(\x^k) + c\epsilon\|\xi^k\| + c\sqrt{\rho\epsilon}\|\xi^k \|^2 < 0,
	\end{aligned}
	\end{equation*}
	where the first inequality is implied by breaking up $\bar{F}_\chi (\x)$ in to its separable local functions and applying Assumption~\ref{assump:fun-outer-lips} and noting that the inequality carries through the expectation operator. Subsequent inequalities are directly obtained via substitutions. The lefthand inequality of the final line stems from the Theorem statement, and the righthand inequality of the final line from~(\ref{item:cond-ii}) of Condition~1.
\end{proof}

\section{Simulation}

We consider a synthetic nonconvex case for our simulation study. The cost functions $f_i$ can be represented as:
\begin{equation*}\label{eq:fns-p}
\begin{aligned}
f_i(x,p_i) &= \frac{1}{2}\alpha_i(x) p_i^2 + \beta_i(x) p_i + \gamma_i.
\end{aligned}
\end{equation*}
Each $\alpha_i:\real\rightarrow\real$ is quartic in $x$ and generated according to~\eqref{eq:alpha}, where each $a_i^2$ is determined such that $\min_x \ \alpha_i(x) = \omega_i > 0$ with $\omega_i\in\U[1,5]$ per Assumption~\ref{assump:fun-inner}. The $\beta_i:\real\rightarrow\real$ are (possibly nonconvex) quadratic, and $\gamma_i = 0$.
\begin{equation*}\label{eq:alpha}
\begin{aligned}
&\alpha_i(x) = a_i^1(x-z_i^1)(x-z_i^2)(x-z_i^3)(x-z_i^4) + a_i^2, \\
&a_i^1\in\U[0.5,1.5], z_i^1\in\U[-2,-1], z_i^2\in\U[-1,0], z_i^3\in\U[0,1], z_i^4\in\U[1,2], \\
&\beta_i(x) = b_i^1(x-z_i^5)(x-z_i^6), \quad b_i^1\in\U[-1,1], z_i^5\in\U[-2,0],z_i^6\in\U[0,2],
\end{aligned}
\end{equation*}

We compare our DiSCRN method with
gradient-based and Newton-based updates of the same batch sizes, where
the gradient-like and Newton-like updates are computed via:
\begin{equation*}
	\begin{aligned}
	m_{g}^k(\x) &= F^S (\x^k) +
	(\x-\x^k)^\top g^k + \sum_i \frac{\eta_g}{2} \|x_i - x_i^k\|^2, \\
	m_{H}^k(\x) &= F^S (\x^k) + (\x-\x^k)^\top g^k +\frac{1}{2} (\x-\x^k)^\top H^k (\x-\x^k)
	+ \sum_i \frac{\eta_H}{2} \|x_i - x_i^k\|^2,
	\end{aligned}
\end{equation*}
We obtain $\x^{k+1}$ empirically for all three methods by
implementing~\eqref{eq:saddle-submodel} until the updates become very
small. We found that both $\eta_g$ and $\eta_H$ must be sufficiently large
to ensure stability, and
$\nabla^2_x F(x) \succ -\eta_H I_d$ to ensure $m_H^k (x)$ bounded. We take $\Delta = 0.1, S = 20, n = 40, \vert
\E \vert = 120, \Pref = 40, \D_i = \U[0,1.5] \ \forall i, \rho = 50, 
\eta_g = 100, \eta_H = 50$.

We note substantially improved performance of DiSCRN over the more traditional gradient-based and Newton-based approaches. In particular, the trajectory finds a minimizer in roughly half and one-third the number of outer-loop iterations required by Newton and gradient, respectively. It is clear that, for $x^{k+1} \approx x^k$, the cubic regularization is less dominant than the squared regularizations, allowing the DiSCRN trajectory to be influenced more by the problem data $g^k, H^k$. As for the parameters $(\rho, \eta_g, \eta_H)$, $\eta_H = 50$ and $\eta_g = 100$ were roughly the lowest possible values without inducing instability. By contrast, reducing $\rho$ to values $\sim 10^{-1}$ was still stable for DiSCRN. We noticed a clear tradeoff between
$S$ and $\Delta$, with small $S\sim 10^0$ requiring $\Delta\sim
10^{-1}$ to converge and large $S\sim 10^3$ converging even for large
$\Delta\sim 10^2$, which is implied by Theorem~\ref{thm:discrn}. Finally, DiSCRN achieves reduced disagreement compared to gradient and Newton; this could be in part due to~\eqref{eq:saddle-submodel} finding a stationary point of $\Pc 3$ faster, allotting more iterations where the consensus terms dominate the update.

\begin{figure}
	\centering 
	\includegraphics[scale=0.75]{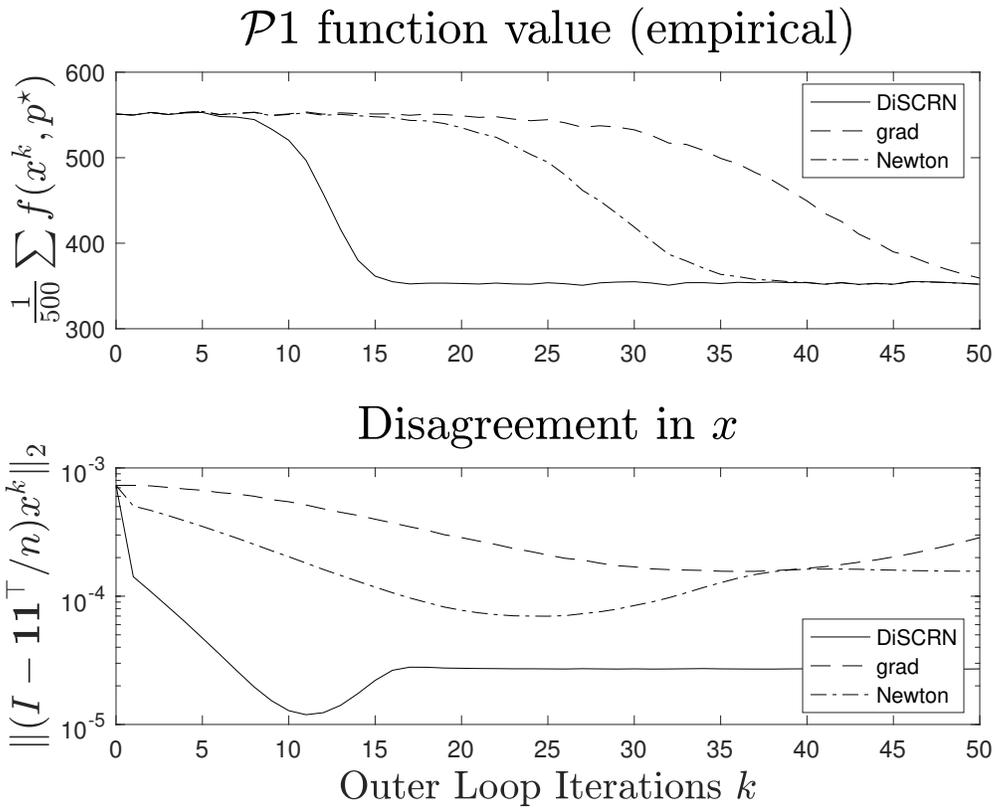}
	\caption{Comparison of CRN method with gradient-based and
		Newton-based approaches. \textbf{Top:} empirical
		approximation of $F(x^k)$, obtained by averaging
		$f(x^k,p^\star)$ over 500 realizations of $\Pc 2$ at each
		$k$. 
		\textbf{Bottom:} agents' disagreement on the
		value of $x$, quantified by $\| (I-\mathbf{1}\mathbf{1}^\top
		/ n) x^k \|_2$.}\label{fig:sims}
\end{figure}

\section*{Acknowledgements}

The material in this chapter, in full, is being revised and prepared for submission to the Systems \& Control Letters. It may appear as \textit{Distributed Stochastic Nested Optimization via Cubic Regularization}, T.~Anderson and S.~Mart{\'i}nez. The dissertation author was the primary investigator and author of this paper.

\chapter{Distributed Resource Allocation with Binary Decisions via Newton-like Neural Network Dynamics}
\label{chap:NNN}

This chapter aims to solve a distributed resource allocation problem with binary local constraints.
The problem is formulated as a binary program with a cost
function defined by the summation of
agent costs plus a global mismatch/penalty
term.
We propose a
modification of the Hopfield Neural Network (HNN) dynamics in order to
solve this problem while incorporating a novel Newton-like weighting
factor. This addition lends itself to fast avoidance of saddle
points, which the gradient-like HNN is susceptible to. Turning to a
multi-agent setting, we reformulate the problem and develop a
distributed implementation of the Newton-like dynamics. We show that
if a local solution to the distributed reformulation is obtained, it
is also a local solution to the centralized problem. A main
contribution of this work is to show that the probability of
converging to a saddle point of an appropriately defined energy
function in both the centralized and distributed settings is zero
under light
assumptions.
Finally, we enlarge our algorithm with an annealing technique which
gradually learns a feasible binary solution. Simulation results
demonstrate that the proposed methods are competitive with
centralized greedy and SDP relaxation approaches in terms of
solution quality, while the main advantage of our approach is a
significant improvement in runtime over the SDP relaxation method
and the distributed quality of implementation.

\section{Bibliographical Comments}

Quadratic programs with nonconvex binary
constraints are known to be NP-hard in general,
see~\cite{PC-AS:95,DL-XS-SG-JG-CL:10}. In this chapter, we consider a problem which is quite applicable to the economic dispatch problem in power networks, see~\cite{WGG-OM-EH-AG-LGN:19,XH-JY-TH-CL:19,MV-JBG-NQ-AC-MS:20} for recent examples in microgrid environments and~\cite{TA-CYC-SM:18-auto} for a distributed Newton-like method in a more abstract setting. However, none of these examples address devices with binary constraint sets. The binary problem is, however, desirable to approach in a distributed context~\cite{ZY-AB-HZ-NZ-QX-CK:16,PY-YH-LF:16}. Greedy
algorithms~\cite{TC-CL-RR-CS:09} have been proposed for binary
programs, such as the well-known Traveling Salesman Problem (TSP), but
it is well documented that these methods can greatly suffer in
performance~\cite{GG-AY-AZ:02} except in cases where the cost function
is submodular~\cite{GLN-LAW-MLF:78,MS-SB-HV:10}. A more modern
approach to solving optimization problems with a binary feasibility
set is to cast them as a semidefinite program (SDP) with a nonlinear
rank constraint, see~\cite{SP-FR-HW:95,LV-SB:96,SB-LV:97} for some
classical references or~\cite{ZQL-WKM-AMCS-YY-SZ:10,PW-CS-AH-PT:17}
for more recent work on the topic.
By relaxing the rank constraint, a convex problem is obtained whose
solution can be shown to be equal to the optimal dual value of the
original problem, see e.g.~\cite{PP-SL:03}. However, it is necessary
in these approaches to either impose a single centralized coordinator
to compute the solution and broadcast it to the actuators or agents,
or schedule computations, which suffers from scalability issues,
privacy concerns, and does not enjoy the simpler and more robust
implementation of a distributed architecture in a large network.

Neuro-dynamic programming is a different paradigm for addressing
nonconvex problems with computational tractability,
see~\cite{DPB-JNT:96} for a broad reference. A neural-network based
method for binary programs was first developed by Hopfield
in~\cite{JH-DT:85}, which was originally proposed in order to address
TSPs. We refer to this method from here on as a Hopfield Neural
Network (HNN). This method provided a completely different avenue for
approaching binary optimizations, and followup works are found
in~\cite{KS:96,JM:96,BKP-BKP:92,SB-ZA:02}. These works formalize and
expand the framework in which the HNN method is applicable. However,
these algorithms essentially implement a gradient-descent on an
applicable nonconvex energy function, which is susceptible to being
slowed down by convergence to saddle-points. There are avenues for
Newton-like algorithms in nonconvex environments to address this
issue, which incorporate some treatment of the negative Hessian
eigenvalues in order to maintain a monotonic descent of the cost
function, see e.g.~\cite{PG-WM-MW:81,YD-RP-CG-KC-SG-YB:14}. A recently
developed method employs a Positive-definite Truncated inverse
(PT-inverse) operation on the Hessian of a nonconvex energy or cost
function in order to define a nonconvex Newton-descent
direction~\cite{SP-AM-AR:19}, although the technique does not
presently address binary settings. Perhaps more importantly, all
variants of existing HNN methods and the aforementioned works for
nonconvex Newton-like algorithms are framed for centralized
environments in which each agent knows global information about the
state of all other agents, which is not
scalable.

\section*{Statement of Contributions}
The contributions of this chapter are threefold. We start by considering
a binary programming problem formulated as a summation of local costs
plus a squared global term. By leveraging a specific choice for the
cost functions, we adapt the setting to an HNN framework. Then, we
propose a novel modification of the dynamics with a PT-inverse of the
Hessian of an appropriate energy function to define centralized
$\binpromc$ ($\binpac$). We prove a rigorous convergence result 
to a local minimizer, thus excluding saddle-points, with probability
one, given some mild assumptions on the algorithm parameters and
initial condition.
Thirdly, we reformulate the problem so that it is solvable via a
distributed algorithm by means of an auxiliary variable. We show that
local solutions of the distributed reformulation are equivalent to
local solutions of the centralized one, and we define a corresponding
energy function and distributed algorithm for which we show
convergence to a local minimizer with probability
one. 
Simulations validate that our method is
superior to SDP relaxation approaches in terms of runtime and
scalability and outperforms greedy methods in terms of
scalability.

\section{Problem Statement and Dual Problem}

Here, we formally state the nonconvex optimization problem
we wish to solve and formulate its dual for the sake of deriving a
lower bound to the optimal cost.

We aim to find an adequate solution to a resource
  allocation problem where the optimization variables take the form
of binary decisions over a population of $n$ agents. We
  note that the problem we consider is  applicable to generator
  dispatch and active device response in an economic dispatch power
  systems setting~\cite{CAISO-BPM:18}, but the remainder of the chapter
  will frame it primarily as resource allocation. Let each
agent $i\in\until{n}$ be endowed with a
  decision variable $x_i$ and a cost $c_i\in\real$, a value which
  indicates the incremental cost of operating in the $x_i=1$ state
  versus the $x_i=0$ state. We do not impose a sign restriction on
  $c_i$, but this may be a common choice in the power systems setting
  where $x_i=1$ represents an ``on" device state and $x_i=0$
  represents ``off." Additionally, each agent is endowed with a
  parameter $p_i$ which represents some incremental consumption or
  generation quantity when operating in the $x_i=1$ state versus
  $x_i=0$ and also a passive cost $d_i$.

We are afforded some design choice in the cost function
models for $x_i\notin\{0,1\}$, and for each $i \in
  \until{n}$, so we design abstracted cost functions
$f_i:[0,1]\rightarrow \real$ that satisfy $f_i(0) = d_i$
and $f_i(1) = c_i + d_i, \forall i$. This
  design choice is intrinsic to a cost model for any separable binary
  decision optimization context. In particular, the value of
  $f_i(x_i)$ for any $x_i\notin\{0,1\}$ is only relevant to the
  algorithm design, but need not have a physical interpretation or
  pertain to the optimization model since these points are
  infeasible. With this in mind, we enlarge the cost
  model by adopting the following:
\begin{assump}\longthmtitle{Quadratic Cost Functions}\label{ass:costs}
	The local cost functions $f_i$ take the form
	\begin{equation*}
	f_i(x_i) = \dfrac{a_i}{2} (x_i - b_i)^2 - \dfrac{a_ib_i^2}{2} + d_i,
	\end{equation*} with $a_i, b_i, d_i \in \real$.
\end{assump}
Note that, for any value $c_i = f_i(1)-f_i(0)$, there exists a
family of coefficients $a_i, b_i$ such that $(a_i/2) (1 -
b_i)^2 - (a_i/2) b_i^2 = c_i$. Further, the constant terms ensure $f_i(0) = d_i$ and $f_i(1) = c_i + d_i$. The design of
$a_i, b_i$ will be discussed in Section~\ref{sec:all-hop}.

The problem we aim to solve can now be formulated as:
\begin{equation*} 
  \Pc 1: \
  \underset{x\in\{0,1\}^n}{\text{min}} \
  f(x) = \sum_i^n f_i(x_i) + \dfrac{\gamma}{2}\left(p^\top x - \subscr{P}{r}\right)^2.  
\end{equation*}
Here, $\subscr{P}{r}\in\real$ is a given reference value to be matched
by the total output $p^\top x$ of the devices, with
$p\in\real^n$ having entries $p_i$. This matching is enforced by means
of a penalty term with coefficient $\gamma >0$ in $\Pc 1$.
In the power systems setting, $\subscr{P}{r}$ can
  represent a real-power quantity to be approximately matched by the
  collective device-response. The coefficient $\gamma$ and the signal
  $\subscr{P}{r}$ are determined by an Independent System Operator
(ISO) and communicated to a Distributed Energy Resource Provider
(DERP) that solves $\Pc 1$ to obtain a real-time dispatch solution,
see~\cite{CAISO-BPM:18} for additional information.

  The primal $\Pc 1$ has an associated dual $\mathcal{D}1$ which takes
  the form of a semidefinite program (SDP) whose optimal value lower
  bounds the cost of $\Pc 1$. This SDP is
\begin{subequations} \label{eq:dual-prob}
	\begin{align}
	\mathcal{D}1: \ & \underset{\mu\in\real^n,\Delta\in\real}{\text{max}}
	& & \Delta,  \label{eq:dual-cost}\\
	& \text{subject to} & & \begin{bmatrix} \dfrac{1}{2}Q(\mu) &
	\xi(\mu) \\ \xi(\mu)^\top & \zeta - \Delta
	\end{bmatrix} \succeq 0. \label{eq:gen-fiedler-const}
	\end{align}
\end{subequations}
In $\mathcal{D}1$, $Q : \real^n \rightarrow \real^{n\times n}$ and
$\xi: \real^n \rightarrow \real^n$ are real-affine functions of $\mu$ and $\zeta$ is a constant. These definitions are $Q(\mu) = \left(\diag{a/2 + \mu} + \dfrac{\gamma}{2}pp^\top\right),
\xi(\mu) = ((a_i b_i)_i + \mu + \gamma \subscr{P}{r}p),$ and
$\zeta = \sum_{i=1}^n \dfrac{a_i b_i^2}{2} + \dfrac{\gamma}{2}\subscr{P}{r}^2$. See~\cite{SB-LV:04} for more detail on the derivation of $\mathcal{D}1$. 

\section{Centralized Newton-like Neural Network}\label{sec:all-hop}

In this section, we develop the Centralized $\binpromc$, or 
$\binpac$, which is well suited for solving $\P1$ in a centralized
 setting.

To draw analogy with the classic Hopfield Neural Network approach we
will briefly introduce an auxiliary variable $u_i$ whose relation to
$x_i$ is given by the logistic function $g$ for each $i$:
\begin{align}
\label{eq:act-fn}
x_i &= g(u_i) = \dfrac{1}{1+e^{-u_i/T}}, \quad &&u_i\in\real, \\
u_i &= g^{-1}(x_i) = -T \ \log\left(\dfrac{1}{x_i} - 1\right), \quad
&&x_i\in\left( 0, 1\right), \nonumber
\end{align}
with temperature parameter $T > 0$. 

Let $x\in (0,1)^n, u\in\real^n$ be vectors with entries given by $x_i,
u_i$. To establish our algorithm, it is appropriate to first define an
energy function related to $\Pc 1$. Consider
\begin{equation}\label{eq:energyc}
E(x) = f(x) + \dfrac{1}{\tau}\sum_i \int_{0}^{x_i} g^{-1} (\nu) d\nu,
\end{equation}
where $\tau > 0$ is a time-constant and for $z \in [0,1]$,
	\begin{equation*}
	\int_{0}^{z} g^{-1}(\nu)d\nu = \begin{cases}
	T\left(\log(1-z) - z \log (\frac{1}{z} - 1)\right), & z\in(0,1), \\
	0, & z\in\{0,1\}. 
	\end{cases}
      \end{equation*}
    The classic HNN implements dynamics of the form $\dot{u} = -\nabla_x
    E(x)$, where the equivalent dynamics in $x$
    can be computed as $\dot{x} = - \nabla_x E(x)dx/du$. 
     These dynamics can  be thought of to
      model the interactions between neurons in a neural network or
      the interconnection of amplifiers in an electronic circuit,
      where in both cases the physical system tends toward low energy
      states, see~\cite{JH-DT:85,KS:96}. 
      In an optimization setting, low energy 
      states draw analogy to low cost solutions. We 
      now describe our modification to the classical HNN dynamics.

Recall that the domain of $x$ is $(0,1)^n$ and our elementwise notation for $\log$ and division. We have the expressions $\nabla_x E(x) = -Wx -v - (T/\tau) \log\left(1/x_i-1\right)_i$ and $dx/du = (x-(x_i^2)_i)/T$,
where $W = -\diag{a} - \gamma pp^\top\in \real^{n\times n}$ and $v = (a_i b_i)_i + \gamma P_{r} p\in \real^n$ are defined via $f$. From this point forward, we work mostly in terms of $x$ for the sake
of consistency.  Consider modifying the classic HNN dynamics with a
PT-inverse $(\vert H(x) \vert_m)^{-1} \succ 0$ as in~\cite{SP-AM-AR:19}, where $H(x) = \nabla_{xx} E(x)$. The $\binpac$ dynamics are then given~by:
\begin{equation}\label{eq:binpac}
\begin{aligned}
\dot{x} &= -(\vert H(x)\vert_m)^{-1} \diag{\frac{dx}{du}} \nabla_x{E(x)} \\
&=(\vert H(x)\vert_m)^{-1} \diag{\frac{(x_i-x_i^2)_i}{T}} \left(Wx + v +
\dfrac{T}{\tau}\log\left(1/x_i-1\right)_i\right).
\end{aligned}
\end{equation}

These dynamics lend
to the avoidance of saddle points of $E$ via inclusion of the PT-inverse weighting $(\vert H(x)\vert_m)^{-1}$, in contrast to the more first-order flavor of the classic HNN dynamics. To see this, consider
the eigendecomposition $H(\tilde{x}) = Q^\top \Lambda Q$ at some
$\tilde{x}$ near a saddle point, i.e. $\nabla_x E(\tilde{x})\approx
0$. If many entries of $\Lambda$ are small in magnitude and remain
small in the proximity of $\tilde{x}$, then the gradient is
changing slowly along the ``slow" manifolds associated with the eigenspace of
the small eigenvalues. This is precisely what the PT-inverse is
designed to combat: the weighting of the
dynamics is increased along these manifolds by a factor that is
inversely proportional to the magnitude of the
eigenvalues. Additionally, negative eigenvalues of the Hessian
are flipped in sign, which causes attractive
manifolds around saddle points to become repellent.

It is desirable for $E$ to be concave on most of its domain so the
trajectories are pushed towards the feasible points of $\Pc 1$;
namely, the corners of the unit hypercube. To examine this, the Hessian of $E$ can be
computed as $H(x) = \frac{d^2 f}{dx^2} + \frac{1}{\tau}\diag{\frac{dg^{-1}(x)}{dx}} = -W + \dfrac{T}{\tau}\diag{\frac{1}{(x_i-x_i^2)_i}}.$
Notice that the second term is positive definite on $x\in(0,1)^n$ and
promotes the convexity of $E$, particularly for elements $x_i$ close to $0$ or $1$. For a fixed $T,\tau$, choosing $a_i <
-\gamma\Vert p\Vert^2 - 4T/\tau, \forall i$ guarantees $E(x) \prec 0$
at $x=(0.5) \ones_n$. 
Generally speaking, choosing $a_i$ to be negative and
large in magnitude lends itself to concavity of $E$ over a larger subset of its domain and to trajectories converging closer
to the set $\{0,1\}^n$. However, this comes at the expense of not exploring a rich subset
of the domain. At the end of this section, we develop a Deterministic
Annealing (DA) approach inspired by~\cite{KR:98} for the online
adjustment of $T,\tau$ to obtain an effective compromise between
exploration of the state space and convergence to a feasible point of
$\Pc 1$.

We now characterize the equilibria of~\eqref{eq:binpac} for
$x\in[0,1]^n$. 
It would appear that $x$ with some components $x_i\in\{0,1\}$ are
candidate equilibria due to the $x_i-x_i^2$ factor vanishing. However,
the dynamics are not well defined here due to the $\log$ term. Additionally, note that $\lim_{x_i\rightarrow \delta} e_i^\top H(x) e_i = \infty, \ \delta \in \{0,1\}, \forall i,$
where $e_i$ is the $\supscr{i}{th}$ canonical basis vector. Due to the
$\frac{T}{\tau(x_i - x_i^2)}$ term dominating $W$ in the expression for $H$ when $x_i$ values
are close to $\{0,1\}$, it follows that an eigenvalue of $(\vert
H(x)\vert_m)^{-1}$ approaches zero as $x_i\rightarrow 0$ or $1$ with
corresponding eigenvector approaching $v_i = e_i$:
\vspace{-0.5cm}\begin{equation*}
  \lim_{x_i\rightarrow \delta} = v_i^\top (\vert H(x)\vert_m )^{-1}
  v_i = \frac{T}{\tau}(x_i - x_i^2) = 0, \quad 
  \delta \in \{0,1\}, \forall i.
\end{equation*}
Using this fact, and ignoring $T,\tau > 0$, we can compute the 
undetermined limits in the components of $\dot{x}$ as $x_i \rightarrow
\delta \in \{0,1\}$ by repeated applications of L'Hospital's rule:
\begin{equation}\label{eq:lhospital-eq}
  \lim_{x_i\rightarrow \delta} \log\left(\dfrac{1}{x_i}-1\right)(x_i - x_i^2)^2 
  = \begin{cases}
  0, & \delta = 0^+, \\
  0, & \delta = 1^-. \end{cases}
\end{equation}
Thus, components $x_i\in\{0,1\}$ constitute candidate equilibria. We
will, however, return to the first line of~\eqref{eq:lhospital-eq} in the proof of
Lemma~\ref{lem:fwd-inv-c} to show that they are unstable.  As for
components of $x$ in the interior of the hypercube, the expression
$\dot{x} =0$ can not be solved for in closed form. However, we provide
the following Lemma which shows that the set of equilibria is finite.
\begin{lem}\longthmtitle{Finite Equilibria}~\label{lem:finite-eq}
  Let $\X$ be the set of equilibria of~\eqref{eq:binpac}
  satisfying $\dot{x} = 0$ on $x\in[0,1]^n$. The set $\X$ is finite.
\end{lem}
\begin{proof}
	First consider only $\X\cap (0,1)^n$. Note that $(\vert H(x)\vert_m
	)^{-1} \succ 0$ (by construction) and \\ $\diag{(x_i-x_i^2)_i/T}\succ 0$ on
	$x\in(0,1)^n$, so we focus on
	\begin{equation}\label{eq:grad-zero}
	Wx + \frac{T}{\tau}\log\left(1/x_i-1\right)_i + v = \zeros_n.
	\end{equation}
	Examining the above expression elementwise, it is nonconstant,
	continuous, and its derivative changes sign only a finite
	number of times. Therefore, the total number of zeros on
	$(0,1)^n$ must be finite.
	
	Now consider the $\supscr{i}{th}$ element
	of~\eqref{eq:grad-zero} for $x_j\rightarrow 0$ or $1$ for all
	$j$ in an arbitrary permutation of
	$\until{n}\setminus{\{i\}}$. Since the number of these
	permutations is finite, and each permutation still gives rise
	to a finite number of solutions to~\eqref{eq:grad-zero} in the
	$\supscr{i}{th}$ component, it follows that $\X$ is finite.
\end{proof}

To demonstrate the qualitative behavior of equilibria in a simple
case, consider a one-dimensional example with $a > - \gamma p^2 -
4T/\tau$ and recall that, for $x\in(0,1)$, the sign of $-\nabla_x
E(x)$ is the same as $\dot{x}$. In Figure~\ref{fig:xdot1d}, we observe
that $-\nabla_x E(x)$ monotonically decreases in $x$, and a globally
stable equilibrium exists in the interior $x\in(0,1)$ near
$x=0.5$. On the other hand, $a < - \gamma p^2 - 4T/\tau$ gives way
to $3$ isolated equilibria in the interior (one locally unstable near
$x=0.5$ and two locally stable near $x\in\{0,1\}$). This behavior
extends in some sense to the higher-dimensional case. Therefore, for a
scheme in which $T$ and $\tau$ are held fixed, we prescribe $a < -
\gamma \| p\|^2 - 4T/\tau$. We provide a Deterministic Annealing
(DA) approach inspired by~\cite{KR:98} for the online adjustment of $T,\tau$ in the following subsection which
compromises with this strict design of $a$.

\begin{figure}[h]
  \centering
	\includegraphics[scale = 0.5]{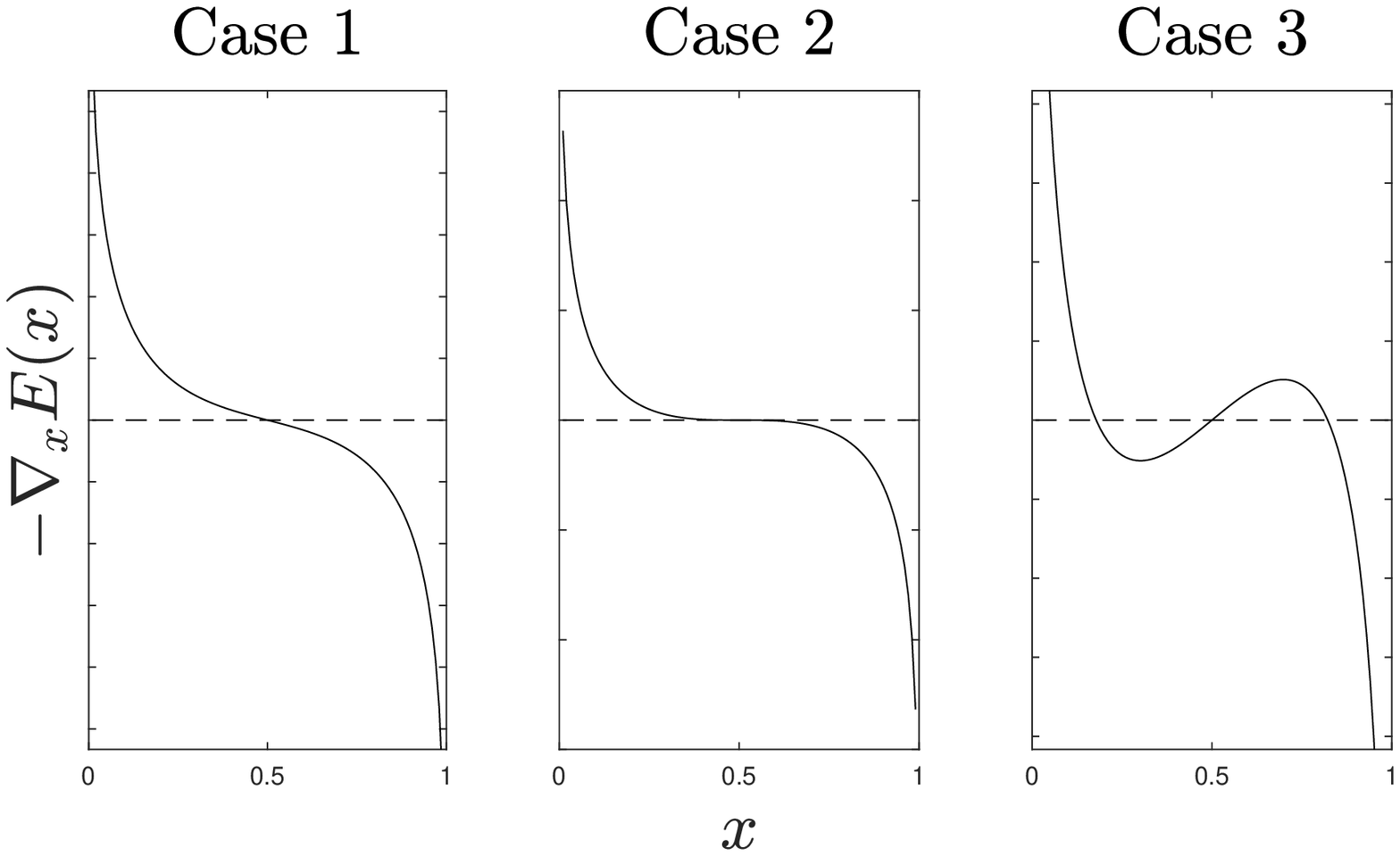}
	\includegraphics[scale = 0.5]{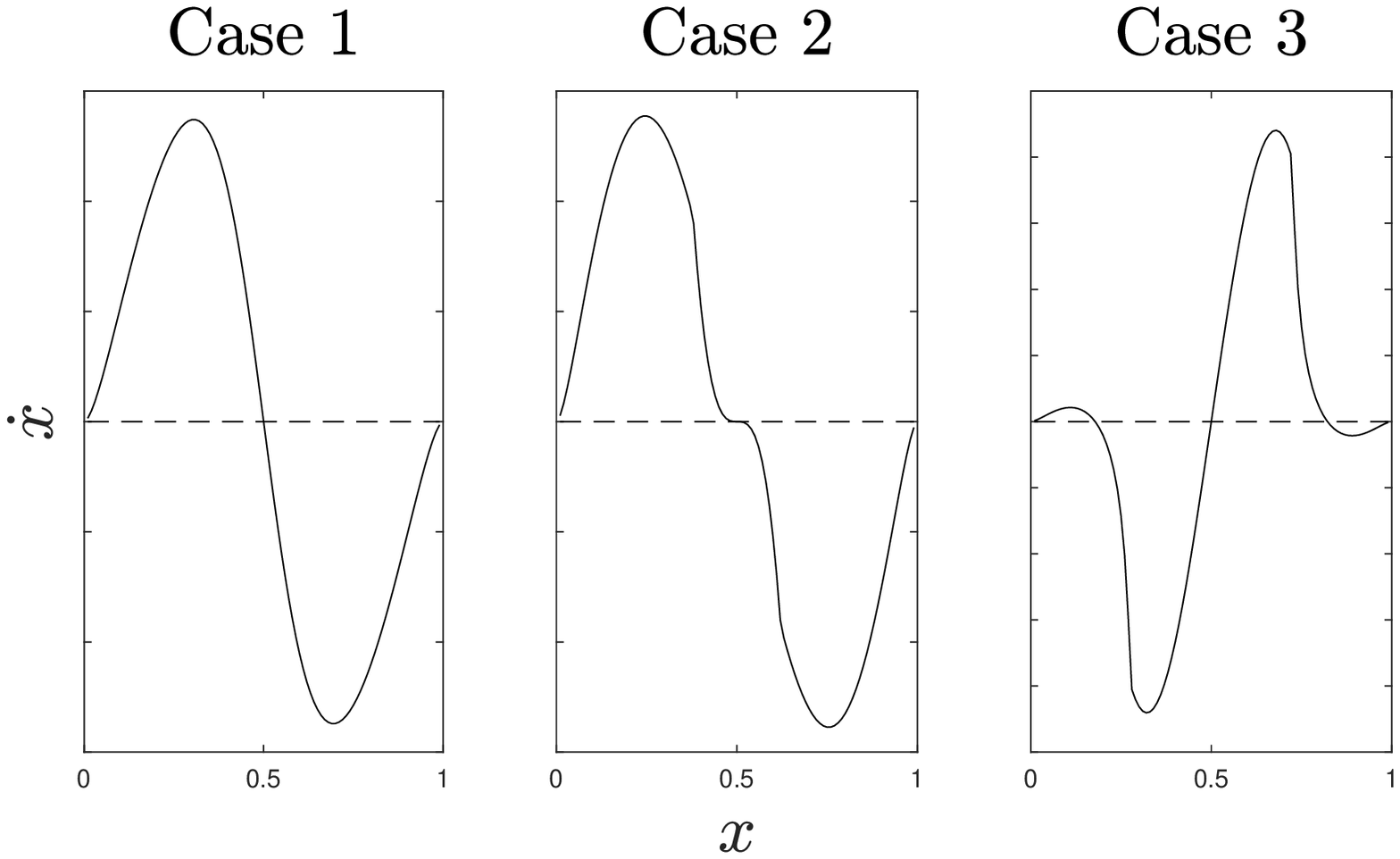}
	\caption{Illustration of $-\nabla_x E(x)$ (top) and $\dot{x}$
          (bottom) for three instances of $a$. Case 1: $a > -\gamma
          \Vert p \Vert ^2 - 4T/\tau$, Case 2: $a = -\gamma \Vert p
          \Vert ^2 - 4T/\tau$, Case 3: $a < -\gamma \Vert p \Vert ^2 -
          4T/\tau$.}
	\label{fig:xdot1d}
\end{figure}

Finally, we establish a Lemma about the domain of the trajectories
of~\eqref{eq:binpac}.
\begin{lem}\longthmtitle{Forward Invariance of the Open Hypercube}\label{lem:fwd-inv-c}
  The open hypercube $(0,1)^n$ is a forward-invariant set under the
  $\binpac$ dynamics~\eqref{eq:binpac}.
\end{lem}
\begin{proof}
	Consider again the terms of $\dot{x}$ elementwise. There are two cases to
	consider for evaluating $x_i$: $x_i = \varepsilon$ and $x_i =
	1-\varepsilon$ for some $0< \varepsilon \ll 1$ sufficiently small
	such that the terms of $(\vert H(x)\vert_m)^{-1}$ are still
	dominated by $(1/x_i-x_i^2)$ and the $Wx + v$ are still dominated by
	the $\log$ term. Then, consider the expression
	\begin{equation}\label{eq:log-express}
	\log\left( 1/x_i - 1\right)(x_i-x_i^2)^2.
	\end{equation}
	For $x_i = \varepsilon \approx 0$,~\eqref{eq:log-express} evaluates
	to a small positive value, and for $x_i = 1-\varepsilon \approx
	1$,~\eqref{eq:log-express} evaluates to a small negative value. We
	have argued that these are the dominating terms regardless of values
	of the remaining components of $x$, and so we conclude that
	$x_i\in\{0,1\}$ are componentwise anti-stable and that elements of
	$x$ will never approach $0$ or $1$. Thus, the open hypercube is
	forward invariant.
\end{proof}

Knowing that $\Pc 1$ is generally NP-hard, it is unlikely
that a non-brute-force algorithm exists that can converge to a global
minimizer. For this reason, we aim to establish
asymptotic stability to a local minimizer of $E$. We first establish some assumptions.

\begin{assump}\longthmtitle{Random Initial Condition}\label{ass:init}
  The initial condition $x(0)$ is chosen randomly according to a
  distribution $\Pp$ that is nonzero on sets that have nonzero volume
  in $[0,1]^n$.
\end{assump}
An appropriately unbiased initial condition for our algorithm is
$x(0)\approx(0.5)\ones_n$, 
which is adequately far from the local minima located near corners of
the unit cube. So, we suggest choosing a uniformly random $x(0) \in \B
((0.5)\ones_n ,\epsilon)$, where $0 < \epsilon \ll
1$.

\begin{assump}\longthmtitle{Choice of $T,\tau$}\label{ass:t-tau}
  The constants $T, \tau > 0$ are each chosen randomly according to a
  distribution $\bar{\Pp}$ that is nonzero on sets that have nonzero
  volume on $\real_+$. 
\end{assump}
Similarly to $x(0)$, we suggest choosing these constants uniformly
randomly in a ball around some nominal $T_0, \tau_0$,
i.e. $T\in\B(T_0,\epsilon), \tau\in\B(\tau_0,\epsilon), 0 < \epsilon
\ll 1$. The $T_0,\tau_0$ themselves are design parameters stemming from the neural network model, and we provide some intuition for selecting these in the simulation Section.

Now we state the main convergence result of $\binpac$ in Theorem~\ref{thm:cent-cvg}, which states that for a random choice of $T,\tau$, an
	initial condition chosen randomly from $(0,1)^n$ converges
	asymptotically to a local minimizer of $E$ with probability one.
\begin{thm}\longthmtitle{Convergence of $\binpac$}\label{thm:cent-cvg}
Given an initial condition $x(0)\in(0,1)^n$, the
    trajectory $x(t)$ under $\binpac$ converges asymptotically to a
    critical point $x^\star$ of $E$. In addition, under
  Assumption~\ref{ass:init}, on the random choice of initial
  conditions, and Assumption~\ref{ass:t-tau}, on the random choice of
  $T,\tau$, the probability that $x(0)$ is in the set
  $\underset{\hat{x}}{\cup} \W^s (\hat{x})$, where $\hat{x}$ is a
  saddle-point or local maximum of $E$, is zero.
\end{thm}
\begin{proof}
	Let $\X$ be the set of all critical points of $E$. We first
	establish that $E$ decreases along the trajectories of $\binpac$ and
	that $x(t)$ converges asymptotically to $\X$. Differentiating $E$ in
	time, we obtain:
	\begin{equation}\label{eq:dE-dt}
	\begin{aligned}
	\dfrac{dE}{dt} &=
	\dot{x}^\top \nabla_x E(x) = \dot{x}^\top\left(-Wx -v + g^{-1}(x)/\tau\right) \\
	&= -\dot{x}^\top \diag{\frac{T}{(x_i-x_i^2)_i}}\vert H(x)\vert_m \dot{x} <0, \quad \text{for} \ \dot{x} \neq 0, \
	x\in(0,1)^n.
	\end{aligned}
	\end{equation}
	Recall that $x(t)\in(0,1)^n$ for all $t\geq 0$ due to
	Lemma~\ref{lem:fwd-inv-c}. From~\eqref{eq:binpac} and the discussion
	that followed on equilibria, $\dot{x} = 0$ implies $\nabla_x E(x) =
	0$ due to $(\vert H(x) \vert_m)^{-1} \succ 0$ and $\diag{(x_i-x_i^2)_i/T}
	\succ 0$ on $x\in(0,1)^n$. The domain of $E$ is the compact set
	$[0,1]^n$ (per the definition of the integral terms), and $E$ is
	continuous and bounded from below on this domain, so at least one
	critical point exists. Combining this basic fact
	with~\eqref{eq:dE-dt} shows that the $\binpac$ dynamics
	monotonically decrease $E$ until reaching a critical point. More
	formally, applying the LaSalle Invariance Principle\cite{HKK:02}
	tells us that the trajectories converge to the largest invariant set
	contained in the set $dE/dt = 0$. This set is $\X$, which is finite
	per Lemma~\ref{lem:finite-eq}. In this case, the LaSalle Invariance Principle
	additionally establishes that we converge to a single
	$x^\star\in\X$.

	The proof of the second statement of the theorem relies on an
	application of the Stable Manifold Theorem (see~\cite{JG-PH:83}) as well as Lemma~\ref{lem:finite-eq}. Let $\dot{x} = \varphi_{T,\tau}(x)$
	for a particular $T,\tau$. 
	We aim to show that 
	$\Pp[\cup_{\hat{x}} \setdef{\W_s (\hat{x})}{\hat{x} \text{ is a saddle
			or local maximum}} ] = 0$
	under Assumptions~\ref{ass:init}-\ref{ass:t-tau}. It is
	sufficient to show that, for each critical point $x^\star$
	such that $\varphi_{T,\tau}(x^\star) = 0$, and almost all
	$T,\tau$, $D\varphi_{T,\tau} (x^\star)$ is full rank and
	its eigenvalues have non-zero real
	parts.  
	The reason for this argument is the following: let $x^\star$ be a
	critical point with $D\varphi_{T,\tau} (x^\star)$ full rank and
	eigenvalues with non-zero real parts.  If the eigenvalues do
	\emph{not} all have positive real parts, then some have negative
	real parts, which indicates that $x^\star$ is a saddle or local
	maximum of $E$. These negative real-part eigenvalues induce an
	unstable manifold of dimension $n_u \geq 1$. As such, the globally
	stable set $\W_s (x^\star)$ is a manifold with dimension $n - n_u
	< n$, and $\Pp\left[x(0)\in\W_s(x^\star) \right]=0$ per
	Assumption~\ref{ass:init}.
	
	To argue this case, define $h:(0,1)^n\times \real\times\real
	\rightarrow \real$ as
	\begin{equation*}
	h(x,T,\tau) = \det{D\varphi_{T,\tau}(x)}.
	\end{equation*}
	We now leverage Assumption~\ref{ass:t-tau} and~\cite{BM:15} to claim
	first that $\bar{\Pp}\left[h(x^\star,T,\tau) = 0\right] = 0$
	for each $x^\star\in\X$, i.e. $D\varphi_{T,\tau} (x^\star)$ is full
	rank for each $x^\star$ with probability one w.r.t.~$\bar{\Pp}$.  We
	first address the points $x$ for which the function $h$ is
	discontinuous. Define $\hat{\X}$ as the set of $x$ for which the
	truncation of the eigenvalues of $H(x)$ becomes active, i.e. the
	discontinuous points of $h$. Although we do not write it as such,
	note that $H$ is implicitly a function of $T,\tau$ and that the
	eigenvalues of $H$ can be expressed as nonconstant real-analytic
	functions of $T,\tau$. Considering this fact and an arbitrary $x$,
	the set of $T,\tau$ which give $x\in\hat{\X}$ has measure zero with
	respect to $\real^2$\cite{BM:15}. Thus, for particular $T,\tau$, $h$
	is $C^{\infty}$ almost everywhere. Applying once more the argument
	in~\cite{BM:15} and Assumption~\ref{ass:t-tau} with the fact that
	$h$ is a nonconstant real analytic function of $T,\tau$ we have that
	\begin{align*}
	\bar{\Pp}\,[\mathcal{T}(\hat{x}) \triangleq  \{&(T,\tau) \,| \,
	h(\hat{x}, T,\tau) = 0 \}]  = 0, \quad \forall\, \hat{x}\notin \hat{\X}.
	\end{align*} 
	Now consider the set of critical points as an explicit function of
	$T,\tau$ and write this set as $\X(T,\tau)$. 
	Recalling Lemma~\ref{lem:finite-eq}, the set of $\hat{x}$ that we
	are interested in reduces to a finite set of critical points
	$x^\star\in \X(T,\tau)$. Thus, we can conclude that
	$\bar{\Pp}(\cup_{x^\star \in \X(T,\tau)}\,\mathcal{T}(x^\star))
	\le \sum_{x^\star \in \X(T,\tau)} \bar{\Pp}(\mathcal{T}(x^\star))=
	0$.
	
	There is an additional case which must be considered, which is that
	$h(x^\star ,T,\tau)\neq 0$, but some eigenvalues of
	$D\varphi_{T,\tau}(x^\star)$ are purely imaginary and induce stable
	center manifolds, which could accommodate the case of a globally
	stable set which is an $n$-dimensional manifold (i.e. the ``degenerate saddle" case). We consider the
	function $h$ mostly out of convenience, but the argument can be
	extended to a function $\mathbf{h}: (0,1)^{n}\times \real \times
	\real \rightarrow \complex^n$ which is a map to the roots of the
	characteristic equation of $D\varphi_{T,\tau}(x)$. We are concerned
	that each element of $\mathbf{h}(x,T,\tau )$ should have a nonzero
	real part almost everywhere. To extend the previous case to this,
	consider the identification $\complex \equiv \real^2$ and compose
	$\mathbf{h}$ with the nonconstant real analytic function $\zeta(w,z)
	= w$, for which the zero set is $w \equiv 0$, corresponding to the
	imaginary axis in our identification. From this, we obtain a
	nonconstant real-analytic as before whose zero set is the imaginary
	axis. Applying the argument in~\cite{BM:15} in a
	similar way as above, $\mathbf{h}(x,T, \tau )$ has nonzero real
	parts for almost all $(T,\tau )$ for each
	$x$. Therefore, the probability of a particular saddle point or local maximum
	$x^\star$ having a nonempty stable center manifold is zero for
	arbitrary $x(0)$ satisfying Assumption~\ref{ass:init} and $T,\tau$ satisfying
	Assumption~\ref{ass:t-tau}.
\end{proof}

We now define a Deterministic Annealing (DA) variant inspired
by~\cite{KR:98} to augment the $\binpac$ dynamics and provide a method
for gradually learning a justifiably good feasible point of $\Pc 1$. In~\cite{KR:98}, the author justifies the deterministic online tuning of a temperature parameter in the context of data clustering and shows that this avoids poor local optima by more thoroughly exploring the state space. Similarly, we aim to learn a sufficiently good solution trajectory by allowing the dynamics to explore the interior of the unit hypercube in the early stages of the algorithm, and then to force the trajectory outward to a feasible binary solution by gradually adjusting $T$ or $\tau$ online. 

Consider either reducing the temperature $T$ or increasing the time
constant $\tau$ during the execution of $\binpac$. 
This reduces the terms in $E$ which promote
convexity, particularly near the boundaries of the unit hypercube. As
$T,\tau$ are adjusted, for $a\prec -\gamma\Vert p\Vert^2$, the domain
of $E$ becomes gradually more concave away from the corners of the
unit hypercube. Thus, starting with $T_0/\tau_0$ sufficiently large, the
early stages of the algorithm promote exploration of the interior of
the state space. As $T/\tau$ is reduced at a rate dictated by $\beta$, the equilibria of $E$ are
pushed closer to (and eventually converge to) the feasible points of
$\Pc 1$. The update policy we propose is described formally in
Algorithm~\ref{alg:da}, and we further explore its performance in simulation. 

\begin{algorithm}
	\caption{Determinisitc Annealing}\label{alg:da}
	\begin{algorithmic}[1]
		\Procedure{Det-Anneal}{$\beta > 1,T_0,\tau_0,t_d$}
		\State Initialize $x(0)$
		\State $T\gets T_0, \tau\gets \tau_0$
		\While{\texttt{true}}
		\State Implement $\binpac$ for $t_d$ seconds\label{line:da-condition}
		\State $\tau \gets \beta \tau \text{\quad or\quad} T \gets (1/\beta) T$
		\EndWhile
		\EndProcedure
	\end{algorithmic}
\end{algorithm}

Note that Algorithm~\ref{alg:da} leads to a hybrid
  dynamic system with discrete jumps in an enlarged state $\phi =
  (x,T,\tau)$, which can cast some doubt on basic existence and
  uniqueness of solutions. We refer the reader to Propositions 2.10
  and 2.11 of~\cite{RG-RS-AT:09} to justify existence and uniqueness
  of solutions in the case of $t_d > 0$ fixed.

\begin{cor}\longthmtitle{Convergence to Feasible Points}
  Under Assumptions~\ref{ass:init}-\ref{ass:t-tau} and $a \prec -\gamma\| p\|^2$, the $\binpac$
  dynamics augmented with Algorithm~\ref{alg:da} converge
  asymptotically to feasible points of $\Pc 1$.
\end{cor}

The result of the Corollary is quickly verified by inspecting the
terms of $H(x)$. 
The function $E$ 
is smooth, strictly concave near $x=(0.5)\ones_n$ for small $T/\tau$ due to the design of $a_i$, and becomes strictly convex as the elements of $x$ approach
$0$ or $1$, corresponding
to isolated local minima of $E$, due to the $T/\tau$ term dominating $H(x)$. As the quantity
$T/\tau$ is reduced under Algorithm~\ref{alg:da}, these local minima
are shifted asymptotically closer to corners of the unit hypercube, i.e.
feasible points of $\Pc 1$.

\section{Distributed Hopfield Neural Network}\label{sec:dist-hop}

With the framework of the previous section we formulate a problem $\P
2$ which is closely related to $\Pc 1$, but for which the global
penalty term can be encoded by means of an auxiliary decision
variable. This formulation leads to the Distributed $\binpromd$, or $\binpad$, which we rigorously analyze for its
convergence properties.

It is clear from the PT-inverse operation and $W$ being nonsparse that
$\binpac$ is indeed centralized. In this section, we design a
distributed algorithm in which each agent $i$ must only know
$p_j, j\in\N_i$ and the value of an auxiliary variable $y_j, j\in\N_i
\cup\N_i^2$, i.e. it must have communication with its two-hop neighbor
set. If two-hop communications are not directly available, the
algorithm can be implemented with two communication rounds per
algorithm step. We provide comments on a one-hop algorithm in
Remark~\ref{rem:one-hop}.

\begin{assump}\longthmtitle{Graph Properties and Connectivity}\label{ass:graph-conn}
  The graph $\mathcal{G} = (\N,\mathcal{E})$ is undirected and
  connected; that is, a path exists between any two pair of nodes and,
  equivalently, its associated Laplacian matrix $L = L^\top$ has rank
  $n-1$.
\end{assump}
Now consider the $n$ linear equations $(p_i x_i)_i + Ly =
(\subscr{P}{r}/n)\ones_n.$ Notice that, by multiplying from the left
by $\ones_n^\top$ and applying $\ones_n^\top L = \zeros_n^\top$, we
recover $p^\top x = \subscr{P}{r}$. Thus, by augmenting the state with
an additional variable
$y\in\real^n$, 
we can impose a distributed penalty term. We now formally state the
distributed reformulation of~$\Pc 1$:
\begin{equation*} 
  \Pc 2: \
  \underset{x\in\{0,1\}^n,y\in\real^n}{\text{min}} \
  \tilde{f}(x,y) = \sum_i^n f_i(x_i) + \dfrac{\gamma}{2}\sigma^\top \sigma,  
\end{equation*}
where the costs $f_i$ again satisfy $f_i(1) - f_i(0) = c_i$ and we
have defined $\sigma = (p_i x_i)_i + Ly - (\subscr{P}{r}/n) \ones_n$ for
notational simplicity. Before proceeding, we provide some context on
the relationship between $\Pc 1$ and $\Pc 2$.

\begin{lem}\longthmtitle{Equivalence of P1 and P2}\label{lem:equiv}
  Let Assumption~\ref{ass:graph-conn}, on graph connectivity, hold, and
  let $(x^\star,y^\star)$ be a solution to $\Pc 2$. Then, $x^\star$ is
  a solution to $\Pc 1$ and $f(x^\star) = \tilde{f}(x^\star,y^\star)$.
\end{lem}
\begin{proof}
	The equivalence stems from the global term and the flexibility in
	the unconstrained $y$ variable. Notice
	\begin{equation*}
	\begin{aligned}
	\dfrac{\gamma}{2} \sigma^\top \sigma & = \dfrac{\gamma}{2}
	\sigma^\top (I_n -
	\ones_n\ones_n^\top/n)\sigma + \dfrac{\gamma}{2} \sigma (\ones_n\ones_n^\top/n)\sigma \\
	&= \dfrac{\gamma}{2} \sigma^\top (I_n -
	\ones_n\ones_n^\top/n)\sigma + \dfrac{\gamma}{2}(p^\top x -
	\subscr{P}{r})^2.
	\end{aligned}
	\end{equation*}
	We have recovered the original global term of $\Pc 1$ in the
	bottom line, so now we deal with the remaining term. The
	matrix $I_n - \ones_n \ones_n^\top/n \succeq 0$ has $\image
	I_n - \ones_n \ones_n^\top/n = \spn \{\ones_n\}^\perp = \image
	L$, given that $L$ is connected. Thus, because $y$ is
	unconstrained and does not enter the cost anywhere else, we
	can compute the set of possible minimizers of $\tilde{f}$
	in closed form with respect to any $x$ as
	\begin{equation*}
	\begin{aligned}
	y^\star &\in \setdef{-L^\dagger \big((p_i x_i)_i - (\subscr{P}{r}/n)\ones_n\big) + \theta \ones_n}{\theta\in\real} \\
	&= \setdef{-L^\dagger(p_i x_i)_i+\theta\ones_n}{\theta\in\real}.
	\end{aligned}
	\end{equation*}
	Moreover, substituting a $y^\star$ gives $\sigma
	\in \spn\{\ones_n\}$, and it follows that the problem $\Pc 2$
	reduces precisely to $\Pc 1$.
\end{proof}

To define $\binpad$, we augment the centralized $\binpac$
with gradient-descent dynamics in $y$ on a newly obtained energy function $\widetilde{E}$ of $\Pc 2$. Define $\widetilde{E}$ as
\begin{equation}\label{eq:energyd}
\widetilde{E}(x,y) = \tilde{f}(x,y) + \dfrac{1}{\tau} \sum_i \int_0^{x_i} g^{-1}(\nu) d\nu.
\end{equation}

In Section~\ref{sec:all-hop}, we obtained a matrix $W$ which was
nonsparse. Define $\widetilde{W}, \tilde{v}$ for $\widetilde{E}$ via $\tilde{f}$ as $\widetilde{W} = -\diag{a + \gamma (p_i^2)_i},
\tilde{v} = (a_i b_i)_i + \gamma\diag{p}\left((\subscr{P}{r}/n)\ones_n
- Ly\right).$
Compute the Hessian of $\widetilde{E}$ with respect to only $x$ as $\widetilde{H}(x) = \nabla_{xx}\widetilde{E}(x,y) = -\widetilde{W} + (T/\tau)\diag{1/x-(x_i^2)_i}$.
Since $\widetilde{H}(x)$ is diagonal, the $\supscr{ii}{th}$ element of
the PT-inverse of $\widetilde{H}(x)$ can be computed locally by each
agent $i$ as:
\begin{equation*}
  (\vert \widetilde{H}(x)\vert_m)^{-1}_{ii} = \begin{cases}
    \vert \widetilde{H}(x)_{ii} \vert^{-1}, 
    &  \vert\widetilde{H}(x)_{ii}\vert \geq m, \\
    1/m, & \text{o.w.}
\end{cases}
\end{equation*}
where $\widetilde{H}(x)_{ii} = a_i + \gamma p_i^2 + T/\tau (x_i -
x_i^2)^{-1}$. 
The $\binpad$ dynamics, which are PT-Newton descent in
$x$ and gradient descent in $y$ on $\widetilde{E}$, are then stated
as:
\begin{equation}\label{eq:binpad}
	\begin{aligned}
          \dot{x} &=(\vert \widetilde{H}(x)\vert_m)^{-1} \diag{\frac{(x_i-x_i^2)_i}{T}} \left(\widetilde{W}x + \dfrac{T}{\tau}\log\left(1/x_i-1\right)_i + \tilde{v}\right), \\
          \dot{y} &= -\alpha\gamma
          L\left(
          	(p_i x_i)_i
            + Ly\right),
	\end{aligned}
\end{equation}
where $\alpha = \diag{\alpha_i}$ is a diagonal matrix of arbitrary positive gains $\alpha_i > 0$.
Due to the new matrices $\widetilde{W},\tilde{v}$ and the sparsity of
$L$, $\dot{x}$ can be computed with one-hop information and $\dot{y}$
with two-hop information (note the $L^2$ term);
thus,~\eqref{eq:binpad} defines a distributed algorithm. Additionally,
recalling the discussion on parameter design, the problem data $a$ and
$b$ can now be locally designed.

Before proceeding, we establish a property of the domain of $y$ and some distributed extensions of Lemmas~\ref{lem:finite-eq} and~\ref{lem:fwd-inv-c}.
\begin{lem}\longthmtitle{Domain of Auxiliary Variable}\label{lem:y-domain}
	Given an initial condition $y(0)$ with $\ones_n^\top y(0) = \kappa$, the trjaectory $y(t)$ is contained in the set
	\begin{equation}\label{eq:ydom}
	\Y = \setdef{\omega + (\kappa/n)\ones_n}{\ones_n^\top \omega = 0}.
	\end{equation}
\end{lem}
\begin{proof}
	The proof is trivially seen by multiplying $\dot{y}$ in~\eqref{eq:binpad} from the left by $\ones_n$ and applying the null space of $L$.
\end{proof}

\begin{lem}\longthmtitle{Closed Form Auxiliary Solution}\label{lem:y-soln}
	For an arbitrary fixed $x\in[0,1]^n$, the unique minimizer $y^\star$ contained in $\Y$ of both $\tilde{f}$ and $\widetilde{E}$ is given by
	\begin{equation}\label{eq:ystar}
	y^\star = -L^\dagger \left(
	p_i \tilde{x}_i\right)_i + \frac{\kappa}{n}\ones_n.
	\end{equation}
	This is also the unique equilibrium of~\eqref{eq:binpad} in $\Y$.
\end{lem}
\begin{proof}
	The first term is computed by setting $\nabla_y\tilde{f}(x,y^\star) = \zeros_n$ (resp. $\nabla_y\widetilde{E}(x,y^\star) = \zeros_n)$ and solving for $y^\star$. There is a hyperplane of possible solutions due to the rank deficiency of $L$, but we are looking for the unique solution in $\Y$. The second term therefore follows from~\eqref{eq:ydom}.
	The fact that this point is also the unique equilibrium in $\Y$ follows from the fact that $\dot{y} = -\alpha\nabla_y \widetilde{E}(x,y^\star)$.
\end{proof}

\begin{lem}\longthmtitle{Finite Equilibria (Distributed)}~\label{lem:finite-eq-dist}
	Let $\widetilde{\X}\times\widetilde{\Y}$ be the set of equilibria of~\eqref{eq:binpad}
	satisfying $(\dot{x},\dot{y}) = 0$ on $(x,y)\in[0,1]^n \times \Y$. The set $\widetilde{\X}\times\widetilde{\Y}$ is finite.
\end{lem}
\begin{proof}
	The proof follows closely to the proof of Lemma~\ref{lem:finite-eq} with the variation that $\tilde{v}$ in the expression for $\dot{x}$ is now a function of $y$. Given the result of Lemma~\ref{lem:y-soln}, we may directly substitute the unique $y^\star$~\eqref{eq:ystar} for any $x$. Because $y^\star$ is simply a linear expression in $x$, the same argument as in Lemma~\ref{lem:finite-eq} that $\widetilde{\X}$ is finite follows. 
\end{proof}

We now extend the results of Theorem~\ref{thm:cent-cvg} to the
distributed case of solving $\Pc 2$ via $\binpad$. 
We have the following theorem on the trajectories of
$(x(t),y(t))$
under~\eqref{eq:binpad}, which can be interpretted as establishing convergence to a local minimizer with probability one.

\begin{thm}\longthmtitle{Convergence of $\binpad$}\label{thm:dist-cvg}
  Given an initial condition $(x(0),y(0))\in(0,1)^n\times \real^n$,
 the trajectory $(x(t),y(t))$ under $\binpad$
  converges asymptotically to a critical point $(x^\star,y^\star)$ of
  $\widetilde{E}$. In addition, under Assumption~\ref{ass:init}, on
  the random choice of initial condition $x(0)$, and
  Assumption~\ref{ass:t-tau}, on the random choice of $T,\tau$,
  the probability that $(x(0),y(0))$ is in the set
    $\underset{\hat{x},\hat{y}}{\cup} \W^s (\hat{x},\hat{y})$, where
    $(\hat{x},\hat{y})$ is a saddle-point or local maximum of
    $\widetilde{E}$, is zero. Lastly, all local minima
  $(x^\star,y^\star)$ of $\widetilde{E}$ are globally optimal in $y$:
  $\widetilde{E}(x^\star,y)\geq \widetilde{E}(x^\star,y^\star),
  \forall y\in\real^n$.
\end{thm}
\begin{proof}
	The first part of the proof to establish convergence to a critical point follows
	from a similar argument to the proof of
	Theorem~\ref{thm:cent-cvg}. Differentiating $\widetilde{E}$ with
	respect to time gives:
	\begin{equation}\label{eq:dEtildedt}
	\begin{aligned}
	\dfrac{d\widetilde{E}}{dt} &= \begin{bmatrix}
	\dot{x} \\ \dot{y}
	\end{bmatrix}^\top \begin{bmatrix}
	\nabla_x \widetilde{E}(x,y) \\ \nabla_y \widetilde{E}(x,y)
	\end{bmatrix} = \begin{bmatrix}
	\dot{x} \\ \dot{y}
	\end{bmatrix}^\top 
	\begin{bmatrix}
	-\widetilde{W}x -\tilde{v} + g^{-1}(x)/\tau \\ -\alpha^{-1}\dot{y}
	\end{bmatrix} \\
	&= -\dot{x}^\top\diag{T/(x_i-x_i^2)_i}
	\vert \tilde{H}(x)\vert_m\dot{x} - \alpha^{-1}\dot{y}^\top \dot{y} < 0, \\
	& \qquad \qquad \qquad \dot{x} \neq 0 \ \text{or} \ \dot{y} \neq
	0, \quad (x,y)\in (0,1)^n \times \mathcal{Y}.
	\end{aligned}
	\end{equation}
	Thus, $\widetilde{E}$ monotonically decreases along the
	trajectories of $\binpad$. Given~\eqref{eq:dEtildedt}, we call again on the forward invariance property of the open hypercube for the distributed case via Lemma~\ref{lem:fwd-inv-d}, stated below, which verifies that $(x,y)\in (0,1)^n \times \mathcal{Y}$ at all times.
	
	Due to the deficiency induced by $L$, $\widetilde{E}$ is not radially unbounded in $y$ over all of $\real^n$, so we must be careful before applying the LaSalle Invariance Principle. Instead, define $\widetilde{E}$ only on $[0,1]^n\times \Y$ in consideration of Lemma~\ref{lem:y-domain}. Radial unboundedness in $\widetilde{E}$ is then obtained given any $y(0)$, and it follows that
	the trajectories converge to
	largest invariant set contained in $d\widetilde{E}/dt
	= 0$ per the LaSalle Invariance Principle~\cite{HKK:02}. This is the finite set of critical points of $\widetilde{E}$ per Lemma~\ref{lem:finite-eq-dist}, and so it additionally follows that we converge to a single critical point $(x^\star,y^\star)$.  
	
	Because $\widetilde{E}$ is convex in $y$, it follows that
	for any fixed $x$ there exist only local minima of
	$\widetilde{E}$ with respect to $y$. In consideration of this,
	we need only apply the Stable Manifold Theorem~\cite{JG-PH:83} to $x$. The
	argument for this develops similarly to the proof of
	Theorem~\ref{thm:cent-cvg}, and we conclude that the
	trajectories of $\binpad$ converge to a local minimizer
	$(x^\star,y^\star)$ of $\widetilde{E}$ with probability
	one. 
	
	The final part of the Theorem statement
	that $\widetilde{E}(x^\star,y) \geq \widetilde{E}(x^\star,y^\star),
	\forall y\in\real^n$ can also be seen from the convexity of $\widetilde{E}$
	in $y$ and applying the first-order condition of
	convexity:
	\begin{equation*}
	\widetilde{E}(x^\star,y) \geq \widetilde{E}(x^\star,y^\star) + (y-y^\star)^\top \nabla_y \widetilde{E}(x^\star,y^\star)
	\end{equation*}
	along with $\nabla_y\widetilde{E}(x^\star,y^\star) =
	\zeros_n$.
\end{proof}

\begin{lem}\longthmtitle{Forward Invariance of the Open Hypercube (Distributed)}\label{lem:fwd-inv-d}
	The set $(0,1)^n \times \Y$ is a forward-invariant set under the
	$\binpad$ dynamics~\eqref{eq:binpad}.
\end{lem}
\begin{proof}
	The forward invariance of $\Y$ is already established per its definition and Lemma~\ref{lem:y-domain}, but we must establish that the trajectories $y(t)$ remain bounded in order to apply the argument in Lemma~\ref{lem:fwd-inv-c} to the proof of Theorem~\ref{thm:dist-cvg}. Compute the Hessian of $\widetilde{E}$ with respect to $y$ as:
	\begin{equation*}
	\nabla_{yy}\widetilde{E} = \gamma L^2 \succeq 0.
	\end{equation*}
	Due to the connectedness of $L$, the eigenspace associated with the $n-1$ strictly positive eigenvalues of $\gamma L^2$ is parallel to $\Y$. Therefore, $\widetilde{E}$ is strictly convex in $y$ on this subspace, and it follows that $\widetilde{E}$ is bounded from below on $\Y$. Due to $d\widetilde{E}/dt \leq 0$~\eqref{eq:dEtildedt} and the continuity of $\widetilde{E}$ in $y$, it follows that $y(t)$ is bounded for all $t$. Given this, the argument from Lemma~\ref{lem:fwd-inv-c} applies to the trajectories $x(t)$, and the set $(0,1)^n \times \Y$ is forward invariant under $\binpad$~\eqref{eq:binpad}.
\end{proof}

\begin{rem}\longthmtitle{One-Hop Distributed
    Algorithm}\label{rem:one-hop} \rm The proposed distributed algorithm requires
  two-hop neighbor information, which may be intractable in some
  settings. The source of the two-hop term stems from the quadratic
  $\gamma$ penalty term. However, it is possible to define a one-hop
  distributed algorithm via a Lagrangian-relaxation route.
	
  Consider posing $\Pc 2$ with the $\gamma$ term instead as a linear
  constraint: $\sqrt{\gamma /2} ((p_i x_i)_i + Ly) = \sqrt{\gamma /2}
  (\subscr{P}{r}/n)\ones_n$. Applying Lagrangian
    relaxation to this problem introduces a Lagrange multiplier on
  the linear terms, and from there it would be appropriate to define a
  saddle-point-like algorithm along the lines
  of~\cite{AC-EM-SHL-JC:18-tac} in which gradient-ascent in the dual
  variable is performed. This changes the nature of the penalty from
  squared to linear, so the underlying optimization model is different
  in that sense, but it follows that this approach could be
  implemented with one-hop information.
  
  We note that, in some distributed contexts, penalty terms or constraints can be imposed via $\sqrt{L}$ which then appears as $L$ in the associated squared terms of the dynamics (in place of $L^2$). However, the linear $L$ also appears in our algorithm, and substituting $\sqrt{L}$ would not inherit the sparsity of the communication graph. Therefore we leave the design of a fully one-hop mixed first-order/second-order algorithm as an open problem. 
\end{rem}

\section{Simulations}\label{sec:sims}

Our simulation study is split in to two parts; the first focuses on 
numerical comparisons related to runtime and solution quality, and
the second is a 2D visualization of the trajectories of the
Distributed Annealing (DA) variants for both the centralized and
distributed NNN methods.

\subsection{Runtime and Solution Quality Comparison}\label{ssec:runtime-soln}

In this section, we compare to a greedy method stated as
Algorithm~\ref{alg:greedy} and a semidefinite programming (SDP)
relaxation method stated as Algorithm~\ref{alg:sdp}. In short, the
greedy method initializes the state as $x = \zeros_n$ and iteratively
sets the element $x_i$ to one which decreases the cost function the
most. This is repeated until no element remains for which the updated
state has lower cost than the current state. For the SDP method, a
convex SDP is obtained as the relaxation of $\Pc 1$, see
e.g.~\cite{LV-SB:96}. We use the shorthand $\texttt{SDPrlx}(\bullet)$
to indicate this in the statement of Algorithm~\ref{alg:sdp}. This SDP is solved
using CVX software in MATLAB~\cite{website:cvx} and a lowest-cost
partition is computed to construct a feasible solution. For the
sake of convenience in stating both algorithms, we have defined $f':
2^n\rightarrow\real$ to be the set function equivalent of $f$,
i.e. the cost of $\Pc 1$. That is, $f'(\SSS) = f(x)$, where $i\in\SSS$
indicates $x_i = 1$ and $i\notin\SSS$ indicates $x_i = 0$. Finally, we additionally compare to a brute force method which we have manually programmed as an exhaustive search over the entire (finite) feasibility set.

\begin{algorithm}
	\caption{Greedy Method}\label{alg:greedy}
	\begin{algorithmic}[1]
		\Procedure{Greedy}{$f'$}
		\State $\SSS \gets \emptyset$
		\State $\texttt{done} \gets \texttt{false}$
		\While{$\texttt{done} = \texttt{false}$}
		\State $i^\star \gets \underset{i\notin\SSS}{\argmin} \ f'(\SSS \cup \{i\})$
		\If{$f'(\SSS\cup \{i^\star\}) < f'(\SSS)$}
		\State $\SSS \gets \SSS \cup \{i^\star\}$
		\Else
		\State $\texttt{done} \gets \texttt{true}$
		\EndIf
		\EndWhile
		\State $x_i \gets \begin{cases}
		0, & i\notin \SSS, \\
		1, & i\in \SSS.
		\end{cases}$
		\State \textbf{return} $x$
		\EndProcedure
	\end{algorithmic}
\end{algorithm}

\begin{algorithm}
	\caption{SDP Relaxation Method}\label{alg:sdp}
	\begin{algorithmic}[1]
		\Procedure{SDP}{$f'$}
		\State $\P_{\text{SDP}} \gets \texttt{SDPrlx}(\Pc 1)$
		\State $x^\star \gets \underset{x}{\argmin} \P_{\text{SDP}}$
		\State $\SSS \gets \emptyset$
		\State $\texttt{done} \gets \texttt{false}$
		\While{$\texttt{done} = \texttt{false}$}
		\State $i^\star \gets \underset{i\notin\SSS}{\argmax} \ x_i$
		\If{$f'(\SSS\cup \{i^\star\}) < f'(\SSS)$}
		\State $\SSS \gets \SSS \cup \{i^\star\}$
		\Else
		\State $\texttt{done} \gets \texttt{true}$
		\EndIf
		\EndWhile
		\State $x_i \gets \begin{cases}
		0, & i\notin \SSS, \\
		1, & i\in \SSS.
		\end{cases}$
		\State \textbf{return} $x$
		\EndProcedure
	\end{algorithmic}
\end{algorithm}

%

In Figure~\ref{fig:runtime} 
we plot the runtime in
MATLAB on a 3.5GHz Intel Xeon E3-1245 processor over increasing problem size
$n$ for each of six methods: a brute force search, the aforementioned
greedy and SDP methods, the HNN first proposed in~\cite{JH-DT:85}
(i.e. the gradient-like version of $\binpac$), and the $\binpac$ and
$\binpad$ methods we developed in Sections~\ref{sec:all-hop}
and~\ref{sec:dist-hop}. The first obvious observation to make is that
the runtime of brute force method increases at a steep exponential
rate with increasing $n$ and exceeds 120 seconds at $n=22$, making it
intractable for even medium sized problems. Next, we note that there
are some spikes associated with the HNN method around $n=25$
to $n=40$. These are reproducible, and we suspect that this is due to
the emergence of saddle-points and increasing likelihood of
encountering these along the trajectory as $n$ increases. This is a well-documented problem observed in literature, see e.g.~\cite{YD-RP-CG-KC-SG-YB:14}, and we also confirm it empirically in this setting by observing that share of iterations for which the Hessian is indefinite (as opposed to positive definite) tends to grow as $n$ increases. We also note
that $\binpac$ scales relatively poorly, which can be attributed to a
matrix eigendecomposition being performed at each discretized
iteration of the continuous-time algorithm. For $\binpad$, the matrix
being eigendecomposed is diagonal, which makes it a trivial operation
and allows $\binpad$ to scale well. We note that the SDP method scales the worst amongst the non brute-force methods. Unsurprisingly, the
greedy method remains the fastest at large scale, although recall that
the motivation of developing our method is for it to be distributed
and that a greedy approach can not be distributed due to the global penalty term.

As for algorithm performance as it pertains to the cost of the
obtained solution, we fix $n=50$ and additionally include DA variants of both $\binpac$ and $\binpad$. We also omit the brute force method due
to intractability. For the sake of comparison, we compute a
performance metric $Q$ and provide it for each method in
Table~\ref{table:perf}.
The metric $Q$ is computed
as follows: for each trial, sort the methods by solution cost. Assign
a value of 6 for the best method, 5 for the second-best, and so on,
down to the seventh-best (worst) receiving zero. Add up these scores
for all 100 trials, and then normalize by a factor of 600 (the maximum
possible score) to obtain $Q$. Note that $Q$ does not account for runtime in any way.

It should be unsurprising that the
tried-and-true centralized greedy and SDP methods perform the
best. However, we note that they were beaten by our methods in a
significant number of trials, which can be seen by noting that a $Q$ score for
two methods which perform best or second-best in all trials would sum
to $1100/600 = 1.83$, while $Q(\text{greedy})+Q(\text{SDP})=1.75$, or a
cumulative pre-scaled score of $1050$, indicating that our methods
outperformed these methods in net 50 ``placement spots" over the 100
trials. In general, we find that the DA version of the NNN algorithms
obtains better solutions than the non-DA version, confirming the
benefit of this approach. We also find that 
$\binpad$ generally outperforms $\binpac$.
It's possible that an initially ``selfish" trajectory in $x$ is beneficial, which would
neglect the global penalty until $y$ adequately converges, although this is speculative. Lastly,
we note that the HNN method never performs better than worst, which we
attribute to the steepest-descent nature of gradient algorithms which
do not use curveature information of the energy function. It might be
possible that the stopping criterion forces HNN to terminate near
saddle-points, although we do not suspect this since we observe the
Hessian is positive-definite in the majority of termination
instances.

As for parameter selection, we find that choosing $m \ll 1$ is generally
best, since $m\geq 1$ would always produce a PT-inverse Hessian with
eigenvalues contained in $(0,1]$. This effectively scales down
$\dot{x}$ in the eigenspace associated with Hessian eigenvalue
magnitudes greater than $1$, but does not correspondingly scale up
$\dot{x}$ in the complementary eigenspace associated with small
eigenvalues. Additionally, choosing $T/\tau$ greater than
$1$ in the fixed case 
tended to be
effective. This may be related to selecting $a_i < -\gamma\Vert
p\Vert^2 - 4T/\tau$ to guarantee anti-stability from $(0.5)\ones_n$,
and would explain why a high $T_0/\tau_0$ that decreases in the
DA learning variant performs so well. In general, for the DA learning variant, we recommend choosing $T_0,\tau_0$ so that $T_0/\tau_0 \gg 1$ and also $\beta>1$ sufficiently large so that $T/\tau \ll 1$ by algorithm termination, which gives rise to a robust exploration/exploitation tradeoff. Finally, all $\alpha\approx 1$
seem to behave roughly the same, with only $\alpha \ll 1$ and
$\alpha \gg 1$ behaving poorly (the former leading to slow convergence
in $y$ and ``selfish" behavior in $x$, and the latter being
destabilizing in the discretization of $\dot{y}$).


\begin{table}[]
	\centering
	\caption{Comparison of performance metric $Q$ for $100$ randomized trials with $n=50$.} 
	\label{table:perf}
	\scalebox{1}{
		\begin{tabular}{|c|c|c|c|c|}
			\hline
			Method & $Q$ \\ \hline \hline
			$\binpac$ & 0.2161 \\ \hline
			$\binpac$-DA & 0.2891 \\ \hline
			$\binpad$ & 0.5443 \\ \hline
			$\binpad$-DA & 0.7005 \\ \hline
			HNN & 0 \\ \hline
			Greedy & 0.8411 \\ \hline
			SDP & 0.9089 \\ \hline
	\end{tabular}}
\end{table}

\begin{table}[]
	\centering
	\caption{Problem data and parameter choices (where relevant) for performance comparison. Problem data $p_i, c_i$ is generated randomly from given distributions for each of $100$ trials.}
	\label{table:data}
	\scalebox{1}{
		\begin{tabular}{|c|c|c|c|c|}
			\hline
			Data or parameter & Value \\ \hline \hline
			$n$ & $50$ \\ \hline
			$p_i$ & $\U[1,50]$ \\ \hline
			$c_i$ & $p_i^e$, $e\sim\U[2,3]$ \\ \hline
			$\subscr{P}{r}$ & $1500$ \\ \hline
			$\gamma$ & $1$ \\ \hline
			$T_0$ & $1$ \\ \hline
			$\tau_0$ & $0.1$ \\ \hline
			$m$ & $0.1$ \\ \hline
			$\alpha$ & $1$ \\ \hline
			Learning steps & $10$ \\ \hline
			$\beta$ & $1.4$ \\ \hline
			$n$ & $50$ \\ \hline
	\end{tabular}}
\end{table}

\begin{figure}[h]
	\centering
	\includegraphics[scale = 0.6]{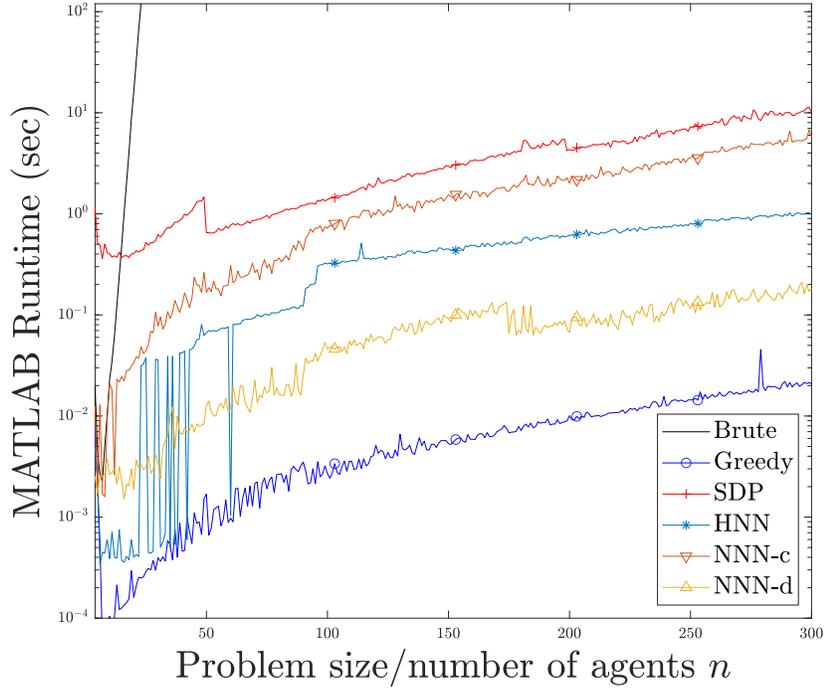}
	\caption{Runtime of each method for increasing problem sizes.}
	\label{fig:runtime}
\end{figure}

\subsection{Learning Steps and 2-D Trajectories}\label{ssec:traj-da}

Next, for the sake of understanding how the learning rate $T/\tau$ affects the trajectories of the solutions, we have provided Figure~\ref{fig:traj} which plots the 2-D trajectories of $\binpac$ and $\binpad$ with $T/\tau$ being gradually reduced over 15 learning steps. The contours of the energy function for the final step are also plotted. The problem data and choice for $a$ is:
\begin{equation*}
\begin{aligned}
c &= (2,1)^\top, \quad p = (3,1)^\top, \quad \subscr{P}{r} = 2.8, \quad \gamma = 4, \quad 
a = -(10,10)^\top. 
\end{aligned}
\end{equation*}
Note that, in each case, the trajectory approaches the optimal solution $x^\star = (1,0)^\top$. However, it is worth noting that a steep saddle point occurs around $x=(0.75,0.6)^\top$. Intuitively, this corresponds to a high risk of the trajectory veering away from the optimal solution had the DA not been implemented. With the opportunity to gradually learn the curveature of the energy function, as shown by stabilization to successive equilibria marked by $\times$, each algorithm is given the opportunity to richly explore the state space before stabilizing to the optimal solution $(1,0)^\top$. Further studying the learning-rate $T/\tau$ and a more complete analysis of Algorithm~\ref{alg:da} and the parameter $\beta$ are subjects of future work.

\begin{figure}[h]
	\centering
	\subfloat[]{{\includegraphics[scale=.75]{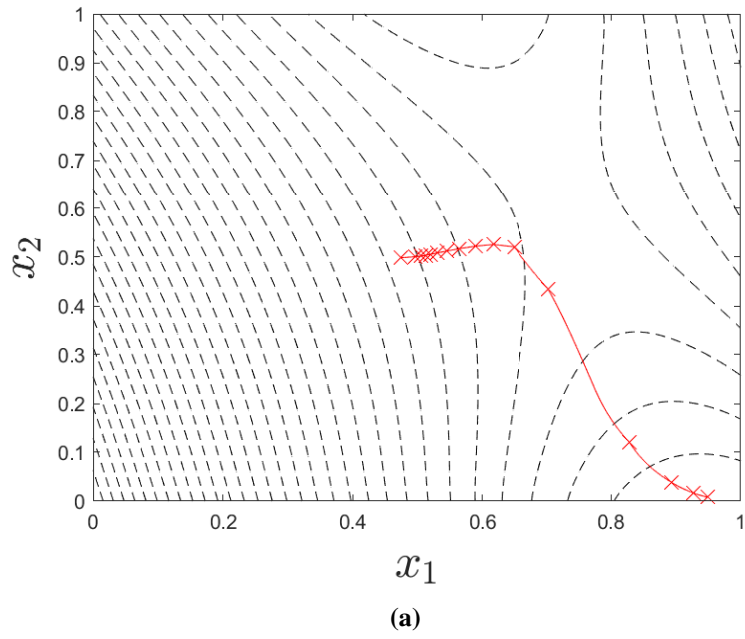} }}\label{fig:centnewt} \\
	\subfloat[]{{\includegraphics[scale=.75]{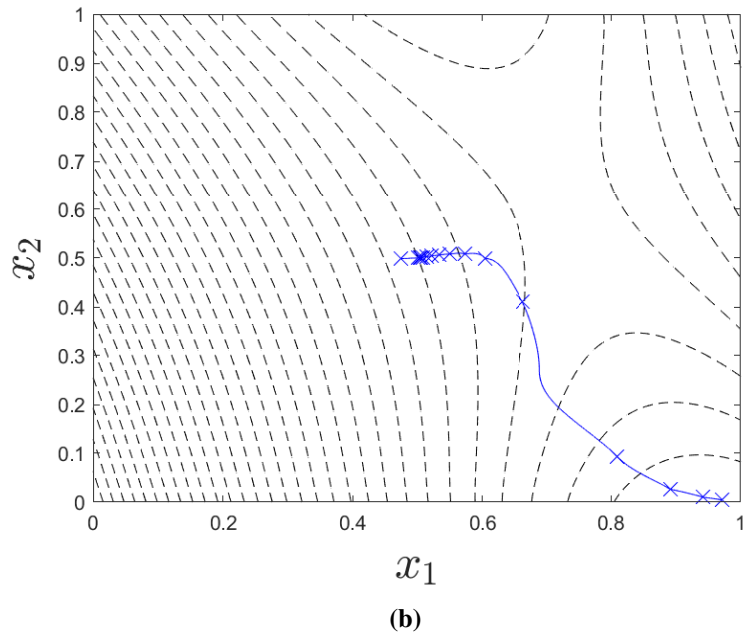} }}\label{fig:distnewt}
	\caption{Centralized $\binpac$ (a) and distributed $\binpad$ (b) trajectories in 2D with $15$ learning steps. Stable equilibrium points between learning steps indicated by $\times$, contours of $E$ and $\widetilde{E}$ in final step indicated by dashed lines.}
	\label{fig:traj}
\end{figure}

\section*{Acknowledgements}

The material in this chapter, in full, is provisionally accepted in Automatica. It is expected to appear as \textit{Distributed Resource Allocation with Binary Decisions via Newton-like Neural Network Dynamics}, T.~Anderson and S.~Mart\'{i}nez. The dissertation author was the primary investigator and author of this paper.

\include{chap_top}
\chapter{Frequency Regulation with Heterogeneous Energy Resources: A Realization using Distributed Control}
\label{chap:active_load}

This chapter presents one of the first real-life demonstrations of coordinated and distributed resource control for secondary frequency response in a 
power distribution grid. A series of tests involved up to 69 heterogeneous active distributed energy resources consisting of air handling units, unidirectional and bidirectional electric vehicle charging stations, a battery energy storage system, and 107 passive distributed energy resources consisting of building loads and solar photovoltaic systems. The distributed control setup consists of a set of Raspberry Pi end-points exchanging messages via an ethernet switch. Actuation commands for the distributed energy resources are obtained by solving a power allocation problem at every regulation instant using distributed ratio-consensus, primal-dual, and Newton-like algorithms. The problem formulation minimizes the sum of distributed energy resource costs while tracking the aggregate setpoint provided by the system operator. We demonstrate accurate and fast real-time distributed computation of the optimization solution and effective tracking of the regulation signal over 40-minute time horizons. An economic benefit analysis confirms eligibility to participate in an ancillary services market and demonstrates up to \$49k of potential annual revenue for the selected population of distributed energy resources.

The results of this chapter are the outcome of a project under the ARPA-e Network Optimized Distributed Energy Systems (NODES) program\footnote{\url{https://arpa-e.energy.gov/arpa-e-programs/nodes}},which postulates DER aggregations as virtual power plants that enable variable renewable penetrations of at least 50\%.
The vision of the NODES program was to employ state-of-the-art tools from control systems, computer science, and distributed systems to optimally respond to dynamic changes in the grid by leveraging DERs while maintaining customer quality of service. The NODES program required testing with at least 100 DERs at power. Here, we demonstrate the challenges and opportunities of testing on a heterogeneous fleet of DERs for eventual operationalization of optimal distributed control at frequency regulation time scales.

\section{Bibliographical Comments} To the best of our knowledge,  real-world testing of frequency regulation by DERs has been limited. A Vehicle-to-Grid (V2G) electric vehicle (EV)~\cite{WK-VU-KH-KK-SL-SB-DB-NP:08} and two Battery Energy Storage Systems (BESS)~\cite{MS-DS-AS-RT-RL-PCK:13} provided frequency regulation.
76 bitumen tanks were integrated with a simplified power system model to provide frequency regulation via a decentralized control algorithm in~\cite{MC-JW-SJG-CEL-NG-WWH-NJ:16}. In buildings, a decentralized control algorithm controlled lighting loads in a test room~\cite{JL-WZ-YL:17},  centralized frequency control was applied to an air handling unit (AHU)~\cite{YL-PB-SM-TM:15,EV-ECK-JM-GA-DSC:18}, an inverter and four household appliances~\cite{BL-SP-SA-MVS:18}, and four heaters in different rooms~\cite{LF-TTG-FAQ-AB-IL-CNJ:18}. A laboratory home with an EV and an AHU, and a number of simulated homes were considered for demand response in~\cite{KB-XJ-DV-WJ-DC-BS-JW-HS-ML:16} through an aggregator at a 10~s level. Technologies for widespread, but centrally controlled, cycling of air conditioners directly by utilities~cf.~\cite{SDGE} and aggregators are common place for peak shifting, but occur over time scales of minutes to hours. Industrial solutions enabling heterogeneous DERs to track power signals also exist, but they are either centralized, cf.~\cite{SC-PA:16} or require all-to-all communication~\cite{AT-SZ-SR:17}.

Our literature review exposes the following limitations: (i) centralized control or need for all-to-all communication~\cite{WK-VU-KH-KK-SL-SB-DB-NP:08,MS-DS-AS-RT-RL-PCK:13,YL-PB-SM-TM:15,EV-ECK-JM-GA-DSC:18,BL-SP-SA-MVS:18,LF-TTG-FAQ-AB-IL-CNJ:18,KB-XJ-DV-WJ-DC-BS-JW-HS-ML:16,SDGE,SC-PA:16,AT-SZ-SR:17}, which does not scale to millions of DERs; (ii) small numbers of DERs~\cite{WK-VU-KH-KK-SL-SB-DB-NP:08,MS-DS-AS-RT-RL-PCK:13,YL-PB-SM-TM:15,EV-ECK-JM-GA-DSC:18,BL-SP-SA-MVS:18,LF-TTG-FAQ-AB-IL-CNJ:18,KB-XJ-DV-WJ-DC-BS-JW-HS-ML:16}; (iii) lack of diversity in DERs~\cite{WK-VU-KH-KK-SL-SB-DB-NP:08,MS-DS-AS-RT-RL-PCK:13,MC-JW-SJG-CEL-NG-WWH-NJ:16,JL-WZ-YL:17,YL-PB-SM-TM:15,EV-ECK-JM-GA-DSC:18,LF-TTG-FAQ-AB-IL-CNJ:18}, with associated differences in tracking time scales and accuracy. No trial has been reported that demonstrated generalizability to a real scenario with (i) scalable distributed control and a (ii) large number of (iii) heterogeneous DERs.

\section*{Statement of Contributions}
To advance the field of real-world testing of DERs for frequency control, we conduct a series of tests using a group of up to 69 active and 107 passive heterogeneous DERs on the University of California, San Diego (UCSD) microgrid~\cite{BW-JD-DW-JK-NB-WT-CR:13}. To the best of the authors' knowledge, this is the first work to consider such a large, diverse portfolio of real physical DERs for secondary frequency response. As such, the major contributions of this work are:
\begin{itemize}
    \item A detailed account of the testbed, including the DER actuation and sampling interfaces, the distributed optimization setup, and communication framework.
    \item A description of techniques to work around technical barriers, provision of lessons learned, and suggestions for future improvement. 
    \item Evaluation of the performance of both the cyber and physical layers, including an evaluation of eligibility requirements for and the economic benefit of participating in the ancillary services market.
\end{itemize}

\textit{Chapter Overview.} 
Frequency regulation is simulated on the UCSD microgrid using real controllable DERs (Section \ref{ders}) to follow the PJM RegD signal~\cite{PJM-signal:19} interpolated from 0.5Hz to 1Hz (Sections \ref{regulation-signal}). The DER setpoint tracking is formulated as a power allocation problem at every regulation instant (Section \ref{optimization-statements}), and uses three types of provably convergent distributed algorithms from~\cite{ADDG-CNH-NHV:12,AC-JC:16-allerton,AC-BG-JC:17-sicon,TA-CYC-SM:18-auto} to solve the optimization problem; see the Appendix.
Setpoints are computed distributively on multiple Raspberry Pi's communicating via ethernet switches (Section \ref{computing-setup}). The setpoints are implemented on up to 176 DERs at power using dedicated command interfaces via TCP/IP communication (Section \ref{actuation-interface}), the DER power outputs monitored (Section \ref{power-measurements}), and their tracking performance evaluated (Section \ref{error-metrics}). Results for the various test scenarios (Section \ref{test}) show that the test system tracks the signal with reasonable error despite delays in response and inaccurate tracking behavior of some groups of DERs, and qualifies for  participation in the PJM ancillary services market (Section~\ref{results}).

\section{Problem Setting}
This chapter validates real-world DER controllability for participation in secondary frequency regulation through demonstration tests implemented on a real distribution grid. The tests showcase the ability of aggregated DERs to function as a single market entity that responds to frequency regulation requests from the independent system operators (ISO) by optimally coordinating DERs. The goal is to monitor and actuate a set of real controllable DERs to collectively track a typical automatic generation control (AGC) signal issued by the ISO. 

Three different distributed coordination schemes optimize the normalized contribution of each DER to the cumulative active power signal. Unlike simulated models, the use of real power hardware exposes implementation challenges associated with measurement noise, sampling errors, data communication problems, and DER response. To that end, precise load tracking is pursued at timescales that differ by DER type consistent with individual DER responsiveness and communication latencies, yet meet frequency regulation requirements in aggregation. 

The 69 kV substation and 12 kV radial distribution system owned by UCSD to operate the 5~km$^2$ campus was the chosen demonstration testbed. It has diverse energy resources with real-time monitoring and control capabilities, allowing for active load tracking. This includes over 3~MW of solar photovoltaic (PV) systems, 2.5~MW/5~MWh of BESS, building heating ventilation and air conditioning (HVAC) systems in 14 million square feet of occupied space, and over 200 unidirectional V2G (V1G) and V2G EV chargers. The demonstration tests used a representative population of up to 176 such heterogeneous DERs to investigate tracking behavior of specific DER types as well as their cooperative tracking abilities.  While the available DER capacity at UCSD far exceeds the minimum requirements for an ancillary service provider set by most ISOs (typically $\sim$ 1~MW), logistical considerations and controller capabilities dictated the choice of a DER population size with less aggregate power capacity (up to 184~kW) for this demonstration. Since this magnitude of power is insufficient to measurably impact the actual grid frequency, we chose to simulate frequency regulation by following a frequency regulation signal.

\section{Test Elements}\label{sec:elements}

Here, we elaborate on the different elements of the validation tests. 
These include the optimization formulation employed to compute  DER setpoints (Section~\ref{optimization-statements}), the reference AGC signal (Section~\ref{regulation-signal}) and types of DERs used to track it (Section~\ref{ders}), the computing platform (Section \ref{computing-setup}), the actuation  (Section~\ref{actuation-interface}) and monitoring interfaces (Section~\ref{power-measurements}),  the performance metrics used to assess the cyber and physical layers, and eligibility for market participation (Section \ref{error-metrics}).

\subsection{Optimization Formulation}\label{optimization-statements}
The optimization model for AGC signal tracking using DERs can be mathematically stated as a separable resource allocation problem subject to box constraints as follows: 

\begin{equation}\label{eq:opt}
\begin{aligned}
\underset{p\in\real^n}{\text{min}} \
&f(p) = \sum_{i=1}^n f_i(p_i), \\
\text{s.t.} \ &\sum_{i=1}^n p_i = \Pref, \\
&p_i\in [\pu_i, \po_i], \quad \forall i\in\N = \{1,\dots,n\}.
\end{aligned}
\end{equation}

The agents $i\in\N$ each have local ownership of a decision variable $p_i\in\real$, representing an active power generation or consumption quantity (setpoint), a local convex cost function $f_i$, and local box constraints $[\pu,\po ]$, representing active power capacity limits. $\Pref$ is a given active power reference value determined by the ISO and transmitted to a subset of the agents as problem data, see e.g.~\cite{CAISO:18}. $\Pref$ is a signal that changes over time, so a new instance of~\eqref{eq:opt} is solved in 1~s intervals corresponding to these changes. 

For the validation tests, we used two types of cost functions: constant and quadratic. Constant functions were used for the Ratio-Consensus (RC) solver, which turns the optimization into a feasibility problem. Quadratic functions were used for the primal-dual based (PD) and Distributed Approximate Newton Algorithm (DANA) methods, see the Appendix.
The quadratic functions were artificially chosen to produce satisfactorily diverse and representative solutions for each DER population. We split the total time period of the signal, $\Pref$ into three equal segments, and implemented RC, PD, and DANA in that order. Box constraints $[\pu_i,\po_i ]$ were typically centered at zero for simplicity, see Section~\ref{ders}.

\subsection{Regulation Signal 
}\label{regulation-signal}
The 40~min RegD signal published by PJM~\cite{PJM-signal:19} served as the reference AGC signal for the validation tests, and was used to obtain the value for $\Pref$ in~\eqref{eq:opt}. The normalized RegD signal, contained in $[-1,1]$, was interpolated from 0.5~Hz to 1~Hz. The signal was then treated by subtracting the normalized contributions of building loads and PV systems, cf. Section~\ref{ders}. Finally, the normalized signal was scaled by a factor proportional to the total DER capacity $\sum_i (\po_i - \pu_i)$ before sending to the optimization solvers. More precisely,
\begin{equation}\label{eq:norm-sig}
    \Pref = \beta \frac{\sum_i (\po_i - \pu_i)}{\| P_\text{RegD} + P_\text{PV} - P_\text{b} \|_{\infty} }\left(P_\text{RegD} + P_\text{PV} - P_\text{b}\right),
\end{equation}
where $P_\text{RegD}$ refers to the normalized RegD signal data, $P_\text{PV}$ and $P_\text{b}$ respectively refer to the normalized PV generation and building load data obtained from the UCSD ION server as described in Section \ref{power-measurements}, and $0 < \beta < 1$ is an arbitrary scaling constant. For most test scenarios, $\beta = 0.75$ to prevent extreme set points that would require all DERs to operate at either $\po_i$ or $\pu_i$ simultaneously, which may be infeasible in some time steps due to slower signal update times, see Table~\ref{table:devices}. Each $P$ in~\eqref{eq:norm-sig} is a vector with 2401 elements corresponding to each 1~s time step's instance of~\eqref{eq:opt} over the 40~min time horizon. The acquired target regulation signal is characterized by steep positive and negative ramps that range from -14~kW to +16~kW over 1~s intervals and an average absolute ramp-rate of 1.7 kW/s. 

\subsection{DERs}\label{ders}
The reference AGC signal was to be collectively tracked using DERs consisting of HVAC AHUs, BESS, V1G and V2G EVs, PV systems, and whole-building loads. Since PV systems and (non-AHU) building loads were not controllable, they participated in the test as passive DERs. Consequently, the active DERs were commanded to track a modified target signal derived by subtracting the net active power output of passive DERs from the reference AGC signal and applying appropriate scaling (cf. Section \ref{regulation-signal}). Table~\ref{table:devices} lists the typical net power capacity $\po_i - \pu_i$ of the different active DER types.

\begin{table}[htb]
\centering
\caption{DER counts and characteristics for each test.}\label{table:devices}
\begin{tabular}{|c|c|c|c|c|}
\hline
\textbf{DER Type}  & \textbf{AHU} & \textbf{V1G EV} & \textbf{V2G EV} & \textbf{BESS} \\ \hline
\makecell{\textbf{\# DERs for} \\ \textbf{Test 0}}  & 7 & 4  & 5 & 1 \\ \hline
\makecell{\textbf{\# DERs for} \\ \textbf{Test 1}} & 34 & 29 & 5 & 1 \\ \hline
\makecell{\textbf{\# DERs for} \\ \textbf{Test 2}} & 34 & 17 & 6 & 1 \\ \hline
\makecell{\textbf{Signal update}\\ \textbf{ times}} & 1 min & \makecell{5 min  (Test 0 \& 1), \\ 1 min  (Test 2)} & 1 sec & 20 sec \\ \hline
\makecell{\textbf{Typical power} \\ \textbf{rating per DER type}} & 2~kW & \makecell{3.3~kW  (Test 0 \& 1), \\ 4.9~kW  (Test 2)} & 5~kW & 3~kW \\ \hline
\end{tabular}
\end{table}

The contribution of each active DER to the target signal was defined with respect to a baseline power, around which $[\pu_i,\po_i ]$ was centered, to enable tracking of both positive and negative ramps in the target signal. For DERs like V2G EVs and BESS, which were capable of power adjustments in both directions, the baseline was 0~kW. The baseline for V1G EVs was defined to be halfway between their allowed minimum and maximum charging rates, where the former was restricted by the SAE J1772 charging standard to 1.6~kW. Similarly, the baseline for AHUs was defined to be half of their power draw when on. Further, since AHUs were limited to binary on-off operational states, the continuous and arbitrarily precise AHU setpoints obtained by solving \eqref{eq:opt} were rounded to the closest discrete setpoint obtained from a combination of on-off states before actuation.

AHU control was restricted, by UCSD Facilities Management, to specifying only DER setpoints and duration of actuation; since building automation controllers could not be modified, model-based designs were impossible. This was to avoid malfunctioning or disruptions to real physical infrastructure in the networked building management system that also controls lighting, security, and fire protection systems. 

\begin{figure*}
\centering 
\includegraphics[scale=0.33]{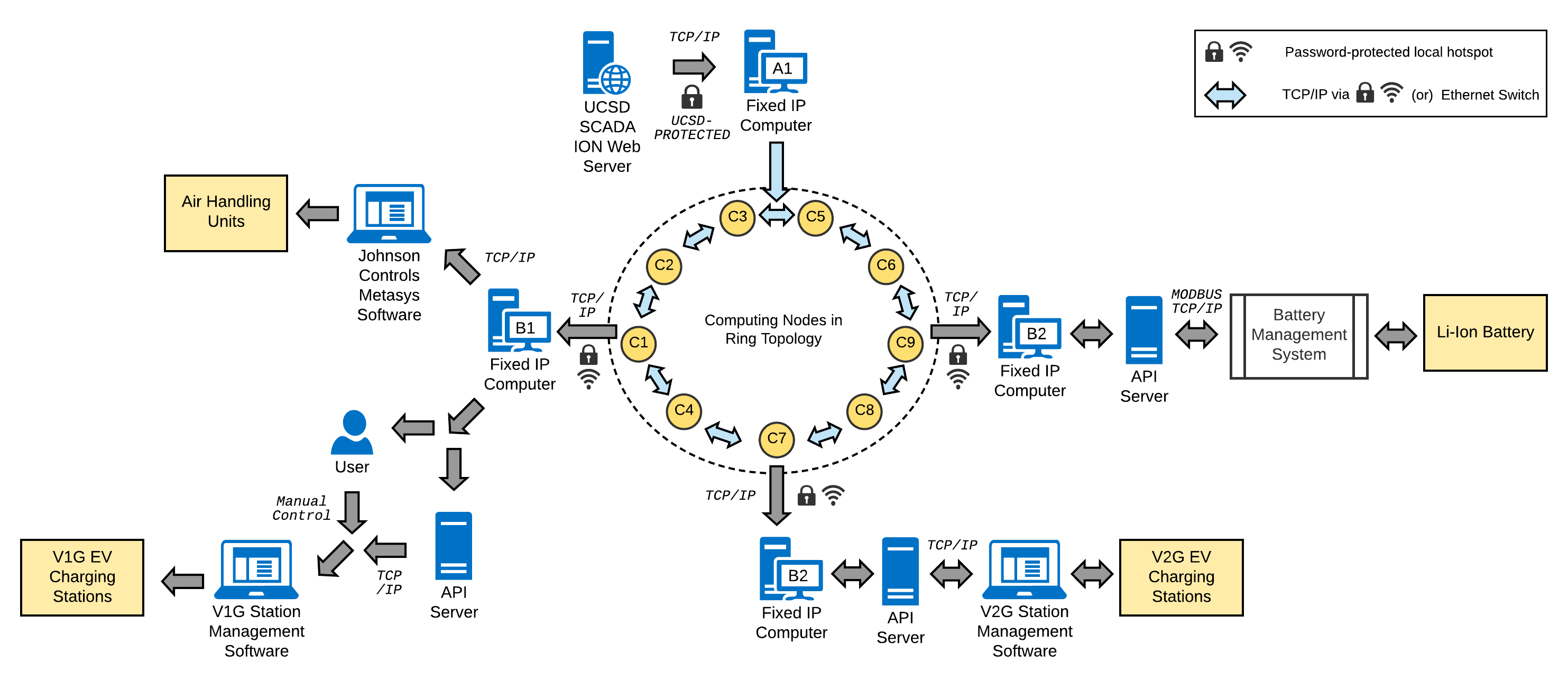}
\caption{Communication architecture for computation and actuation of control policies.}\label{fig:network-diagram}
\end{figure*}

\subsection{Computing Setup}\label{computing-setup}
The DER active power setpoints were computed using a set of 9 Linux-based nodes, named C1-C9, that communicate with each other over an undirected ring topology, cf. Fig.~\ref{fig:network-diagram}. 
As one of the sparsest
network topologies, where message passing occurs only between a small number of neighbors, the ring topology presents a challenging scenario for distributed control.
Since there were more active DERs than computing nodes, the 9 nodes were mapped subjectively to the 69 active DERs such that nodes C1-C2 computed the actuation setpoints for the AHUs, C3 for V1G EVs, C4-C8 for V2G EVs and C9 for the BESS. 

Each computing node generated actuation commands as CSV files containing the power setpoints for their respective group of DERs at a uniform update rate of 1~Hz. Preliminary testing revealed different response times across DER types, with AHUs and V1G EVs exhibiting slower response than other active DER types. DERs with response times greater than 1~s were subject to a stair-step control signal with a signal update time consistent with DER responsiveness and constant setpoints during intermediate time steps. Table~\ref{table:devices} lists the signal update times for the different DER types. 

\subsection{Actuation Interfaces and Communication Framework}\label{actuation-interface}
The actuation commands were issued using fixed IP computers through dedicated interfaces that varied by DER type as depicted in Fig.~\ref{fig:network-diagram}. The setpoints for AHUs were issued through a custom Visual Basic program that interfaced with the Johnson Control Metasys building automation software. The power rate of the BESS was set via API-based communication with a dedicated computer that controlled the battery inverter. The V1G and V2G EVs charging rates were adjusted through proprietary smart EV charging platforms of the charging station operators. EVs using ChargePoint\textsuperscript{\tiny\textregistered} V1G stations were manually controlled via the load shedding feature of ChargePoint’s station management software. The actuation of EVs using PowerFlex\textsuperscript{\tiny\textregistered} V1G chargers and Nuvve\textsuperscript{\tiny\textregistered} V2G chargers was automated and commands were issued via API-based communication.

\subsection{Power Measurements}\label{power-measurements}
The active power of all DERs was metered at a 1~Hz frequency. The power outputs of PV systems and building loads were obtained prior to the test from their respective ION meters by logging data from the UCSD ION Supervisory Control and Data Acquisition (SCADA) system. A moving average filter with a 20~s time horizon was used to remove noise from the measured data for these passive DERs. V2G EVs and BESS power data were acquired using the same interfaces that were used for their actuation, which logged data from dedicated power meters.

Since neither AHUs nor the ChargePoint V1G EVs had dedicated meters, they were monitored via their respective building ION meters by subtracting a baseline building load from the building meter power output. Assuming constant baseline building load, any change in the meter outputs can be attributed to the actuation of AHUs and V1G EVs. This assumption is justifiable considering the tests were conducted at 0400 PT to 0600 PT on a weekend, when building occupancy was likely zero and building load remained largely unchanged. Noise in the ION meter outputs observed as frequent 15~-~30~kW spikes in the measured data for AHUs (Fig.~\ref{fig:trial0trial1}) and ChargePoint V1G EVs was treated by removing outliers and passing the resulting signal through a 4~s horizon moving average filter. Here, outliers refer to points that change in excess of 50\% of the mean of the   40~min signal in a 1~s interval.

\subsection{Performance Metrics}\label{error-metrics}
The performance of the distributed implementation (cyber-layer) was measured by the normalized mean-squared-error (MSE) between the distributed and true (i.e. exact) centralized optimization solutions. The true solutions were computed for each instance of~\eqref{eq:opt} using a centralized CVX solver in MATLAB~\cite{website:cvx}. The MSE was normalized by dividing by the mean of the squares of the true solutions. 

The tracking performance of the DERs was evaluated through (i) the root-mean-squared-error (RMSE) in tracking 
\begin{equation}\label{eq:rmse-calc}
    \text{RMSE} = \sqrt{\frac{\sum_{t=1}^T (P_t^{\text{prov}}-P_t^{\text{tar}})^2}{\sum_{t=1}^T (P_t^{\text{tar}})^2}},
\end{equation}
where $P_t^{\text{prov}}$ is the total power that was provided (measured), and $P_t^{\text{tar}}$ is the target (commanded) regulation power at time step $t\in\{1,\dots,T=2401\}$; and (ii) the tracking delay, computed as the time shift of the measured signal which yields the lowest RMSE between the commanded and measured signals. 

The PJM Performance Score~$S$ following~\cite[Section 4.5.6]{PJM:20} was computed as a test for eligibility to participate in the ancillary services market, and is given by the mean of a Correlation Score~$S_c$, Delay Score~$S_d$, and Precision Score~$S_p$:
\begin{align*}
            S_c &= \frac{1}{T-1}\sum_{t=1}^T \frac{(P_t^{\text{prov}} - \mu^{\text{prov}})(P_t^{\text{tar}} - \mu^{\text{tar}})}{\sigma^{\text{prov}}\sigma^{\text{tar}}}, \\
        S_d &= \bigg\lvert \frac{\delta - 5 \text{ min}}{5 \text{ min}} \bigg\rvert, \quad
        S_p = 1 - \frac{1}{T} \sum_{t=1}^T \bigg\lvert \frac{P_t^{\text{prov}} - P_t^{\text{tar}}}{\mu^{\text{tar}}} \bigg\rvert, \\
        S &= 1/3(S_c + S_d + S_p),
        \end{align*}
where $P_t^{\text{prov}}$ and $P_t^{\text{tar}}$ are as in~\eqref{eq:rmse-calc},  $\mu^{\text{prov}}, \mu^{\text{tar}}$ and $\sigma^{\text{prov}}, \sigma^{\text{tar}}$ denote their respective means and standard deviations, and $\delta$ is the corresponding maximum delay in DER response for when $S_c$ was maximized. A performance score of at least 0.75 is required for participating in the PJM ancillary services market.

\section{Test Scenarios and Results}\label{test-description}
In this section, we describe the test scenarios carried out on the UCSD microgrid and present their outcome, elaborating on the challenges we faced and the differences across the tests.

\subsection{Test Scenarios}\label{test}
\subsubsection{Commonalities}
A series of three tests were conducted on December 12, 2018 (Test~0), April 14, 2019 (Test~1) and December 17, 2019 (Test~2). All three tests involved a 40~min preparatory run followed by a 40~min final test. Table~\ref{table:devices} lists the number and type of DERs used in each test. All tests were carried out during non-operational hours (between 0400 PT and 0540 PT)
to maximize fleet EV availability and to avoid potential disruptions to building occupants. Day-time PV output data from February 24, 2019 was used as a proxy for an actual daytime PV signal. 

\subsubsection{Test~0}\label{test-0}
Test~0 was a preliminary calibration that used only a representative sample of 17 DERs. The purpose of Test~0 was to examine the response times and tracking behavior of every DER type and detect issues related to communication and actuation.  

\subsubsection{Test~1}\label{test-1}

Test~1 was identical to Test~0, but it used a larger population of 69 active DERs and 107 passive DERs.  

a) \textit{DERs.} The V1G and V2G population for Test~1 was composed of UCSD fleet EVs plugged in at ChargePoint and Nuvve charging stations, respectively. Since the ChargePoint V1G EVs were operated via manual input of DER setpoints (an interface to their API had not been developed yet), to avoid overloading the (human) operators, they were grouped into three groups and actuated in a staggered fashion such that each of the three groups maintained a signal update time of 5~min but were commanded 1~min apart from each other.  

b) \textit{Computing Setup.} For both Tests~0 and 1, 9 laptops running a Robotic Operating System (ROS) communicated via local Wi-Fi hotspot to implement the distributed coordination algorithms and compute the DER setpoints. Given that the available power capacity of fast-responding DERs such as V2G and BESS was smaller than slow-responding DERs, the steep ramping demands of the target signal were met by upscaling the power of the fast responding DERs in solving for the contribution of individual DERs. Another option would have been to reduce the number of slow responding DERs, but the funding agency stipulated prioritizing the number and types of heterogeneous DERs over accuracy in signal tracking. A real DER aggregator would instead require a more balanced capacity of slow and fast DERs to ensure feasibility of tracking these ramp features.

\subsubsection{Test~2}\label{test-2}
Test~2 also used the entire population of DERs but substituted the cumbersome V1G population with more capable V1G chargers and used a new distributed computing setup and method of actuation based on lessons learned from Test~1. 

a) \textit{DERs.} The V1G EVs used in Test~1 performed poorly owing to an unreliable actuation-interface that experienced seemingly random stalling and lacked automated control capabilities. Therefore, 17 PowerFlex V1G charging stations at one location replaced the distributed 29 V1G charging stations in Test~1. Since the PowerFlex interface did not permit actuating individual stations, the 17 charging stations participated in the test as a single aggregate DER. The 0930 – 1010 PT timing of the V1G EV part of the test coincided with the start of the workday and a V1G EV population that had only recently plugged in and therefore had ample remaining charging capacity. The EVs were contributed by UCSD employees and visitors randomly plugging in at the PowerFlex charging stations just before the start of the trial. An aggregate signal of 15~kW to 19~kW was distributed equally amongst the 17 EVs.

In addition to the new V1G EVs, the V2G population in Test~2 was replaced with a different set of Nuvve chargers to resolve a tracking/noise issue during discharge-to-grid observed in Test~1 and expanded to include an additional charger, amounting to a total of six V2Gs charging six 5~kW EVs. 

The order of AHU actuation was modified to allow for device settling time and prevent interference. In particular, in Tests~0 and 1, individual AHUs were ordered and actuated using a protocol that was not cognizant of settling times or building groupings, while the protocol was revised in Test~2 to systematically command the entire population of AHUs in a manner which maximized time between consecutive actuations for an individual unit. 

b) \textit{Computing Setup.} Test~2 featured a fully distributed architecture, unlike the ROS-based semi-centralized computing setup in Test~1. The new distributed setup consisted of a network of Raspberry Pi’s that asynchronously communicated with each other via an ethernet switch. In addition, a modified synchronization technique was implemented in the software which improved the fidelity and robustness of message-passing. This upgraded message-passing framework and synchronization technique for both software and hardware resulted in significantly faster communication between nodes.

c) \textit{Two-Stage Actuation.} Test~2 also featured a two-stage approach of actuation that was a result of the DER tracking behavior in Test~1. Some DERs, such as BESS, V1G EVs and V2G EVs, tracked quickly and accurately, whereas others, such as AHUs, tracked poorly. The overall tracking performance in Test~2 was improved by using ``well-behaved" DERs to compensate for AHU tracking errors by incorporating the error signal from actuating AHUs in Stage 1 to the cumulative target signal for BESS, V1G EVs and V2G EVs in Stage 2. Although synchronous actuation of all participating DERs is preferred in practice, the two-stage approach highlights the significance of systematic characterization of DERs in  minimizing~ACE.

\subsection{Test Results} \label{results}
\subsubsection{Distributed Optimization/Cyber-Layer Results}

In Table~\ref{table:error}, we present MSE results of our 1~s real-time Raspberry-pi distributed optimization solutions (the ``cyber-layer'' of the system). 

\begin{table}[htb]
\centering
\caption{Normalized mean-squared-error of distributed solutions obtained from real-time 1-second intervals compared to centralized solver solution for Test~2 (Section \ref{error-metrics})}\label{table:error}
\begin{tabular}{|c|c|c|c|c|}
\hline
\textbf{DER Type}  & \textbf{RC} & \textbf{PD} & \textbf{DANA} & \textbf{all} \\ \hline
AHU & $0$ & $1.4\times 10^{-7}$ & $2.8\times 10^{-9}$ & $4.6\times 10^{-8}$ \\ \hline
V1G EVs & $0$ & $7.0 \times 10^{-8}$ & $1.7\times 10^{-9}$ & $2.3\times 10^{-8}$ \\ \hline
V2G EVs & $0$ & $6.6\times 10^{-5}$ & $5.0\times 10^{-7}$ & $2.1\times 10^{-5}$ \\ \hline
BESS & $0$ & $2.0\times 10^{-6}$ & $9.1\times 10^{-8}$ & $6.5\times 10^{-7}$ \\ \hline
Total & $0$ & $1.8\times 10^{-5}$ & $1.1\times 10^{-7}$ & $4.9\times 10^{-6}$ \\ \hline
\end{tabular}
\end{table}

RC converged to the exact solution in all instances. This is unsurprising, as the RC problem formulation does not account for individual DER costs and thus, is a much simpler problem with a closed-form solution. For PD and DANA, we obtained excellent convergence, with errors on the order of $0.001\%$ in the worst cases. In general, DANA tended to converge faster than PD and obtained more accurate solutions. For our application with 1~s real-time windows, accuracy and convergence differences did not affect the physical layer results in any tangible way, but applications with more stringent accuracy or speed requirements may benefit from using a faster algorithm like DANA. The differences between DER populations can be largely attributed to the faster time scale of the V2G EVs (and to a lesser extent the BESS), see Table~\ref{table:devices}. Since the V2G EVs were responsible for the high-frequency component of $\Pref$, the solver was required to converge to new solutions at every time step, which induced more error compared to the slow V1G EVs and AHUs with relatively static solutions.

\subsubsection{Physical-Layer Test Results}\label{ssec:phys-results}
We now present the results of the tracking performance pertaining to the physical-layer of the experiment. We provide only some selective plots for Test~0 and Test~1 in Fig.~\ref{fig:trial0trial1}, and a complete set of plots for each Test~2 DER population in Fig.~\ref{fig:trial2}. Error and tracking delay data defined in Section \ref{error-metrics} is given in Table~\ref{table:phys-error} for Test~1 and Test~2. Data for Test~0 is omitted due to its preliminary nature. The optimal shift described in Section \ref{error-metrics} is applied to each time series and hence some areas in plots may appear like the provided signal anticipated the target.

\begin{figure}[hbt!]
\centering 
\includegraphics[width=0.9\linewidth]{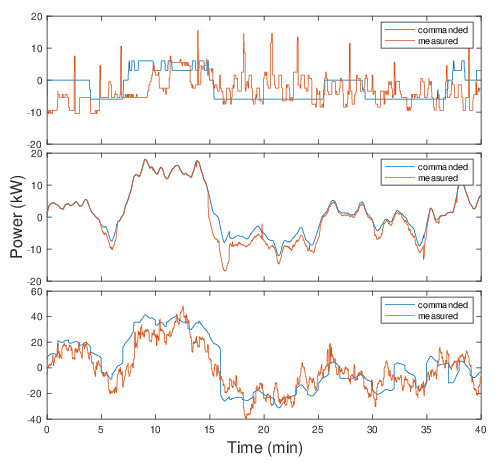}
\caption{\textbf{Top:} AHU response in Test~0. \textbf{Middle:} V2G response in Test~1. \textbf{Bottom:} Total response in Test~1.}\label{fig:trial0trial1}
\vspace*{-2ex}
\end{figure}

\begin{figure}[hbt!]
\centering 
\includegraphics[width=0.8\linewidth]{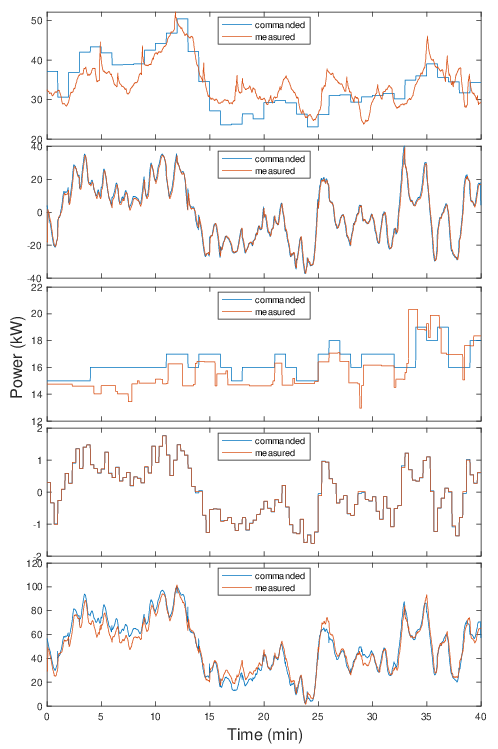}
\caption{From \textbf{top to bottom}, AHU, V2G EVs, V1G EVs, BESS, and total responses in Test~2.}\label{fig:trial2}
\vspace*{-2ex}
\end{figure}

Signal tracking accuracy in Test~0 was generally poor despite the small number of DERs employed, largely due to inexperience in actuating the AHUs and V1Gs. In particular, Fig.~\ref{fig:trial0trial1} reveals some oscillations in the AHU response. It is overall difficult to determine if even large-feature, low-frequency components of the signal were tracked. Further, data gathering for V1Gs and AHUs was done via noisy and unreliable building ION meters, which motivated the need for outlier treatment (Section~\ref{power-measurements}) in Tests~1 and~2, and resulted in the smoother and better tracking signal in the top plot of Fig.~\ref{fig:trial2}. 

Test~1 yielded a 
111\% rMSE
for AHUs. We speculate that the small 4~s delay is not representative of the actual AHU delay due to random correlations dominating the time shift for this large error. This is confirmed by a much better AHU response in Test~2 with rMSE 
12\%, where a 105~s delay is more likely to be representative of the true AHU actuation delay.
Given the poor visibility into AHU and V1G controllers explained in Section~\ref{test}, it is challenging to identify the source of the poor tracking behavior. We speculate that DER metering at the building level rather than the DER level was a major source of error for AHU and V1G in Test~1. This was largely resolved in Test~2 by utilizing a different population of V1Gs with dedicated meters and by modifying the actuation scheme for AHUs to be less susceptible to metering errors as described in Section~\ref{test-2}. Additionally, the actuation-interface stalling for V1G EVs, described in Section~\ref{test-1}, was dominant in Test~1, resulting in the poor tracking for V1Gs. Actuating-interface issues were resolved in Test~2 by utilizing an automated control scheme for the V1Gs, which led to significantly lower error.

The BESS emerged as the star performer achieving very accurate tracking across all tests with no delay. The V2G EVs also performed relatively well aside from a signal overshoot issue observed during the discharge cycle in Test~1 seen in Fig.~\ref{fig:trial0trial1}. The issue was resolved in Test~2 by using V2G EV charging stations from a different manufacturer (Princeton Power), as described in Section \ref{test-2}. The V2G charging stations deployed for these tests were pre-commercial or early commercial models that had a few operating issues, such as the overshoot issue during Test~1.

The inability of the AHUs to respond to steep, short ramps (Fig.~\ref{fig:trial2}) could be due to slow start-up sequences programmed into the building automation controllers to increase device longevity or due to transients associated with driving their AC induction electric motors. Tackling this would require dynamic models and parameter identification of signal response and delay. With the new V1G EV population in Test~2, tracking delay reduced from 40~s to 10~s and the tracking accuracy improved significantly. The 1~kW bias seen in Fig.~\ref{fig:trial2} is likely due to rounding errors arising from the inability of PowerFlex charging stations to accept non-integer setpoints.

The superior performance of the BESS and V2Gs motivated the two-stage actuation scheme described in Section~\ref{test-2}, which contributed to reducing the total RMSE from 
50\% in Test~1 to 
10\% in Test~2 (compare the bottom plots of Figs~\ref{fig:trial0trial1} and~\ref{fig:trial2}). The two-stage approach allows a sufficiently large proportion of accurately tracking DERs to compensate for the errors of the first stage, where tracking is worse. In this way, poorly-tracking DERs, such as AHUs, can still contribute by loosely tracking some large-feature, low-frequency components of the target signal. The low-frequency contribution reduces the required total capacity of the strongly-performing DERs in the second stage leading to more fine-tuned signal tracking in aggregation. Some recommended rules of thumb for two-stage approach are: (i) Total capacity of first-stage DERs is less than or equal to total capacity of second-stage DERs. (ii) DERs in the first stage are capable of tracking with $<$ 50\% rMSE. (iii) DER cost functions are such that the deviation from the baseline is lower cost for first-stage DERs than for second-stage. (iii) allocates a significant portion of the target signal initially to first-stage DERs, freeing up DER capacity in the second-stage for error compensation. 

\begin{table}[htb]
\caption{\textbf{Left:} Relative root mean-squared-error of tracking error by DER type. \textbf{Right:} Delay (optimal time-shift) of DER responses in sec.}\label{table:phys-error}
\centering
\begin{tabular}{|c|c|c|}
\hline
\textbf{DER Type}  & \textbf{Test~1} & \textbf{Test~2} \\ \hline
AHU & 
1.11 & 
0.12\\ \hline
V1G EVs & 
0.68 & 0.077 \\ \hline
V2G EVs & 
0.30 & 
0.060 \\ \hline
BESS & 0.054 & 0.018 \\ \hline
Total & 
0.50 & 0.097 \\ \hline
\end{tabular}
\qquad
\begin{tabular}{|c|c|c|}
\hline
\textbf{DER Type}  & \textbf{Test~1} & \textbf{Test~2} \\ \hline
AHU & 4 & 105 \\ \hline
V1G EVs & 40 & 10 \\ \hline
V2G EVs & 5 & 3 \\ \hline
BESS & 0 & 0 \\ \hline
Total & N/A & N/A \\ \hline
\end{tabular}
\end{table}

\subsubsection{Economic Benefit Analysis}\label{economic}

Here, we evaluate the economic benefit of the proposed test system, which is vital for wider scale adoption of DERs as a frequency regulation resource in real electricity markets. To this end, we take an approach similar to~\cite{YL-PB-SM-TM:15} to first demonstrate that the testbed is eligible to participate in the PJM ancillary services market. Following the PJM Manual 12~\cite{PJM:20} (Section~\ref{error-metrics}), we compute a Correlation Score $S_c$ = 
0.98, Delay  Score $S_d$ = 
0.65, and Precision Score $S_p$ = 
0.91 from data for Test~2, and obtain a Performance Score $S = 
0.85\geq 0.75$, which confirms the eligibility to participate in the PJM ancillary service market.

Next, we compute the estimated annual revenue assuming that the resources are available throughout the day. Using PJM's capability clearing price data\footnote{\url{https://dataminer2.pjm.com/feed/reg_prices/definition}} 
with our total (active) DER capacity of 184~kW and performance score of 
0.85, the revenue for this population of resources (cf.~\cite[Section 4]{PJM:19}) would be \$135 for July 9, 2020. This gives an estimated amount of \$49,210 as the total annual revenue. Note that the 184~kW DER capacity employed in this work represents less than 5\% of the total DER capacity and less than 0.5\% of the total capacity of the UCSD microgrid, cf.~\cite{BW-JD-DW-JK-NB-WT-CR:13}. As such, the revenue would significantly increase if more microgrid resources are utilized for regulation, even with reduced availability.

 \section*{Appendix: Distributed Coordination Algorithms}\label{sec:appendix}

 In this section we describe the algorithms used in our distributed computing platform to solve~\eqref{eq:opt}. 

 \emph{Ratio-Consensus (RC)}: The ratio-consensus of~\cite{ADDG-CNH-NHV:12} computes equitable contributions from all DERs without DER-specific cost functions (or constant DER costs). The ratio-consensus algorithm for providing $\Pref$ is given by
 \begin{alignat*}{2}
     y_i[k+1] &= \sum_{j\in\N_i} \frac{1}{\vert \N_i\vert}y_j[k], &
     z_i[k+1] &= \sum_{j\in\N_i} \frac{1}{\vert \N_i\vert}z_j[k], \\
     y_i[0] &= \begin{cases} \frac{\Pref}{\vert\I\vert} - \pu_i, & i\in\I, \\
     -\pu_i, & i\notin\I,
     \end{cases} & 
     z_i[0] &= \po_i - \pu_i,
 \end{alignat*}
 where, $k$ is the iteration number, $y_i$ and $z_i$ are two auxiliary variables maintained by each agent, $\N_i$ denotes the neighboring DERs of DER $i$, and $\pu_i$ and $\po_i$ are the minimum and maximum power level for DER $i$ from the problem formulation in Section~\ref{optimization-statements}. $\I$ denotes the subset of DERs which know the value of the reference signal. One can see~that 
 \begin{equation*}
 \begin{aligned}
     p_i^\star &= \pu_i + \underset{k\rightarrow\infty}{\lim} y_i[k]/z_i[k](\po_i - \pu_i) \\
     &= \pu_i + \frac{\Pref - \sum_i \pu_i}{\sum_i \po_i - \pu_i}(\po_i - \pu_i),
 \end{aligned}
 \end{equation*}
 where $p_i^\star$ is then the power assignment for DER $i$.

 \emph{Primal-Dual (PD)}: Both this dynamics and DANA (described next) take into account the cost functions of the DER types when computing the power setpoints, i.e., $f_i$ are nonconstant. These functions are modeled as quadratics, which is a common choice in generator dispatch~\cite{AW-BW-GS:12}. The dynamics is based on the discretization of the primal-dual dynamics~\cite{AC-BG-JC:17-sicon} for the augmented Lagrangian of the equivalent reformulated problem, see~\cite{AC-JC:16-allerton}, and it has a linear rate of convergence to the optimizer. The algorithm is given by
 \begin{equation*}
     \begin{aligned}
         \begin{bmatrix}
             \dot{p}_i \\ \dot{y}_i \\ \dot{\lambda}_i
         \end{bmatrix} =
         \begin{bmatrix}
         -\left(f_i'(p_i)+ \lambda_i + p_i \sum_{j\in\N_i} L_{ij}y_j - \Pref/n \right) \\
         -\left( \sum_{j\in\N_i} L_{ij} (\lambda_j + x_j - \Pref/n) + \sum_{j\in\N_i^2} L_{ij}^2 y_j  \right) \\ 
         p_i + \sum_{j\in\N_i} L_{ij} y_j - \Pref/n
         \end{bmatrix},
     \end{aligned}
 \end{equation*}
 where, $L$ is the Laplacian matrix of the communication graph (see~\cite{FB-JC-SM:09}), $y_i$ is an auxiliary variable, and $\lambda_i$ is the dual variable associated with agent $i$. The update step is followed by a projection of the primal variable $p_i$ onto the box constrained local feasible set. These dynamics converge from any set of initial conditions. Since this algorithm evolves in continuous time, we use an Euler discretization with fixed step-size to implement it in discrete time.

 \emph{Distributed Approximate Newton Algorithm (DANA)}: The Distributed Approximate Newton Algorithm (DANA) of~\cite{TA-CYC-SM:18-auto} has an improved rate of convergence compared to PD. This algorithm solves the equivalent reformulated problem
 \begin{equation}\label{eq:DANA-opt}
 \begin{aligned}
 \underset{z\in\real^n}{\text{min}} \
 &f(p^0 + Lz) = \sum_{i=1}^n f_i(p_i^0 + L_i z), \\
 \text{subject to} \ &\pu - p^0 - Lz \leq \zeros_n, \\
 &p^0 + Lz - \po \leq \zeros_n,
 \end{aligned}
 \end{equation}
 where $p^0$ is a vector of initial power levels of all the DERs with $\sum_i p_i^0 = \Pref$, and $z$ is the new variable of optimization. The continuous time dynamics are given by
 \begin{equation*}
 \begin{aligned}
     \dot{z} &= -A_q\nabla_z \Lagr (z,\lambda), \\
     \dot{\lambda} &= [\nabla_\lambda \Lagr (z,\lambda)]^+_\lambda ,
 \end{aligned}
 \end{equation*}
 where $\Lagr$ is the Lagrangian of~\eqref{eq:DANA-opt} and $A_q$ is a positive definite weighting on the gradient direction which provides distributed second-order information. For brevity, we do not provide the full details of the algorithm here, which can instead be found in~\cite{TA-CYC-SM:18-auto}. The cost functions are again taken to be quadratic with strictly positive leading coefficients.

\section*{Acknowledgements}

The material in this chapter, in full, is under revision for publication in IEEE Transactions on Smart Grid. It may appear as \textit{Frequency Regulation with Heterogeneous Energy Resources: A Realization using Distributed Control}, T.~Anderson, M.~Muralidharan, P.~Srivastava, H.V.~Haghi, J.~Cort\'{e}s, J.~Kleissl, S.~Mart\'{i}nez and B.~Washom. The dissertation author was one of three primary investigators and authors of this paper.

We would like to thank numerous people in the UCSD community and beyond for their generous contributions of time and resources to enable such an ambitious project to come together. We extend thanks to: (i) Aaron Ma and Jia (Jimmy) Qiu for assisting with hardware setup and software development for the distributed computation systems; (ii) Kevin Norris for coordinating the fleet vehicles; (iii) Abdulkarim Alamad for overseeing V1G drivers in Test~2; (iv) Kelsey Johnson for managing the Nuvve contributions; (v) Ted Lee, Patrick Kelly, and Steven Low for managing the PowerFlex contribution; (vi) Marco Arciniega, Martin Greenawalt, James Gunn, Josh Kavanagh, Jennifer Rodgers, Patricia Roman and Lashon Smith from UCSD parking
for reserving EV charging station parking spaces; (vii) Charles Bryant, Harley Crace, John Denhart, Nirav Desai, John Dilliott, Mark Gaus, Martin Greenawalt, Gerald Hernandez, Brandon Hirsch, Mark Jurgens, Josh Kavanagh, Jose Moret, Chuck Morgan, Curt Lutz, Jose Moret, Cynthia Wade, Raymond Wampler and Ed Webb for contributing their EVs in Test~1; (viii) Adrian Armenta, Adrian Gutierrez and Minghua Ong who helped with ChargePoint manual control; (ix) Bob Caldwell (Centaurus Prime), Gregory Collins, Charles Bryant, and Robert Austin for programming and enabling the AHU control; (x) Gary Matthews and John Dilliott for permitting the experimentation on ``live'' buildings and vehicles; and (xi) Antoni Tong and Cristian Cortes-Aguirre for supplying the BESS. Finally, we would like to extend a sincere thanks to the ARPA-e NODES program for its financial support and to its leadership, including Sonja Glavaski, Mario Garcia-Sanz, and Mirjana Marden, for their vision and push for the development of large-scale power-in-the-loop testing environments.

\chapter{Conclusion}
In this thesis, we studied a class of separable resource allocation problems, and we developed three types of Newton-like algorithms to approach three different scenarios of the resource allocation. Each algorithm was theoretically analyzed and rigorously shown to satisfy some convergence criteria, and the efficacy of each was validated in simulation with comparisons to relevant alternatives available in literature. We now summarize chapter-by-chapter the more specific conclusions that can be drawn and suggestions for future work.

In Chapter~\ref{chap:DANA}, motivated by economic dispatch problems and separable resource allocation problems in general, this work proposed a class of novel
$\distnewton$ algorithms. We first posed the topology design proplem and provided an
effective method for designing communication weightings.  The
weight design we propose is more cognizant of the problem geometry,
and it outperforms the current literature on network weight
design even when applied to a gradient-like method. Our contribution on the second-order weight design approach is novel but is limited in scope to the given problem formulation. Distributed second-order methods are quite immature in the present literature, so an emphasis of future work is to generalize this weight design notion to a broader class of problems. Ongoing work also includes generalizing the cost functions for box-constrained settings and discretizing the continuous-time algorithm. In addition, we aim to develop distributed Newton-like methods suited to handle more general constraints and design for robustness under uncertain parameters or lossy
communications. Another point of interest is to further study
methods for solving bilinear problems and apply these to weight design within the Newton framework.

Chapter~\ref{chap:DiSCRN} studied a nested, distributed stochastic optimization problem
and applied a Distributed Stochastic Cubic-Regularized Newton (DiSCRN)
algorithm to solve it. In order to compute the DiSCRN update, a batch
of approximate solutions to realizations of the inner-problem are
obtained, and we developed a locally-checkable stopping criterion to
certify sufficient accuracy of these solutions. The accuracy parameter
is directly leveraged in the analysis of the outer-problem, and
simulations justify both faster and more robust convergence properties
than that of comparable gradient-like and Newton-like
approaches. Future work involves developing and analyzing a saddle-point dynamics approach for solving $\P 3$ (extending the work of~\cite{YC-JD:19}), extending the analysis to accommodate small disagreements in the agent states $x_i$,
and exploring adaptive batch size techniques.

In Chapter~\ref{chap:NNN}, we posed an optimal generator dispatch problem for settings in
which the agents are generators with binary controls. We first showed that the centralized problem is amenable to
solution via a Centralized $\binpromc$ approach and
proved convergence to a local minimizer with probability one
under light assumptions. Next, we developed an
approach to make the dynamics computable in a distributed
setting in which agents exchange messages with their two-hop
neighbors in a communication graph. The methods scale and perform well compared to
standard greedy and SDP-relaxation approaches, and the latter method enjoys the qualities of a distributed algorithm, 
unlike previous approaches. Future research
directions include application of the methods to a broader class of problems which may include additional cost terms or constraints and a deeper analysis of the Deterministic Annealing variant as it pertains to the online adjustment of the learning-rate $T/\tau$.

Chapter~\ref{chap:top} introduced three related problems motivated by studying the algebraic connectivity of a graph by
adding edges to an initial topology or protecting edges under the case
of a disturbance or attack on the network. We developed a novel SDP relaxation to address the NP-hardness of the design and
demonstrated in simulation that it is superior to existing methods
which are greedy and cannot accommodate general constraints. In addition, we studied the
dynamics of the game that may be played between a network coordinator and
strategic attacker. We developed the notion of an
optimal preventive solution for the coordinator and proposed effective heuristics to find such a solution guided by characterizations of the solutions to the attacker's problem. Future work includes characterizing
the performance of our SDP relaxation and developing an algorithm which provably converges to the optimal
preventive strategy.

Finally, in Chapter~\ref{chap:active_load} we presented one of the first real-world demonstrations of secondary frequency response in a distribution grid using up to 176 heterogeneous DERs. The DERs include AHUs, V1G and V2G EVs, a BESS, and passive building loads and PV generators. The computation setup utilizes state-of-the-art distributed algorithms to find the solution of a power allocation problem. We show that the real-time distributed solutions are close to the true centralized solution in an MSE sense. Tests with real, controllable DERs at power closely track the given active-power reference signal in aggregation. These tests highlight the importance of dedicated and noise-free measurement sensors and a well-understood and reliable DER control interface for precise signal tracking. Further, our economic benefit analysis shows a potential annual revenue of \$49K for the chosen DER population. 
As is already recognized by the power systems community and  federal funding agencies such as ARPA-e and NSF, large-scale power-in-the-loop testing is needed for transitioning distributed technologies to real distribution systems. 

We hope that the work of this thesis spurs further study, testing, and ultimately widespread adoption of distributed algorithms by relevant players in industry, particularly in the renewable energy sector. Returning to the philosophical motivation of the Introduction, it is paramount to anticipate and resolve the issues of scale that are emerging as a result of computing systems transitioning from the ``single-cellar" to the ``multi-cellular" model. To this end, rigorous theory must continue to be developed in order to deeply understand distributed intelligence systems and to ensure they continue to improve quality of life and serve humanity.

\appendix

\backmatter
\bibliographystyle{plain} 
\bibliography{alias,SMD-add,SM,JC}

\end{document}